\documentclass[11pt,reqno]{article}

\RequirePackage{amsthm,amsmath,amsfonts,amssymb}
\RequirePackage[authoryear]{natbib}
\RequirePackage[colorlinks,citecolor=blue,urlcolor=blue]{hyperref}
\RequirePackage{graphicx}
\usepackage{caption}
\usepackage{subcaption}
\usepackage{amsmath, amssymb, amscd, amsthm, amsfonts}
\usepackage{graphicx}
\usepackage[perpage,symbol*]{footmisc}
\usepackage{float}
\usepackage{hyperref}
\usepackage{pgfplots}
\usepackage{mathtools}
\usepackage{color}  
\usepackage{tikz}  
\usepackage{quiver}
\usepackage{bbm}

\newcommand{\commentout}[1]{}

\usetikzlibrary{shapes, arrows, calc, arrows.meta, fit, positioning} 
\tikzset{  
    -Latex,auto,node distance =1.5 cm and 1.3 cm, thick,
    state/.style ={ellipse, draw, minimum width = 0.9 cm}, 
    point/.style = {circle, draw, inner sep=0.18cm, fill, node contents={}},  
    bidirected/.style={Latex-Latex,dashed}, 
    el/.style = {inner sep=2.5pt, align=right, sloped}  
}  

\usepackage{authblk}

\DeclareMathOperator*{\argmin}{arg\,min}

\usepackage{algorithm,algcompatible,amsmath}
\DeclareMathOperator*{\argmax}{\arg\!\max}
\algnewcommand\INPUT{\item[\textbf{Input:}]}%
\algnewcommand\OUTPUT{\item[\textbf{Output:}]}%

\usepackage{enumitem}

\usepackage[english]{babel}
\setcitestyle{authoryear,open={(},close={)}}

\usepackage{colortbl,dcolumn}

\usepackage{listings} 
\usepackage{xcolor} 

\renewcommand{\raggedright}{\leftskip=0pt \rightskip=0pt plus 0cm}
\raggedright

\def\hilite<#1>{%
	\temporal<#1>{\color{blue!25}}{\color{magenta}}%
	{\color{blue!55}}}

\newcolumntype{H}{>{\columncolor{blue!20}}c!{\vrule}}
\newcolumntype{H}{>{\columncolor{blue!20}}c}

\newtheorem{theorem}{Theorem}

\newtheorem{remark}{Remark}

\newtheorem{definition}{Definition}
\newtheorem{example}{Example}

\newtheorem{lemma}{Lemma}
\newtheorem{assumption}{Assumption}
\newtheorem{proposition}{Proposition}

\numberwithin{equation}{section} 
\numberwithin{theorem}{section}
\numberwithin{lemma}{section} 
\numberwithin{corollary}{section}
\numberwithin{definition}{section}
\numberwithin{proposition}{section} 
\numberwithin{remark}{section}
\numberwithin{example}{section}
\DeclareMathOperator\supp{supp}

\renewcommand{\oddsidemargin}{0mm}

\usepackage{setspace}
\def\R{\mathbb R}

\def\Z{\mathbb Z}

\def\Z{\mathbb Z}

\def\Sbb{\mathbb S}

\newcommand{\T}{\intercal}
\DeclareMathOperator*{\Def}{Def}
\DeclareMathOperator*{\CV}{CV}

\definecolor{mydarkgreen}{rgb}{0,0.4,0}

\makeatletter
\def\@makefnmark{}

\makeatother
\usepackage{amsfonts,amssymb}
\usepackage{mathrsfs}

\makeatletter
\newcommand{\colim@}[2]{%
  \vtop{\m@th\ialign{##\cr
    \hfil$#1\operator@font colim$\hfil\cr
    \noalign{\nointerlineskip\kern1.5\ex@}#2\cr
    \noalign{\nointerlineskip\kern-\ex@}\cr}}%
}
\newcommand{\colim}{%
  \mathop{\mathpalette\colim@{\rightarrowfill@\scriptscriptstyle}}\nmlimits@
}
\renewcommand{\varinjlim}{%
  \mathop{\mathpalette\varlim@{\rightarrowfill@\scriptscriptstyle}}\nmlimits@
}
\renewcommand{\varprojlim}{%
  \mathop{\mathpalette\varlim@{\leftarrowfill@\scriptscriptstyle}}\nmlimits@
}

\newcommand{\cdedit}[1]{{\color{cyan}  #1}}

\begin{document}
\vspace{-15in}	
	\title{\LARGE {\textbf{Statistical Inference on Grayscale Images via the Euler-Radon Transform}}}
	\author[1,*]{Kun Meng}
    \author[2]{Mattie Ji}
    \author[3]{Jinyu Wang}
    \author[2]{Kexin Ding}
    \author[4]{Henry Kirveslahti}
	\author[5]{Ani Eloyan}
    \author[5,6]{Lorin Crawford}
		
	\affil[1]{\small Division of Applied Mathematics, Brown University, RI, USA}
 \affil[2]{\small Department of Mathematics, Brown University, RI, USA}
    \affil[3]{\small Data Science Initiative, Brown University, RI, USA}
    \affil[4]{\small Laboratory for Topology and Neuroscience, EPFL, Lausanne, Switzerland}
    \affil[5]{\small Department of Biostatistics, Brown University School of Public Health, RI, USA}
	\affil[6]{\small Microsoft Research New England, Cambridge, MA, USA}
	\affil[*]{Corresponding Author: e-mail: \texttt{kun\_meng@brown.edu}.}
	
	\maketitle

\begin{abstract}

Tools from topological data analysis have been widely used to represent binary images in many scientific applications. Methods that aim to represent grayscale images (i.e., where pixel intensities instead take on continuous values) have been relatively underdeveloped. 
In this paper, we introduce the Euler-Radon transform, which generalizes the Euler characteristic transform to grayscale images by using o-minimal structures and Euler integration over definable functions. Coupling the Karhunen–Loève expansion with our proposed topological representation, we offer hypothesis-testing algorithms based on the $\chi^2$ distribution for detecting significant differences between two groups of grayscale images. We illustrate our framework via extensive numerical experiments and simulations. \footnote{
\begin{itemize}
\item \textbf{Keywords:} Euler calculus; Karhunen–Loève expansion; o-minimal structures; smooth Euler-Radon transform.

    \item \textbf{Abbreviations:} BIS, binary image segmentation; CT, computed tomography; CDT, cell decomposition theorem; DERT, dual Euler-Radon transform; ECT, Euler characteristic transform; GBM, glioblastoma multiforme; iid, independently and identically distributed; LECT, lifted Euler characteristic transform; MEC, marginal Euler curve; MRI, magnetic resonance imaging; Micro-CT, micro-computed tomography; PHT, persistent homology transform; PET, positron emission tomography; RCLL, right continuous with left limit; SECT, smooth Euler characteristic transform; TDA, topological data analysis; WECT, weighted Euler curve transform.
\end{itemize}
}
\end{abstract}


\section{Introduction}\label{Introduction}

The analysis of grayscale images is important in many fields. In medical imaging, data can be derived from different modalities including magnetic resonance imaging (MRI), computed tomography (CT), positron emission tomography (PET), and micro-computed tomography (Micro-CT). Analyzing the variation among pixel intensities in these images can help diagnose diseases, detect abnormalities in human tissues, and monitor the effectiveness of different treatment strategies (e.g., see Figure \ref{fig: Grayscale_image_from_Prof_Duan}). Beyond medicine, grayscale imaging is essential in astronomy for capturing details in celestial bodies and other phenomena \citep{howell2006handbook}, in geology for studying rock and mineral compositions using electron microscopy \citep{reed2005electron}, and in meteorology for interpreting satellite imagery to predict weather patterns \citep{kidder1995satellite}. 

There is an important distinction between binary and grayscale images. Unlike binary images, where pixels are exactly one of two colors (usually black and white), pixels in grayscale images take on continuous values. A binary image can be modeled as a binary-valued function which is often expressed in the following form
\begin{align}\label{eq: generic form of black-white images}
\mathbbm{1}_K(x)=\left\{\begin{aligned}
1, \ \ & \mbox{ if the color of point (pixel) $x$ is white},\\
0, \ \ & \mbox{ if the color of point (pixel) $x$ is black},
\end{aligned}
\right.
\end{align}
where the subscript $K$ denotes the region of white points in the image (which we will refer to as a ``shape" throughout this paper). In contrast, a grayscale image must be modeled as a real-valued function. Specifically, the grayscale intensity of each pixel is represented as the function value at that corresponding point.

Many methods in the field of topological data analysis (TDA) \citep{carlsson2009topology, vishwanath2020limits} have been developed for analyzing binary images (equivalently, shapes). Some of these works include the persistent homology transform (PHT) \citep{turner2014persistent}, the Euler characteristic transform (ECT) \citep{turner2014persistent, ghrist2018persistent}, and the smooth Euler characteristic transform (SECT) \citep{crawford2020predicting} --- all of which are proposed statistics used to represent shapes while preserving the complete information they contain. \cite{crawford2020predicting} applied the SECT on MRI-derived binary images taken from tumors of glioblastoma multiforme (GBM) patients. Here, the authors used the resulting summary statistics from the SECT in a class of Gaussian process regression models to predict survival-based outcomes. \cite{wang2021statistical} utilized the ECT for sub-image analysis which aims to identify geometric features that are important for distinguishing various classes of shapes. \cite{marsh2022detecting} recently presented the DETECT: an extension of the ECT to analyze temporal changes in shapes. The authors demonstrated their approach by studying the growth of mouse small intestine organoid experiments from segmented videos. Lastly, \cite{meng2022randomness} recently used the SECT framework to introduce a $\chi^2$ distribution-based approach to test hypotheses on random shapes, with the corresponding mathematical foundation being established therein through algebraic topology, functional analysis with  Sobolev embeddings \citep{brezis2011functional}, and probability theory using the Karhunen–Loève expansion \citep{alexanderian2015brief}.

Previous TDA methods that use Euler characteristic-based invariants are well suited for binary images and shapes with clearly defined boundaries (e.g., the region $K$ of white points defined via the binary-valued function in Eq.~\eqref{eq: generic form of black-white images}). Grayscale images, on the other hand, are arrays of 2-dimensional pixel or 3-dimensional voxel intensities that represent varying levels of brightness in a continuous space (e.g., see Figures \ref{fig: Grayscale_image_from_Prof_Duan} and \ref{fig:grayscale_image}). These images lack the clear boundaries that can be used to define and compute the Euler characteristic. Consequently, due to the inherent continuous nature of grayscale images, both the Euler characteristic and the corresponding TDA methods that leverage it are not immediately applicable. 

One natural way to apply the Euler characteristic to grayscale images is through binary image segmentation (BIS) --- a process that divides points in a grayscale image into foreground (the shape of interest) and background. That is, BIS can serve as a preprocessing step to convert a grayscale image into a binary format, facilitating the application of the Euler characteristic. Numerous state-of-the-art applications of Euler characteristic-based TDA methods to grayscale images rely on BIS. For example, \cite{crawford2020predicting} used the computer-assisted program MITKats \citep{chen2017fast} to threshold MRI scans of GBM tumors and convert them to binary formats for their analyses with the SECT (refer to Figure 3 in \cite{crawford2020predicting}). Unfortunately, there are several challenges associated with BIS that can impede on the effectiveness and accuracy of downstream analyses with TDA methods. One of the primary challenges is selecting an appropriate threshold to distinguish between foreground and background points. Improper choice of a BIS threshold may result in the issue of over- or under-segmentation. Moreover, pinpointing a proper threshold is often not straightforward (especially when the images have low image contrast, large noise, or complex heterogeneity) and can be computationally demanding. This selection process can be especially challenging for medical images of soft tissues (including organs, tumors, and blood vessels). Most importantly, performing BIS inevitably causes a loss of information in images of interest. By generalizing Euler characteristic-based statistics and enabling them to be directly applicable to grayscale images without the need for BIS, we will increase their utility for better powered shape analyses.

\begin{figure}[h]
    \centering
    \includegraphics[width=172mm]{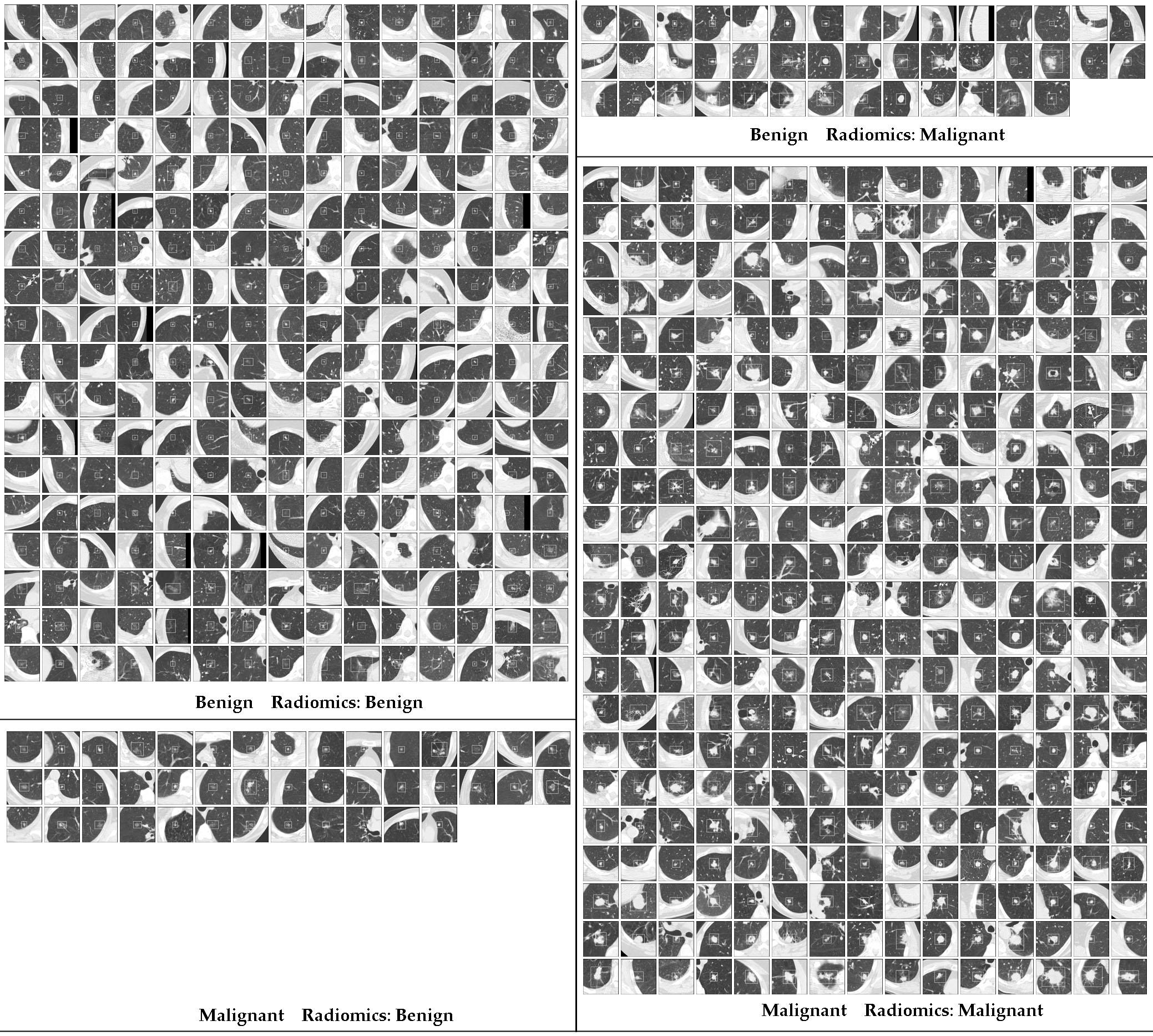}
    \caption{CT scans of lung cancer tumors labeled as ``benign" or ``malignant." This figure has been previously published in \cite{maldonado2021validation}. The details of this figure are provided therein.}
    \label{fig: Grayscale_image_from_Prof_Duan}
\end{figure}

Radiomics \citep{aerts2014decoding} has been used to estimate features from grayscale images to predict outcomes of interest, such as patient survival time in cancer. Most commonly used radiomics approaches are based on estimating parameters that describe various image intensity features, which are then employed for prediction. For example, \cite{just2014improving} described the estimation of moments of intensity histograms, \cite{brooks2013quantification} proposed a pixel-based distance-dependent approach using the deviation from a linear intensity gradation, and \cite{eloyan2020tumor} described an approach for estimation of intensity histogram and shape features taking into account the correlation structure of pixel intensities. Additionally, \cite{aerts2014decoding} showed the computation of shape and texture features from the image intensities. While radiomic features computed from grayscale images can be used to compare populations of images, they also inevitably lose information about the grayscale images of interest. This loss reduces the power of downstream statistical analyses.

The major contributions that we present in this paper all center around developing a generalization of the ECT for grayscale images. A summary of these include:
\begin{enumerate}
    \item Utilizing o-minimal structures \citep{van1998tame} and the framework proposed in \cite{baryshnikov2010euler} for Euler integration over definable functions, we introduce the Euler-Radon transform (ERT). This ERT serves as a topological summary statistic that aims to unify the ECT \citep{turner2014persistent, ghrist2018persistent}, the weighted Euler curve transform (WECT) \citep{jiang2020weighted}, and the marginal Euler curve (MEC) \citep{kirveslahti2023representing}. Notably, when the ERT is employed on binary images, it coincides with the ECT. Moreover, akin to the framework presented in \cite{kirveslahti2023representing}, our ERT does not rely on the diffeomorphism assumption posited in many state-of-the-art methods \citep[e.g.,][]{ashburner2007fast}. Lastly, unlike previous TDA-based methods that rely on BIS and standard radiomic approaches, the ``Schapira inversion formula" \citep{schapira1995tomography} guarantees that our proposed ERT summary preserves all information within the pixel intensity arrays of grayscale images.

    \item Using the proposed ERT as a building block, we also introduce the smooth Euler-Radon transform (SERT). When applied to binary images, the SERT coincides with the SECT \citep{crawford2020predicting, meng2022randomness}. Importantly, the SERT represents the grayscale images as functional data. Numerous tools from functional data analysis (FDA) \citep{hsing2015theoretical} and functional analysis \citep{brezis2011functional} are applicable with the SERT \citep[e.g., the Karhunen–Loève expansion][]{alexanderian2015brief}.

    \item Using the ERT and SERT, we propose several statistical algorithms aimed at detecting significant differences between paired collections of grayscale images (e.g., analyzing CT scans of malignant and benign tumors from Figure \ref{fig: Grayscale_image_from_Prof_Duan}). Particularly, our proposed algorithms combine the Karhunen–Loève expansion with a permutation-based approach. Using simulations, we show that our hypothesis test is uniformly powerful across various scenarios and does not suffer from type I error inflation. These algorithms are a generalization of results presented in \cite{meng2022randomness}.
\end{enumerate}
Beyond the contributions outlined above, due to the resemblance between the ECT and our proposed ERT, this paper paves the way for generalizing a series of ECT-based methods \citep{crawford2020predicting, wang2021statistical, marsh2022detecting} from binary images to grayscale images. 

The remainder of this paper is organized as follows. In Section \ref{section: Representations of Grayscale Images via the Euler Calculus}, we review some existing representations of grayscale images. In Section \ref{section: The Euler-Radon Transform of Grayscale Functions}, we first review the basics of Euler calculus that have been described in \cite{van1998tame}, \cite{baryshnikov2010euler}, and \cite{ghrist2014elementary}. Then, using Euler calculus, we define the ERT and detail its properties. In Section \ref{section: Existing Frameworks}, we comprehensively describe the relationship between our proposed ERT and existing topological representations of grayscale images \citep{jiang2020weighted, kirveslahti2023representing}. Section \ref{subsection: proof-of-concept simulation} offers a proof-of-concept example that illustrates the behavior of our proposed SERT. In Section \ref{section: Alignment of Images and the Invariance of ERT}, we propose an ERT-based alignment approach for preprocessing grayscale images prior to statistical inference without relying on correspondences between them. In Section \ref{section: Statistical Inference of Grayscale Functions}, we propose several statistical algorithms designed to differentiate between two sets of grayscale images. The performance of the proposed algorithms is presented in Section \ref{section: Numerical Experiments} using simulations. Lastly, we conclude this paper in Section \ref{section: Conclusions and Future Research} and discuss several future research directions. The proofs of all theorems in this paper are provided in Appendix \ref{Appendix: Proofs} unless otherwise stated.

\section{Representations of Grayscale Images}\label{section: Representations of Grayscale Images via the Euler Calculus}

The statistical inference on grayscale images necessitates an appropriate representation of grayscale images. For the application of statistical methods akin to the ECT-based methods in \cite{crawford2020predicting}, \cite{wang2021statistical}, \cite{meng2022randomness}, and \cite{marsh2022detecting}, it is imperative that our proposed representation of grayscale images aligns as closely as possible with the ECT. Prior to delving into our proposed approach, this section offers a review of some existing representations of grayscale images. 


In an attempt to employ an ECT-like method to grayscale images of GBM tumors, \cite{jiang2020weighted} transformed each grayscale image into the discrete representation $\sum a_\sigma\cdot\mathbbm{1}_\sigma(x)$ in their preprocessing step, where $\sum$ denotes a finite sum, each $\sigma$ is a simplicial complex, and each weight $a_\sigma$ belongs to $\mathbb{N}:=\{0,1,2,\ldots\}$. In other words, each grayscale image is transformed into a weighted sum of a finite set of simplexes (referred to as a ``weighted complex" therein). Subsequently, \cite{jiang2020weighted} introduced the weighted Euler curve transform (WECT) tailored for weighted complexes. A limitation of the WECT method is the dependency of data analysis results on the discretization $\sum a_\sigma\cdot\mathbbm{1}_\sigma(x)$, unless the original image is discrete and can be exactly depicted as a weighted complex. When dealing with a high-resolution grayscale image, the deviation between the original image and its weighted complex discretization can be substantial. Importantly, the reliance of data analysis on the discretization complicates theoretical analysis, which necessitates discretization-free representations of grayscale images.

Many state-of-the-art methods for analyzing grayscale images rely on the diffeomorphism assumption that the grayscale images being analyzed are diffeomorphic \citep[e.g.,][]{ashburner2007fast}. To obviate the diffeomorphism assumption, \cite{kirveslahti2023representing} introduced two novel representations called the lifted Euler characteristic transform (LECT) and the super LECT (SELECT) utilizing a lifting technique. In contrast to the approach proposed by \cite{jiang2020weighted}, neither the LECT nor SELECT depends on the discretization preprocessing for grayscale images. However, the lifting procedure used in both the LECT and SELECT introduces an additional dimension. For example, the LECT represents a $d$-dimensional grayscale image as a function on a $(d+1)$-dimensional manifold (see Eq.~\eqref{eq: def of LECT} for details). This increase in dimensionality distinguishes the LECT and SELECT distinctly from the ECT, precluding straightforward theoretical and methodological generalizations of ECT-based methods to grayscale images \citep{ghrist2018persistent, crawford2020predicting, wang2021statistical, meng2022randomness, marsh2022detecting}. Additionally, the augmented dimensionality elevates the computational expense associated with subsequent statistical inference.

A grayscale image can also be represented using function values where $g(x)$ is used to denote the grayscale intensity at point (pixel) $x$. One straightforward approach to quantify differences between two images can then be done using the Lebesgue integral $\int \vert g^{(1)}(x) - g^{(2)}(x)\vert^2 \,dx$ for grayscale images represented by functions $g^{(1)}$ and $g^{(2)}$. Representing grayscale images through topological invariants presents distinct advantages compared to using function values. These advantages have been documented by \cite{kirveslahti2023representing}. As a toy example, the Euler characteristic of a singular point is 1; whereas, its Lebesgue integral is 0.

\section{Euler-Radon Transform of Grayscale Functions}\label{section: The Euler-Radon Transform of Grayscale Functions}

In this section, we generalize the ECT from binary images to grayscale images via the Euler calculus \citep{baryshnikov2010euler, ghrist2014elementary}. This generalization does not depend on the discretization of grayscale images nor does it introduce an additional dimension. More precisely, we introduce the Euler-Radon transform (ERT) as a means to perform statistical inference on grayscale images. This approach serves as a natural extension of both the ECT and WECT, and has a direct connection to the LECT and SELECT. Furthermore, to apply functional data analysis, we propose a smooth version of the ERT --- the smooth Euler-Radon transform (SERT), which is a generalization of the SECT. 

\subsection{Outline}

To elucidate the conceptual foundation of our ERT as an extension of the ECT, we revisit the ECT of shapes $K$ from the viewpoint of Euler calculus. The discussion in this subsection primarily serves a heuristic purpose; a thorough and precise exposition is given in Section \ref{section: Euler-Randon Transform}. 

Without loss of generality, we assume $K\subseteq B_{\mathbb{R}^d}(0,R):=\{x\in\mathbb{R}^d:\Vert x\Vert < R\}$ for a prespecified radius $R>0$ and $K$ is compact. The ECT of $K$ is a collection of Euler characteristics $\{\chi(K_t^\nu):\, (\nu,t)\in\mathbb{S}^{d-1}\times[0,2R]\}$, where $\chi\left(K_t^\nu\right) \text{ is the Euler characteristic of }K_t^\nu$ and $K_t^\nu:=\{x\in K:\, x\cdot \nu\le t-R\}$ \citep[also][]{meng2022randomness}. The Euler characteristic $\chi\left(K_t^\nu\right)$ can be represented as follows using the Euler functional $\int(\cdot) d\chi$ \citep[e.g.,][]{ghrist2018persistent}
\begin{align}\label{eq: Euler function representation of ECT}
\begin{aligned}
& \chi\left(K_t^\nu\right) = \int \mathbbm{1}_K(x) \cdot R(x,\nu,t)\, d\chi(x),\\
&\text{where }\ \ R(x,\nu,t) := \mathbbm{1}\left\{\left(x,\nu,t\right)\in B_{\mathbb{R}^d}(0,R) \times\mathbb{S}^{d-1}\times[0,T]\,:\, x\cdot\nu\le t-R\right\} \text{ and }T:=2R.
\end{aligned}
\end{align}
In essence, this formulation can be viewed as a generalized Radon transform of the function $\mathbbm{1}_K$ \citep{schapira1995tomography, baryshnikov2011inversion, ghrist2018persistent}. Here, the indicator function $\mathbbm{1}_K$ represents a binary image, as referenced in Eq.~\eqref{eq: generic form of black-white images}. In the development of the ERT, our primary objective is to substitute the indicator function $\mathbbm{1}_K$ in Eq.~\eqref{eq: Euler function representation of ECT} with a real-valued function $g$ representing a grayscale image.

\begin{figure}[!tbp]
\centering
\subfloat[Grayscale image $g(x)$]{
\centering
\includegraphics[width=84mm]
{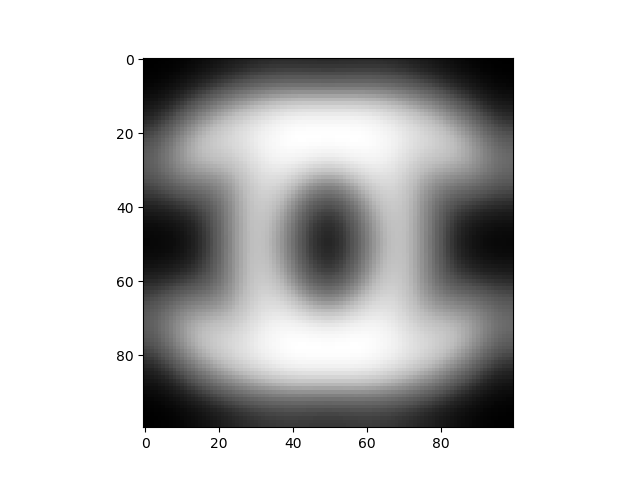}
}
\hfill
\subfloat[Grayscale image $\lambda\cdot g(x)$ with $\lambda=-0.5$]{
\centering

\includegraphics[width=84mm]{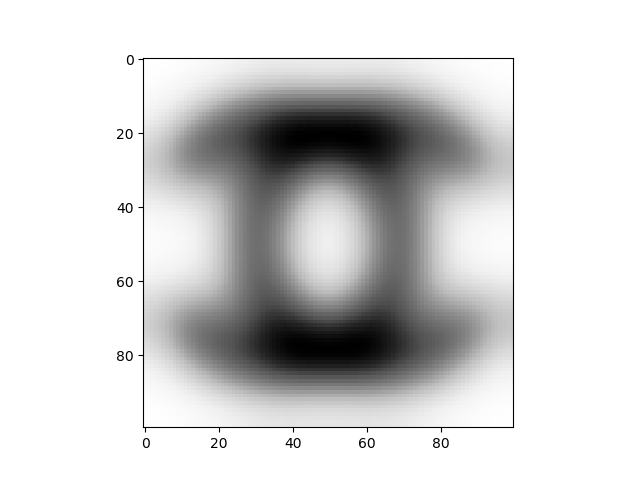}
}
\caption{Panel (a) shows a grayscale image represented by a function $g$ in the form of Eq.~\eqref{eq: general grayscale function}. Here, each white pixel corresponds to value 1 and each black pixel corresponds to 0. Panel (b) presents a rescaled version of $g$, where the negative coefficient $\lambda$ induces a ``white-to-black" transition. In this image, each white pixel corresponds to value 0 and each black pixel corresponds to value -0.5.}
\label{fig:grayscale_image}
\end{figure}

A conventional approach to modeling a grayscale image $g$ takes the form
\begin{align}\label{eq: general grayscale function}
g(x)=\left\{
\begin{aligned}
& 1, \ \ \ \mbox{ if the color at point (pixel) $x$ is white},\\
& \#, \ \ \ \mbox{the value of the grayscale intensity at point (pixel) $x$ scaled between 0 and 1,}\\
& 0, \ \ \ \mbox{ if the color at point (pixel) $x$ is black}.
\end{aligned}
\right.
\end{align}
For reference, we encourage readers to compare Eq.~\eqref{eq: general grayscale function} with Eq.~\eqref{eq: generic form of black-white images}. An example of a grayscale image is presented in Figure~\ref{fig:grayscale_image}(a). In many applications, it becomes important to rescale a grayscale image where $g(x)\mapsto \lambda\cdot g(x)$ for $\lambda\in\mathbb{R}$. However, the configuration in Eq.~\eqref{eq: general grayscale function} is not invariant under this rescaling transform. Therefore, we will no longer view a grayscale image as a function taking values between $[0,1]$. Instead, throughout this paper, we will treat each grayscale image generally as a bounded function $g$, which we subsequently term a grayscale function (see Section \ref{section: Euler-Randon Transform} for details). Obviously, a rescaled bounded function is still bounded, although it may take values outside $[0,1]$. Furthermore, for negative $\lambda$, this rescaling inverts the grayscale spectrum, reminiscent of a ``white-to-black" transition (an example of this is available in Figure~\ref{fig:grayscale_image}(b)). By encompassing general bounded functions — beyond those merely within the $[0,1]$ range — our approach broadens its applicability, encapsulating scalar fields that are beyond the structure posited in Eq.~\eqref{eq: general grayscale function} (e.g., the realizations of Gaussian random fields) \citep{adler2007random, bobrowski_borman_2012}.

As previously mentioned, the key to generalizing the ECT to the ERT is replacing the indicator function $\mathbbm{1}_K$ in Eq.~\eqref{eq: Euler function representation of ECT} with a real-valued function $g$ representing a grayscale image. That is, the ERT of a $d$-dimensional grayscale image $g$ can be heuristically expressed as follows
\begin{align}\label{eq: formal generalization of ECT to real-valued input}
\int g(x) \cdot R(x,\nu,t) \, d\chi(x),
\end{align}
which is a function of $(\nu,t)$ on the $d$-dimensional manifold $\mathbb{S}^{d-1}\times[0,T]$. If the grayscale function $g$ in Eq.~\eqref{eq: formal generalization of ECT to real-valued input} takes only finitely many values in $\mathbb{Z}$ (e.g., $g$ is a discretized version of some underlying high-resolution grayscale image), the integration in Eq.~\eqref{eq: formal generalization of ECT to real-valued input} can be easily defined via the Euler integration over constructible functions (see Section 3.6 of \cite{ghrist2014elementary}) and is equal to the WECT under some tameness conditions \citep{jiang2020weighted}. When the grayscale function $g$ has an ``infinitely fine resolution" (i.e., $g$ is generally real-valued), the rigorous definition of the integration in Eq.~\eqref{eq: formal generalization of ECT to real-valued input} necessitates the techniques developed in \cite{baryshnikov2010euler}. A more elaborate discussion on this, using o-minimal structures, will be discussed in Section \ref{section: Euler Calculus}.

\subsection{Euler Calculus via O-minimal Structures}\label{section: Euler Calculus}

The primary objective of this subsection is to revisit Euler calculus --- to prepare for a rigorous version of the heuristic integration presented in Eq.~\eqref{eq: formal generalization of ECT to real-valued input}. This will result in the precise definition of the ERT which we detail in Section \ref{section: Euler-Randon Transform}. 

\paragraph*{O-minimal Structures and Definable Functions.} Euler calculus has been extensively examined in the literature \citep{baryshnikov2010euler, ghrist2014elementary}. Lebesgue integration is established for measurable functions. In a manner similar to Lebesgue integration, Euler integration begins by determining a set of functions that act as integrands for Euler integrals. These specific functions are termed ``definable functions" and are specified through o-minimal structures. The definition of o-minimal structures is available in \cite{van1998tame} and rephrased as follows:
\begin{definition}\label{def: definability}
An \textbf{o-minimal structure on $\mathbb{R}$} is a sequence $\mathcal{O}=\{\mathcal{O}_n\}_{n\ge1}$ satisfying the following axioms:\\ 
(i) for each $n$, the collection $\mathcal{O}_n$ is a Boolean algebra of subsets of $\mathbb{R}^n$; \\
(ii) $A\in\mathcal{O}_n$ implies $A\times \mathbb{R}\in\mathcal{O}_{n+1}$ and $\mathbb{R}\times A\in\mathcal{O}_{n+1}$; \\
(iii) $\{(x_1,\ldots,x_n)\in\mathbb{R}^n \,:\, x_i=x_j\}\in\mathcal{O}_n$ for all $1\le i < j\le n$; \\
(iv) $A\in\mathcal{O}_{n+1}$ implies $\pi(A)\in\mathcal{O}_{n}$, where $\pi:\mathbb{R}^{n+1}\rightarrow\mathbb{R}^n$ is the projection map on the first $n$ coordinates; \\
(v) $\{r\}\in\mathcal{O}_1$ for all $r\in\mathbb{R}$, and $\{(x,y)\in\mathbb{R}^2 \,:\, x<y\}\in\mathcal{O}_2$; \\
(vi) the only sets in $\mathcal{O}_1$ are the finite unions of open intervals (with $\pm\infty$ endpoints allowed) and points. \\
A set $K$ is said to be \textbf{definable} with respect to $\mathcal{O}$ if $K\in\mathcal{O}$ (i.e., there exists an $n$ such that $K\in\mathcal{O}_n$).
\end{definition}
\noindent A typical example of o-minimal structures is the collection of \textbf{semialgebraic sets} which is defined as: a set $K\subseteq\mathbb{R}^n$ is said to be a semialgebraic subset of $\mathbb{R}^n$ if it is a finite union of sets of the following form
\begin{align*}
    \Big\{x\in\mathbb{R}^n \,:\, p_1(x)=0,\,\ldots,\, p_k(x)=0,\, q_1(x)>0,\,\ldots,\, q_l(x)>0\Big\},
\end{align*}
where $k$ and $l$ are positive integers, and $p_1,\ldots,p_k, q_1,\ldots,q_l$ are real polynomial functions on $\mathbb{R}^n$ (also see Chapter 2 of \cite{van1998tame}). Specifically, if we let $\mathcal{O}_n=$ the collection of semialgebraic subsets of $\mathbb{R}^n$, then $\mathcal{O}=\{\mathcal{O}_n\}_{n\ge1}$ is an o-minimal structure. Since the unit sphere $\mathbb{S}^{n-1}=\{x\in\mathbb{R}^n :\, \Vert x\Vert^2 - 1 = 0\}$ is defined using the polynomial $\Vert x\Vert^2 - 1$, it is definable with respect to this o-minimal structure. It is also true for the open ball $B_{\mathbb{R}^n}(0,R)=\left\{x\in\mathbb{R}^n:\, \Vert x\Vert^2-R<0\right\}$ centered at the origin with radius $R>0$. Throughout this paper, to include many common sets in our framework, we assume the o-minimal structures $\mathcal{O}$ of interest to satisfy the following assumption:
\begin{assumption}\label{Assumption: basic requirements for o-minimal structures of interest}
    The o-minimal structure $\mathcal{O}$ of interest contains all semialgebraic sets.
\end{assumption}
\noindent Importantly, under Assumption \ref{Assumption: basic requirements for o-minimal structures of interest}, we are able to apply the ``triangulation theorem" and the ``trivialization theorem" as presented in \cite{van1998tame}. In particular, the ``triangulation theorem" (see Chapter 8 of \cite{van1998tame}) indicates that each definable set is homeomorphic to a polyhedron, which subsequently suggests that each definable set is Borel-measurable (see ``Definition 1.21" of \cite{klenke2013probability}).

In addition to o-minimal structures, we need the concepts in the following definition, which are also available in \cite{baryshnikov2010euler}.
\begin{definition}\label{def: definability}
Suppose we have an o-minimal structure $\mathcal{O}$ on $\mathbb{R}$. Let $X$ be definable and $Y\subseteq\mathbb{R}^N$ for some positive integer $N$.
\begin{enumerate}
\item A function $g: X\rightarrow Y$ is said to be \textbf{definable} if its graph $\Gamma(g):=\{(x,y)\in X\times Y : y=g(x)\}$ is definable (i.e., $\Gamma(g)\in\mathcal{O}$).
\item Let $\operatorname{Def}(X;Y)$ denote the collection of compactly supported definable functions $g: X\rightarrow Y$. Denote $\operatorname{Def}(X):=\operatorname{Def}(X;\mathbb{R})$.
\item Denote $\operatorname{CF}(X):=\operatorname{Def}(X;\,\mathbb{Z})$. Any function in $\operatorname{CF}(X)$ is called a \textbf{constructible function}.
\item If a definable set is also compact, we call this set a \textbf{constructible set}; the collection of all constructible subsets of $X$ is denoted by $\operatorname{CS}(X)$. Obviously, for any constructible subset $K \subseteq X$, $K\mapsto \mathbbm{1}_K$ is an injective map from $\operatorname{CS}(X)$ to $\operatorname{CF}(X)$.
\end{enumerate}
\end{definition} 
\noindent It is simple to verify that, under Assumption \ref{Assumption: basic requirements for o-minimal structures of interest}, the function $R(x,\nu,t)$ defined in Eq.~\eqref{eq: Euler function representation of ECT} is a constructible function where $R\in\operatorname{CF}(\mathbb{R}^{2d+1})$, and $\{(x,\nu,t)\in\mathbb{R}^{2d+1}\,:\,R(x,\nu,t)=1\}$ is a constructible set.

The term ``tameness" is frequently used in the TDA literature. While ``tame" is often used synonymously with ``definable," certain works \citep[e.g.,][]{bobrowski_borman_2012} attribute ``tameness" to a notion that slightly diverges from the definability outlined in Definition \ref{def: definability}. To ensure clarity, we will use ``definable" in lieu of ``tame." An in-depth exploration of the interplay between definability and tameness can be found in Appendix \ref{section: Definability vs. Tameness}.

\paragraph*{Euler Characteristic.} In Euler calculus, the Euler characteristic $\chi(\cdot)$ plays a role analogous to the Lebesgue measure in the Lebesgue integral theory. For any definable set $K\in\mathcal{O}$, the cell decomposition theorem (CDT) indicates that there exists a partition $\mathcal{P}$ of $K$, where $\mathcal{P}$ is a finite collection of cells (see Chapter 3 of \cite{van1998tame} for the definition of cells and CDT). Then, the Euler characteristic $\chi(K)$ of the definable set $K$ is defined as follows
\begin{align}\label{eq: EC defined using o-min structures}
    \chi(K) := \sum_{C\in\mathcal{P}} (-1)^{\operatorname{dim}(C)},
\end{align}
where $\operatorname{dim}(C)$ denotes the dimension of the cell $C$ (see Chapter 4 of \cite{van1998tame} for the definition of dimensions). One can also show that the value $\chi(K)$ does not depend on the choice of partition $\mathcal{P}$ (see Chapter 4 of \cite{van1998tame}, Section 2 therein). As discussed at the beginning of \cite{baryshnikov2010euler}, the Euler characteristic in Eq.~\eqref{eq: EC defined using o-min structures} is equivalently defined via the Borel-Moore homology, where $\chi(K)=\sum_{n\in\mathbb{Z}}(-1)^n \cdot \operatorname{dim} H_n^{BM}(K;\mathbb{R})$ with $H_*^{BM}$ denoting the Borel-Moore homology \citep{bredon2012sheaf}. Note that $\chi(K)$ is a homotopy invariant if $K$ is compact but is only a homeomorphism invariant in general.

\paragraph*{Euler Integration over Constructible Functions.} For any constructible function $g\in\operatorname{CF}(X)$, its Euler integral is defined as follows (also see Section 3.6 of \cite{ghrist2014elementary})
\begin{align}\label{eq: Euler integration}
    \int_X g(x) \, d\chi(x):=\sum_{n=-\infty}^{+\infty} n\cdot\chi\left(\{x\in X: \, g(x)=n\}\right).
\end{align}
Particularly, $\int_{X} \mathbbm{1}_K(x) \,d\chi(x)=\chi(K)$ for all $K\in\operatorname{CS}(X)$. Using the Euler integration $\int (\cdot)\,d\chi$ defined in Eq.~\eqref{eq: Euler integration}, we may represent the ECT by the following$^\dagger$\footnote{$\dagger$: We use $\{f(x)\}_{x\in X}$ to denote the function $f:X\rightarrow\mathbb{R},\, x\mapsto f(x)$ throughout this paper.}
\begin{align}\label{eq: ECT via Euler integration}
    \begin{aligned}
        \operatorname{ECT}:\ \ &\operatorname{CS}\left(B_{\mathbb{R}^d}(0,R)\right) \rightarrow \mathbb{Z}^{\mathbb{S}^{d-1}\times[0,T]},\\
        & K \mapsto \operatorname{ECT}(K)=\left\{\chi(K_t^\nu)=\int_{B_{\mathbb{R}^d}(0,R)} \mathbbm{1}_K(x)\cdot R(x,\nu,t) \, d\chi(x)\right\}_{(\nu,t)\in\mathbb{S}^{d-1}\times[0,T]},
    \end{aligned}
\end{align}
where $T=2R$ and $R(x,\nu,t)$ is the indicator function defined in Eq.~\eqref{eq: Euler function representation of ECT}. Eq.~\eqref{eq: ECT via Euler integration} is a rigorous version of Eq.~\eqref{eq: Euler function representation of ECT} in the sense that Eq.~\eqref{eq: ECT via Euler integration} specifies the collection of shapes for which the ECT is well-defined. Furthermore, $\chi(K_t^\nu)$ varies ``definably" with respect to $(\nu,t)$, which is precisely presented by the following theorem.
\begin{theorem}\label{thm: tameness theorem}
    Suppose $K\in \operatorname{CS}(B_{\mathbb{R}^d}(0,R))$. Then, we have the following:
    \begin{enumerate}
        \item $\chi(K_t^\nu)$ takes only finitely many values as $(\nu,t)$ runs through $\mathbb{S}^{d-1}\times[0,T]$. In addition, for each integer $z \in \mathbb{Z}$, the set $\{(\nu,t)\in \mathbb{S}^{d-1}\times[0,T]:\, \chi(K_t^\nu)=z\}$ is definable; hence, the function $(\nu,t) \mapsto \chi(K_t^\nu)$ is a definable function.
        \item The function $(\nu,t) \mapsto \chi(K_t^\nu)$ is Borel-measurable.
        \item For each fixed direction $\nu\in\mathbb{S}^{d-1}$, the function $t\mapsto \chi(K_t^\nu)$ has at most finitely many discontinuities. More precisely, there are points $a_1<\ldots<a_k$ in $(0,T)$ such that on each interval $(a_j, a_{j+1})$ with $a_{k+1}=T$, the function is constant.
    \end{enumerate}
\end{theorem}
\noindent The proof of Theorem \ref{thm: tameness theorem} is given in Appendix \ref{proof of thm: tameness theorem}. The third result of Theorem \ref{thm: tameness theorem} indicates that the ``tameness assumption" in (an old version of) \cite{meng2022randomness} is redundant if the shapes of interest are definable. Recall that the SECT of $K\in \operatorname{CS}(B_{\mathbb{R}^d}(0,R))$ is defined by the following \citep{crawford2020predicting, meng2022randomness}
\begin{align}\label{eq: def of SECT}
    \operatorname{SECT}(K)(\nu,t):=\int_0^t \chi(K_\tau^\nu)\,d\tau - \frac{t}{T} \int_0^T \chi(K_\tau^\nu)\,d\tau,\ \ \ \text{for all }(\nu,t)\in\mathbb{S}^{d-1}\times[0,T].
\end{align}
Theorem \ref{thm: tameness theorem} also guarantees that the Lebesgue integrals in Eq.~\eqref{eq: def of SECT} are well-defined and that the map $(\nu,t) \mapsto \operatorname{SECT}(K)(\nu,t)$ is Borel-measurable.

\paragraph*{Euler Integration over Definable Functions.} The Euler integration $\int_X (\cdot)d\chi$ defined in Eq.~\eqref{eq: Euler integration} is exclusively tailored for integer-valued functions within $\operatorname{CF}(X)=\operatorname{Def}(X;\mathbb{Z})$ (e.g., the indicator function $\mathbbm{1}_K$ that represents a binary image). Consequently, it cannot accommodate real-valued grayscale functions $g$ possessing infinitely fine resolutions (e.g., see Figures \ref{fig: Grayscale_image_from_Prof_Duan} and \ref{fig:grayscale_image}). Therefore, to provide a rigorous definition of the integrals in Eq.~\eqref{eq: formal generalization of ECT to real-valued input}, we need a framework that extends beyond the scope of Eq.~\eqref{eq: Euler integration}.

When the integrands are real-valued functions, one needs step-function approximations. We first review the definitions of floor and ceiling functions. For any real number $s$, $\lfloor s\rfloor :=$ the greatest integer less than or equal to $s$, and $\lceil s\rceil :=$ the least integer greater than or equal to $s$. Based on the functional $\int_X (\cdot)d\chi$ in Eq.~\eqref{eq: Euler integration}, \cite{baryshnikov2010euler} proposed the Euler integration functionals $\int_X (\cdot)\,\lfloor d\chi\rfloor$, $\int_X (\cdot)\,\lceil d\chi\rceil$, and $\int_X (\cdot)\, [d\chi]$ for real-valued definable function $g\in\operatorname{Def}(X;\mathbb{R})$ as follows
\begin{align}\label{eq: def of int [dx]}
    \begin{aligned}
       \text{(floor version)}\ \ \ \ \  & \int_X g(x) \,\lfloor d\chi(x)\rfloor := \lim_{n\rightarrow\infty} \left\{ \frac{1}{n}\int_X \lfloor n\cdot g(x)\rfloor \, d\chi(x)\right\}, \\
    \text{(ceiling version)}\ \ \ \ \  & \int_X g(x) \,\lceil d\chi(x)\rceil := \lim_{n\rightarrow\infty} \left\{ \frac{1}{n}\int_X \lceil n\cdot g(x)\rceil \, d\chi(x)\right\}, \\
    \text{(averaged version)}\ \ \ \ \  & \int_X g(x) \, [d\chi(x)] := \frac{1}{2} \left( \int_X g(x) \,\lfloor d\chi(x)\rfloor + \int_X g(x) \,\lceil d\chi(x)\rceil \right).
    \end{aligned}
\end{align}
\cite{baryshnikov2010euler} (``Lemma 3" therein) showed that the limits in Eq.~\eqref{eq: def of int [dx]} exist; hence, the functionals in Eq.~\eqref{eq: def of int [dx]} are well-defined. The following equation indicates that the functionals defined in Eq.~\eqref{eq: def of int [dx]} are generalizations of $\int_X (\cdot)\,d\chi$ where
\begin{align}\label{eq: dX = [dX] for integer-valued functions}
    \begin{aligned}
        \int_X \mathbbm{1}_K(x) \,\lfloor d\chi(x)\rfloor=\int_X \mathbbm{1}_K(x) \,\lceil d\chi(x)\rceil &= \int_X \mathbbm{1}_K(x) \,[d\chi(x)] = \int_X \mathbbm{1}_K(x) \,d\chi(x) = \chi(K),
    \end{aligned}
\end{align}
for all $K\in\operatorname{CS}(X)$. The proof of Eq.~\eqref{eq: dX = [dX] for integer-valued functions} is in Appendix \ref{section: Proof of eq: dX = [dX] for integer-valued functions}. However, for general integrands, $\int_X g(x) \,\lfloor d\chi(x)\rfloor$ and $\int_X g(x) \,\lceil d\chi(x)\rceil$ are not equal (see ``Lemma 1" of \cite{baryshnikov2010euler}). Neither $\int_X (\cdot)\,\lfloor d\chi\rfloor$ nor $\int_X (\cdot)\,\lceil d\chi\rceil$ is linear. What is more, neither of them is homogeneous --- they are only positively homogeneous (see \cite{baryshnikov2010euler} for details). Fortunately, the ``averaged version" $\int_X (\cdot)\,[d\chi]$ is homogeneous. This means that $\int_X \lambda\cdot g(x)\,[d\chi(x)] = \lambda\cdot \int_X g(x)\,[d\chi(x)]$ for all $\lambda\in\mathbb{R}$, which is implied by ``Lemmas 4 and 6" of \cite{baryshnikov2010euler}. 

Within the trio of functionals outlined in Eq.~\eqref{eq: def of int [dx]}, our study predominantly employs the averaged version, $\int_X (\cdot)[d\chi]$, to ensure the homogeneity of the proposed ERT. To elucidate, consider the ERT denoted as $\operatorname{ERT}: g\mapsto \operatorname{ERT}(g)$. Our objective is to maintain homogeneity such that $\operatorname{ERT}(\lambda\cdot g)=\lambda\cdot\operatorname{ERT}(g)$ holds universally for any $\lambda\in\mathbb{R}$. This homogeneity property is validated using the averaged version $\int_X (\cdot)[d\chi]$ as detailed in Theorem \ref{thm: homogeneity of ERT} in Section \ref{section: Euler-Randon Transform}. This choice not only streamlines the theoretical presentation but also yields computational efficiency in practical applications. For example, suppose we need to rescale (and maybe also white-to-black transition as presented in Figure \ref{fig:grayscale_image}) the grayscale function $g$ post $\operatorname{ERT}(g)$ computation. In that case, we may directly rescale the computed ERT --- meaning that we can compute $\lambda\cdot\operatorname{ERT}(g)$ instead of computing the ERT of the rescaled image $\lambda\cdot g$. Rescaling a computed ERT is much more efficient than computing the ERT of a rescaled image.

\subsection{Euler-Randon Transform}\label{section: Euler-Randon Transform}

In this subsection, we introduce the precise definition of the ERT for grayscale images. Without loss of generality, we postulate that all functions representing grayscale images of interest are defined on the open ball $B_{\mathbb{R}^d}(0,R)$ with a prespecified radius $R<\infty$ (after all, there is no infinitely large image in practice). We model grayscale images/functions by the following definition.
\begin{definition}\label{def: grayscale iamges}
    Any element in the function class $\mathfrak{D}_{R,d}$ defined as follows is called a \textbf{grayscale function} or \textbf{grayscale image}
    \begin{align*}
        \mathfrak{D}_{R,d} := \left\{ g\in \operatorname{Def}\left(B_{\mathbb{R}^d}(0,R) \,;\, \mathbb{R}\right) \,:\, \sup_{x}\vert g(x)\vert <\infty \text{ and }\operatorname{dist}\Big(\supp(g), \partial B_{\mathbb{R}^d}(0,R)\Big)>0\right\},
    \end{align*}
    where $\operatorname{dist}(A, B):=\inf\{\Vert a-b \Vert:\, a\in A \text{ and }b\in B\}$ denotes the distance between two sets. 
\end{definition}
\noindent The condition $\operatorname{dist}\left(\supp(g), \partial B_{\mathbb{R}^d}(0,R)\right)>0$ in Definition \ref{def: grayscale iamges} means that the support of every grayscale function is strictly smaller than the domain $B_{\mathbb{R}^d}(0,R)$. This condition simplifies the proofs of Theorem \ref{thm: SERT preserves all information of ERT} and Eq.~\eqref{eq: formula inversion formula of ERT} and can be easily satisfied (e.g., we can always enlarge the radius $R$ and extend $g$ by zero to satisfy the condition).

\paragraph*{Definition of the ERT.} Using the Euler integration $\int_X (\cdot)\,[d\chi]$ defined in Eq.~\eqref{eq: ECT via Euler integration}, we define the ERT of grayscale functions as follows
\begin{align}\label{eq: def of Euler-Radon transform}
\begin{aligned}
\operatorname{ERT}:\ \ & \mathfrak{D}_{R,d} \rightarrow \mathbb{R}^{\mathbb{S}^{d-1}\times[0,T]},\\
& g \mapsto \operatorname{ERT}(g)=\left\{\operatorname{ERT}(g)(\nu,t)\right\}_{(\nu,t)\in \mathbb{S}^{d-1}\times [0,T]}, \\
& \text{where }\ \operatorname{ERT}(g)(\nu,t):=\int_{B_{\mathbb{R}^d}(0,R)} g(x) \cdot R(x,\nu,t) \, [d\chi(x)],
\end{aligned}
\end{align}
and $T=2R$. We may also replace the averaged version $\int (\cdot)\,[d\chi]$ in Eq.~\eqref{eq: def of Euler-Radon transform} with the floor version $\int (\cdot)\,\lfloor d\chi \rfloor$ or ceiling version $\int (\cdot)\,\lceil d\chi \rceil$; the transforms corresponding to $\int (\cdot)\,\lfloor d\chi \rfloor$ and $\int (\cdot)\,\lceil d\chi \rceil$ are denoted as $\lfloor\operatorname{ERT}\rfloor(g)$ and $\lceil\operatorname{ERT}\rceil(g)$, respectively. Eq.~\eqref{eq: dX = [dX] for integer-valued functions} indicates that the ERT, as well as $\lfloor\operatorname{ERT}\rfloor(g)$ and $\lceil\operatorname{ERT}\rceil(g)$, is a generalization of the ECT in the following sense 
\begin{align}\label{eq: ERT is a generalization of ECT}
    \operatorname{ERT}(\mathbbm{1}_K) = \lfloor\operatorname{ERT}\rfloor(\mathbbm{1}_K) = \lceil\operatorname{ERT}\rceil(\mathbbm{1}_K) = \operatorname{ECT}(K), \ \ \ \text{ for all }K \in \operatorname{CS}\left(B_{\mathbb{R}^d}(0,R)\right).
\end{align}

\paragraph*{Properties of the ERT.} Here, we present several properties of the ERT that will be utilized in later sections. First, the homogeneity of $\int (\cdot)\,[d\chi]$ and ``Lemma 6" of \cite{baryshnikov2010euler} directly imply the following theorem (its proof is omitted)
\begin{theorem}\label{thm: homogeneity of ERT}
    Suppose $g\in\mathfrak{D}_{R,d}$. Then, we have the following:
    \begin{enumerate}
        \item $\lfloor\operatorname{ERT}\rfloor$ and $\lceil\operatorname{ERT}\rceil$ are positively homogeneous, such that $\lfloor\operatorname{ERT}\rfloor(\lambda\cdot g)=\lambda\cdot\lfloor\operatorname{ERT}\rfloor(g)$ and $\lceil\operatorname{ERT}\rceil(\lambda\cdot g)=\lambda\cdot\lceil\operatorname{ERT}\rceil(g)$ for all $\lambda>0$.
        \item $\operatorname{ERT}$ is homogeneous, such that $\operatorname{ERT}(\lambda\cdot g)=\lambda\cdot \operatorname{ERT}(g)$ for all $\lambda\in\mathbb{R}$.
    \end{enumerate}
\end{theorem}
\noindent Secondly, we have the following theorem on the measurability of the function $(\nu,t) \mapsto \operatorname{ERT}(g)(\nu,t)$.
\begin{theorem}\label{thm: tameness of ERT}
    Suppose $g\in\mathfrak{D}_{R,d}$. Then, the function $(\nu,t) \mapsto \operatorname{ERT}(g)(\nu,t)$ is Borel-measurable, which holds for $\lfloor\operatorname{ERT}\rfloor$ and $\lceil\operatorname{ERT}\rceil$ as well.
\end{theorem}
\noindent Theorem \ref{thm: tameness of ERT} will be a straightforward result of a later Theorem \ref{thm: relationship between our proposed ERT and the referred existing transforms}.

\paragraph*{Definition of the SERT.} Theorem \ref{thm: tameness of ERT} allows us to define the smooth Euler-Radon transform (SERT) as follows,
\begin{align}\label{eq: def of SERT}
\begin{aligned}
\operatorname{SERT}:\ \ & \mathfrak{D}_{R,d} \rightarrow \mathbb{R}^{\mathbb{S}^{d-1}\times [0,T]},\\
& g \mapsto \operatorname{SERT}(g)=\left\{\operatorname{SERT}(g)(\nu,t)\right\}_{(\nu,t)\in \mathbb{S}^{d-1}\times [0,T]},  \\
& \text{where }\ \operatorname{SERT}(g)(\nu,t) := \int_0^t \operatorname{ERT}(g)(\nu,\tau) \,d\tau-\frac{t}{T}\int_0^T \operatorname{ERT}(g)(\nu,\tau) \,d\tau.
\end{aligned}
\end{align}
Theorem \ref{thm: tameness of ERT} implies that the Lebesgue integrals in Eq.~\eqref{eq: def of SERT} are well-defined, and the function $(\nu,t) \mapsto \operatorname{SERT}(g)(\nu,t)$ is Borel-measurable. Transforms $\lfloor\operatorname{SERT}\rfloor$ and $\lceil\operatorname{SERT}\rceil$ are defined via $\lfloor\operatorname{ERT}\rfloor$ and $\lceil\operatorname{ERT}\rceil$, respectively, in a similar way; they also have the referred properties of $\operatorname{SERT}$. Furthermore, Eq.~\eqref{eq: ERT is a generalization of ECT} implies that $\operatorname{SERT}(\mathbbm{1}_K) = \lfloor\operatorname{SERT}\rfloor(\mathbbm{1}_K) = \lceil\operatorname{SERT}\rceil(\mathbbm{1}_K) = \operatorname{SECT}(K)$ for all $K \in \operatorname{CS}\left(B_{\mathbb{R}^d}(0,R)\right)$.

\paragraph*{Comparing the ERT and SERT.} The following theorem indicates that the SERT preserves all information about the ERT under a regularity condition.
\begin{theorem}\label{thm: SERT preserves all information of ERT}
    Suppose $g\in\mathfrak{D}_{R,d}$. If $t \mapsto \operatorname{ERT}(g)(\nu, t)$ is a right continuous function for each fixed $\nu\in\mathbb{S}^{d-1}$, then $\operatorname{ERT}(g)$ can be expressed in terms of $\operatorname{SERT}(g)$; hence, $\operatorname{ERT}(g)$ and $\operatorname{SERT}(g)$ determine each other.
\end{theorem}
\noindent The proof of Theorem \ref{thm: SERT preserves all information of ERT} is given in Appendix \ref{eq: proof, ERT and SERT determine each other}. We selected right continuity over left continuity in Theorem \ref{thm: SERT preserves all information of ERT} for four reasons. The first reason is to align with Morse theory (see Remark 2.4 in \cite{milnor1963morse}, on page 20 therein, for a right-continuity result). The second reason is that $t\mapsto\operatorname{ERT}(g)(\nu,t)$ is not left continuous in general (see Appendix \ref{section: ECT is not left continuous} for examples with right-continuity). The third reason comes from probability theory. When $g$ is random, the function $t\mapsto\operatorname{ERT}(g)(\nu,t)$ is a stochastic process for each fixed $\nu\in\mathbb{S}^{d-1}$; if $t\mapsto\operatorname{ERT}(g)(\nu,t)$ is right continuous, it automatically becomes a stochastic process whose sample paths are right continuous with left limit (RCLL). Stochastic processes with RCLL sample paths are well studied in probability theory (e.g., see Section 21.4 of \cite{klenke2013probability}). The fourth reason is rooted in the following deliberation: if we view the Euler characteristic $\chi$ as an analog of a probability measure $\mathbb{P}$, then $\operatorname{ERT}(\mathbbm{1}_K)(\nu,t)=\chi(\{x\in K;\, x\cdot\nu+R\le t\})$ is an analog of the cumulative distribution function $F_X(t):=\mathbb{P}(X\le t)$ of a random variable $X$ and the function $t\mapsto F_X(t)$ is right continuous.

\paragraph*{Invertibility of the ERT and SERT.} In practical applications, grayscale images must be discretized into arrays of pixel intensities (or their higher-dimensional counterparts such as voxels) for storage on electronic devices. Consequently, we can represent grayscale images utilized in these contexts as members of the following piecewise-constant function class
\begin{align}\label{eq: def of piece-wise grayscale image space}
    \mathfrak{D}_{R, d}^{pc}:=\left\{g\in\mathfrak{D}_{R,d}:\, g \text{ takes finitely many values in }\mathbb{R}\right\}.
\end{align}
For any piecewise-constant grayscale image $g$ in $\mathfrak{D}_{R, d}^{pc}$, the following theorem indicates that $\operatorname{ERT}(g)$ does not lose any information about the image $g$ and, as a result, justifies the implementation of the ERT in practical applications.  
\begin{theorem}\label{thm: invertibility on piecewise images}
    (i) The restriction of $\operatorname{ERT}$ on $\mathfrak{D}_{R, d}^{pc}$ is invertible for all dimensions $d$. (ii) The restriction of $\operatorname{SERT}$ on $\{g\in\mathfrak{D}_{R, d}^{pc}: \text{ the function $t\mapsto\operatorname{ERT}(g)(\nu,t)$ is right continuous for each fixed }\nu\}$ is invertible for all dimensions $d$.
\end{theorem}
\noindent Theorem \ref{thm: invertibility on piecewise images} extends the invertibility findings of the ECT and SECT in ``Corollary 1" of \cite{ghrist2018persistent}. A comprehensive proof for Theorem \ref{thm: invertibility on piecewise images} can be found in Appendix \ref{proof of thm: invertibility on piecewise images}. The piecewise constant constraint in Theorem \ref{thm: invertibility on piecewise images} can be slightly relaxed. Namely, the invertibility of the ERT holds for all grayscale functions $g\in \mathfrak{D}_{R, d}$ that satisfy the ``Fubini condition" (see Appendix \ref{section: Discussions on the Invertibility of ERT and SERT}). The invertibility of the ERT on $\mathfrak{D}_{R, d}$, instead of $\mathfrak{D}_{R, d}^{pc}$, is still an open problem and is discussed in detail in Appendix \ref{section: Discussions on the Invertibility of ERT and SERT}. The main obstacle in solving the open problem is that the ``Fubini theorem" in Euler calculus \citep[][Section 3.8]{ghrist2014elementary} does not hold over real-valued Euler integrands. 

\subsection{Dual Euler-Radon Transform}

As a remark on the invertibility of the ERT presented in Theorem \ref{thm: invertibility on piecewise images}, we introduce the dual Euler-Radon transform (DERT). Let $R'$ be the dual kernel of $R(x, v, t)$ (see Eq.~\eqref{eq: Euler function representation of ECT}) given by
\begin{align}\label{eq: dual kernel R'}
    R'(\nu, t, x) \coloneqq \mathbbm{1}\left\{\left(\nu,t, x\right)\in \mathbb{S}^{d-1}\times[0,T]\times B_{\mathbb{R}^d}(0,R)\,:\, x\cdot\nu\ge t - R\right\}.
\end{align}
We define the DERT as follows
\begin{align}\label{eq: def of dual Euler-Radon transform}
\begin{aligned}
\operatorname{DERT}:\ \ & \operatorname{Def}(\mathbb{S}^{d-1}\times [0,T]) \rightarrow \mathbb{R}^{B_{\R^d}(0, R)},\\
& h \mapsto \operatorname{DERT}(h)=\left\{\operatorname{DERT}(h)(x):=\int_{\mathbb{S}^{d-1} \times [0, T]} h(\nu, t) \cdot R'(\nu, t, x) \, [d\chi(\nu, t)]\right\}_{x \in B_{\R^d}(0, R)}.
\end{aligned}
\end{align}
Using the DERT, the following provides an inversion formula for recovering $g\in\mathfrak{D}_{R,d}^{pc}$ from $\operatorname{ERT}(g)$
\begin{align}\label{eq: formula inversion formula of ERT}
    g(x') = \frac{1}{(-1)^{d+1}}\cdot (\operatorname{DERT} \circ \operatorname{ERT})(g)(x') - \lim_{\xi\rightarrow \partial B_{\R^d}(0,R)} \frac{1}{(-1)^{d+1}}\cdot (\operatorname{DERT} \circ \operatorname{ERT})(g)(\xi),
\end{align}
where $\lim_{\xi\rightarrow \partial B_{\R^d}(0,R)}$ means that $\xi$ converges to a point on the sphere $\partial B_{\R^d}(0,R)=\{x\in\R^d:\, \Vert x\Vert=R\}$. The details and proof of Eq.~\eqref{eq: formula inversion formula of ERT} are given in Appendix \ref{section: Discussions on the Invertibility of ERT and SERT}.

\section{Existing Frameworks}\label{section: Existing Frameworks}

In Section \ref{section: Euler-Randon Transform}, we demonstrated that the introduced ERT and SERT serve as generalizations of the ECT and SECT, respectively. In this section, we discuss the relationship between our proposed framework and other established transforms: the WECT, LECT, SELECT, and the marginal Euler curve (MEC). 

\paragraph*{LECT and SELECT.} \cite{kirveslahti2023representing} introduced the LECT and SELECT for the analysis of scalar fields (including grayscale images), motivated by the idea of super-level sets implemented in the topology of Gaussian random fields \citep{adler2007random, taylor2008random}. Using the notations introduced in Section \ref{section: Euler-Randon Transform}, the LECT can be represented as follows
\begin{align}\label{eq: def of LECT}
    \begin{aligned}
        \operatorname{LECT}:\ \ & \mathfrak{D}_{R,d} \rightarrow \mathbb{Z}^{\mathbb{S}^{d-1}\times[0,T]\times[0,1]}, \\
    & g \mapsto \operatorname{LECT}(g):=\left\{\operatorname{LECT}(g)(\nu,t,s)\right\}_{(\nu,t,s) \in \mathbb{S}^{d-1}\times[0,T]\times[0,1]}, \\
    & \text{where }\ \operatorname{LECT}(g)(\nu,t,s):=\chi\left(\left\{x\in B_{\mathbb{R}^d}(0,R):\, x\cdot\nu\le t-R \text{ and } g(x)=s\right\}\right).
    \end{aligned}
\end{align}
That is, the LECT transforms $d$-dimensional grayscale images into integer-valued functions defined on a $(d+1)$-dimensional manifold. Inspired by Morse theory \citep{milnor1963morse} and considerations of statistical robustness, the SELECT can be represented as follows
\begin{align}\label{eq: def of SELECT}
\begin{aligned}
    \operatorname{SELECT}:\ \ & \mathfrak{D}_{R,d} \rightarrow \mathbb{Z}^{ \mathbb{S}^{d-1}\times[0,T]\times[0,1] }, \\
    & g \mapsto \operatorname{SELECT}(g):=\left\{\operatorname{SELECT}(g)(\nu,t,s)\right\}_{(\nu,t,s) \in \mathbb{S}^{d-1}\times[0,T]\times[0,1]}, \\
    & \text{where }\ \operatorname{SELECT}(g)(\nu,t,s):=\chi\left(\left\{x\in B_{\mathbb{R}^d}(0,R):\, x\cdot\nu\le t-R \text{ and }g(x) \ge s\right\}\right).
\end{aligned}
\end{align}
We have the following result on the functions $(\nu,t,s) \mapsto \operatorname{LECT}(g)(\nu,t,s)$ and $(\nu,t,s) \mapsto \operatorname{SELECT}(g)(\nu,t,s)$, which is an analog of Theorem \ref{thm: tameness theorem}.
\begin{theorem}\label{thm: tameness of LECT and SELECT}
    Suppose $g\in \mathfrak{D}_{R,d}$. Then, we have the following:
    \begin{enumerate}
        \item $\operatorname{LECT}(g)(\nu,t,s)$ and $\operatorname{SELECT}(g)(\nu,t,s)$ take only finitely many values as $(\nu,t,s)$ runs through $\mathbb{S}^{d-1}\times[0,T]\times[0,1]$. In addition, for each integer $z \in \mathbb{Z}$, the sets $\{(\nu,t)\in \mathbb{S}^{d-1}\times[0,T]\times[0,1]:\, \operatorname{LECT}(g)(\nu,t,s) = z\}$ and $\{(\nu,t)\in \mathbb{S}^{d-1}\times[0,T]\times[0,1]:\, \operatorname{SELECT}(g)(\nu,t,s) = z\}$ are definable; hence, the functions $(\nu,t,s) \mapsto \operatorname{LECT}(g)(\nu,t,s)$ and $(\nu,t,s) \mapsto \operatorname{SELECT}(g)(\nu,t,s)$ are definable.
        \item The functions $(\nu,t,s) \mapsto \operatorname{LECT}(g)(\nu,t,s)$ and $(\nu,t,s) \mapsto \operatorname{SELECT}(g)(\nu,t,s)$ are Borel-measurable.
    \end{enumerate}
\end{theorem}
\noindent Since the proof of Theorem \ref{thm: tameness of LECT and SELECT} is similar to that of Theorem \ref{thm: tameness theorem}, we omit it.

\paragraph*{MEC.} \cite{kirveslahti2023representing} also proposed the marginal Euler curve (MEC) $M_\nu^g (t)$ which is defined via the SELECT as follows
\begin{align}\label{eq: def of MEC}
    M_\nu^g (t) := \int_\mathbb{R} \operatorname{SELECT}(g)(\nu, t, s) \, ds,\ \ \ \text{ for all }(\nu,t)\in\mathbb{S}^{d-1}\times[0,T].
\end{align}
Theorem \ref{thm: tameness of LECT and SELECT} guarantees that the Lebesgue integral in Eq.~\eqref{eq: def of MEC} is well-defined.

\paragraph*{Relationship with the ERT.} The subsequent theorem establishes a connection between our proposed ERT and the LECT and SELECT --- thereby also linking the ERT to both the MEC and WECT.
\begin{theorem}\label{thm: relationship between our proposed ERT and the referred existing transforms}
Suppose $g\in\mathfrak{D}_{R,d}$. Then, we have the following Euler integration representation of the $\operatorname{ERT}$ via the $\operatorname{LECT}$
\begin{align}\label{eq: Euler representation of ERT via LECT}
    \operatorname{ERT}(g)(\nu, t) = \int_\mathbb{R} s \cdot \operatorname{LECT}(g)(\nu, t, s) \,[d\chi(s)], \ \ \ \text{for all }(\nu, t) \in \mathbb{S}^{d-1} \times [0, T].
\end{align}
In addition, we have the following Lebesgue integration representations of the $\lfloor\operatorname{ERT}\rfloor$ and $\lceil\operatorname{ERT}\rceil$ via the $\operatorname{LECT}$ and $\operatorname{SELECT}$
\begin{align}\label{eq: Lebesgue representations of the floor and ceiling ERT}
    \begin{aligned}
        & \lfloor\operatorname{ERT}\rfloor(g)(\nu,t) = \int_0^\infty \operatorname{SELECT}(g)(\nu, t, s) - \operatorname{SELECT}(-g)(\nu, t, s) + \operatorname{LECT}(-g)(\nu, t, s) \,ds,\\
    & \lceil\operatorname{ERT}\rceil(g)(\nu,t) = \int_0^\infty \operatorname{SELECT}(g)(\nu, t, s) - \operatorname{LECT}(g)(\nu, t, s) - \operatorname{SELECT}(-g)(\nu, t, s) \,ds.
    \end{aligned}
\end{align}
In particular, we have the following Lebesgue integration representation of the $\operatorname{ERT}$
\begin{align}\label{eq: Lebesgue representation of ERT}
    \begin{aligned}
        \operatorname{ERT}(g)(\nu,t) = & \int_0^\infty \left\{ \operatorname{SELECT}(g)(\nu, t, s) - \operatorname{SELECT}(-g)(\nu, t, s) \right\} \,ds \\
    & + \frac{1}{2}\int_0^\infty \left\{ \operatorname{LECT}(-g)(\nu, t, s) - \operatorname{LECT}(g)(\nu, t, s) \right\} \,ds.
    \end{aligned}
\end{align}
\end{theorem}
\noindent The proof of Theorem \ref{thm: relationship between our proposed ERT and the referred existing transforms} is in Appendix \ref{proof: relationship between our proposed ERT and the referred existing transforms}. Theorem \ref{thm: tameness of LECT and SELECT} guarantees that the Euler integral and Lebesgue integrals in Eqs.~\eqref{eq: Euler representation of ERT via LECT}-\eqref{eq: Lebesgue representation of ERT} are well-defined. The representations in Theorem \ref{thm: relationship between our proposed ERT and the referred existing transforms} will be used to compute the ERT in the next section. Furthermore, with the Fubini theorem, the Lebesgue integration representations in Eq.~\eqref{eq: Lebesgue representations of the floor and ceiling ERT} and Eq.~\eqref{eq: Lebesgue representation of ERT} imply the result in Theorem \ref{thm: tameness of ERT}.

The representations in Theorem~\ref{thm: relationship between our proposed ERT and the referred existing transforms} connect to our proposed ERT to the MEC and WECT in the following ways:
\begin{enumerate}
    \item If $g(x)\ge 0$ for all $x$, then we have $\{x\in B_{\mathbb{R}^d}(0,R):\, x\cdot\nu\le t-R \text{ and }-g(x) \ge s\}=\emptyset$ for all $s>0$. Hence, $\lfloor\operatorname{ERT}\rfloor(g)(\nu,t) = \int_0^\infty \operatorname{SELECT}(g)(\nu, t, s) \,ds=M_\nu^g(t)$ --- meaning that the $\lfloor\operatorname{ERT}\rfloor$ is equal to the MEC specified in Eq.~\eqref{eq: def of MEC} for nonnegative grayscale functions.
    \item If $g$ is nonnegative and $\operatorname{LECT}(g)(\nu,t,s)=0$ for almost every $s$ with respect to the Lebesgue measure, Eq.~\eqref{eq: Lebesgue representation of ERT} implies that $\operatorname{ERT}(g)(\nu,t)=M_\nu^g(t)$. From this viewpoint, we can easily show that the WECT proposed in \cite{jiang2020weighted} is a special case of our proposed ERT. \cite{jiang2020weighted} models each grayscale image as a ``weighted simplicial complex" in the form $g=\sum_{\sigma\in\mathfrak{S}} a_\sigma\cdot\mathbbm{1}_\sigma$ for a finite sum (each $\sigma$ is a simplex, $a_\sigma\in\mathbb{N}$, and $\mathfrak{S}$ is a finite collection of simplexes), assuming the ``consistency condition" $a_\sigma=\max\{a_\tau: \tau\in\mathfrak{S}\text{ and $\sigma$ is a face of $\tau$}\}$ for all $\sigma\in\mathfrak{S}$. The WECT of $g$ is defined as follows 
    \begin{align*}
        \operatorname{WECT}(g)(\nu,t) := \sum_{d=0}^{\max\{\operatorname{dim}(\sigma):\,\sigma\in\mathfrak{S}\}} (-1)^d\cdot \left(\sum_{\sigma\in\mathfrak{S},\, \operatorname{dim}(\sigma)=d,\,\text{and }\sigma\subseteq\{x\in\mathbb{R}^d: x\cdot\nu\le t-R\}} a_\sigma\right).
    \end{align*}
    \cite{kirveslahti2023representing} show that the WECT of $g$ coincides with the MEC of $g$ where $\operatorname{WECT}(g)(\nu,t)= M_\nu^g(t)$. Furthermore, the definition of the LECT in Eq.~\eqref{eq: def of LECT} indicates that $\operatorname{LECT}(g)(\nu,t,s)=0$ unless $s=a_\sigma\in\mathbb{N}$ for some simplex $\sigma$. Therefore, $\int_{0}^1 \operatorname{LECT}(h)(\nu, t, s) \, ds=0$. Hence, we have $\operatorname{WECT}(g)(\nu,t)= M_\nu^g(t)=\operatorname{ERT}(g)(\nu,t)$ for all $(\nu,t)\in\mathbb{S}^{d-1}\times[0,T]$, if $g=\sum_{\sigma} a_\sigma\cdot\mathbbm{1}_\sigma$ with $a_\sigma\in\mathbb{N}$.
\end{enumerate}

\paragraph*{Dissimilarities between Grayscale Images.} Lastly, we may use the ERT and SERT to measure the (dis-)similarity between two grayscale functions. Inspired by the dissimilarity quantities defined in the literature \citep{turner2014persistent, crawford2020predicting, meng2022randomness, marsh2023stability, wang2023hypothesis}, we introduce the following (semi-)distances between grayscale functions $g_1,g_2\in\mathfrak{D}_{R,d}$
\begin{align}\label{eq: semi distances between grayscale functions}
    \begin{aligned}
    & \operatorname{dist}^{\operatorname{ERT}}_{p,q}(g_1, g_2):= \left\Vert \,\operatorname{ERT}(g_1) - \operatorname{ERT}(g_2) \,\right\Vert_{L_\nu^q L_t^p},  \\
    & \operatorname{dist}^{\operatorname{SERT}}_{p,q}(g_1, g_2):= \left\Vert\, \operatorname{SERT}(g_1) - \operatorname{SERT}(g_2) \,\right\Vert_{L_\nu^q L_t^p},
    \end{aligned}
\end{align}
where $1\le p,q \le \infty$, and the $L_\nu^q L_t^p$-norm $\Vert f\Vert_{L_\nu^q L_t^p}$ of a function $(\nu,t)\mapsto f(\nu,t)$ is defined in a two-step process. Initially, we take the $L^p$-norm $\Vert f(\nu,\cdot)\Vert_{L_t^p}$ with respect to $t\in[0,T]$; then we next take the $L^q$-norm with respect to $\nu\in\mathbb{S}^{d-1}$ (regarding the spherical measure $d\nu$ on $\mathbb{S}^{d-1}$). Theorem \ref{thm: tameness of ERT} implies that the $L_\nu^q L_t^p$-norm used in Eq.~\eqref{eq: semi distances between grayscale functions} is well-defined. Theorem \ref{thm: invertibility on piecewise images} indicates that $\operatorname{dist}^{\operatorname{ERT}}_{p,q}$ is a distance, rather than just a semi-distance, on the space $\mathfrak{D}_{R,d}^{pc}$ defined in Eq.~\eqref{eq: def of piece-wise grayscale image space}. Let $K_1$ and $K_2$ belong to $\operatorname{CS}(B_{\mathbb{R}^d}(0,R))$. Then, we have the following: 
\begin{itemize}
    \item $\operatorname{dist}^{\operatorname{SERT}}_{p,p}(\mathbbm{1}_{K_1}, \mathbbm{1}_{K_2})$ agrees with the SECT distance stated in \cite{crawford2020predicting};

    \item $\operatorname{dist}^{\operatorname{ERT}}_{p,p}(\mathbbm{1}_{K_1}, \mathbbm{1}_{K_2})$ agrees with the ECT distance stated in \cite{curry2022many};

    \item $\operatorname{dist}^{\operatorname{ERT}}_{2,\infty}(\mathbbm{1}_{K_1}, \mathbbm{1}_{K_2})$ agrees with the distance stated in \cite{meng2022randomness}, which underpins the generation of the Borel algebra implemented in that work;

    \item $\operatorname{dist}^{\operatorname{ERT}}_{1,\infty}(\mathbbm{1}_{K_1}, \mathbbm{1}_{K_2})$ aligns with the distance stated in \cite{marsh2023stability} for the examination of the stability of the ECT;

    \item $\operatorname{dist}^{\operatorname{SERT}}_{2,\infty}(\mathbbm{1}_{K_1}, \mathbbm{1}_{K_2})$ agrees with the distance utilized in \cite{wang2023hypothesis} for permutation test.
\end{itemize}
Futhermore, the distances defined in Eq.~\eqref{eq: semi distances between grayscale functions} are the analogs of the following distances proposed in \cite{kirveslahti2023representing}
\begin{align}\label{eq: metrics defined in kirveslahti2023representing}
    \begin{aligned}
        & \operatorname{dist}^{\operatorname{SELECT}}_p(g_1, g_2):= \left(\int_{\mathbb{S}^{d-1}} \int_0^T \int_{\mathbb{R}} \left\vert\, \operatorname{SELECT}(g_1)(\nu,t,s) - \operatorname{SELECT}(g_2)(\nu,t,s) \right\vert^p\, ds\, dt\, d\nu \right)^{1/p}, \\ 
        & \operatorname{dist}^{\operatorname{MEC}}_p(g_1, g_2):=\left(\int_{\mathbb{S}^{d-1}}\int_0^T \left\vert\, M_\nu^{g_1}(t) - M_\nu^{g_2}(t) \,\right\vert^p \,dt \,d\nu \right)^{1/p},
    \end{aligned}
\end{align}
for $1\le p<\infty$. Theorem \ref{thm: tameness of LECT and SELECT} implies that the Lebesgue integrals implemented in the $L^p$-norms in Eq.~\eqref{eq: metrics defined in kirveslahti2023representing} are well-defined. We will compare the performance of all the referred distances from a statistical inference perspective in Section \ref{section: Numerical Experiments}.

\commentout{
\cdedit{
\begin{lemma}
    The semi-distances above are well-defined distances.
\end{lemma}

\begin{proof}
    First, we note that the well-definedness of the (semi) distances depends on the measurability and integrability of the functions involved. This condition is satisfied for the functions in Eq.(2.11) by Lemma~\ref{lem::ert_measurable_integrable}. Therefore, we only need to show that $\operatorname{dist}(h_1, h_2)=0$ implies $h_1=h_2$.

\begin{itemize}
    \item We claim that  $\operatorname{dist}^{L_p}(h_1, h_2)=0$ implies $h_1=h_2$.

        Suppose $\operatorname{dist}^{L_p}(h_1, h_2)=0$. Then, 
$$\left(\int_{B_{\mathbb{R}^d}(0,R)}\vert h_1(x)-h_2(x)\vert^p dx\right)^{1/p} = 0.$$
Since $|h_1 - h_2|^p \geq 0$ for all $x \in B_{\mathbb{R}^d}(0,R)$, it follows that $|h_1 - h_2|^p = 0$ for almost every $x \in B_{\mathbb{R}^d}(0,R)$, which further implies that $h_1 = h_2$ for almost every $x \in B_{\mathbb{R}^d}(0,R)$. So $h_1 = h_2$ under $L_p$ space.
\item Given $\operatorname{dist}^{SERT}_{p,p}(\mathbbm{1}_{K_1}, \mathbbm{1}_{K_2}) = 0$, this is equivalent to say the $L_p$ distance between $ERT(h_1)(\nu, t)$ and $ERT(h_2)(\nu, t)$ has measure zero. 
\end{itemize}
\end{proof}
}
}

\begin{figure}[h]
    \centering
    \includegraphics[scale=0.13]{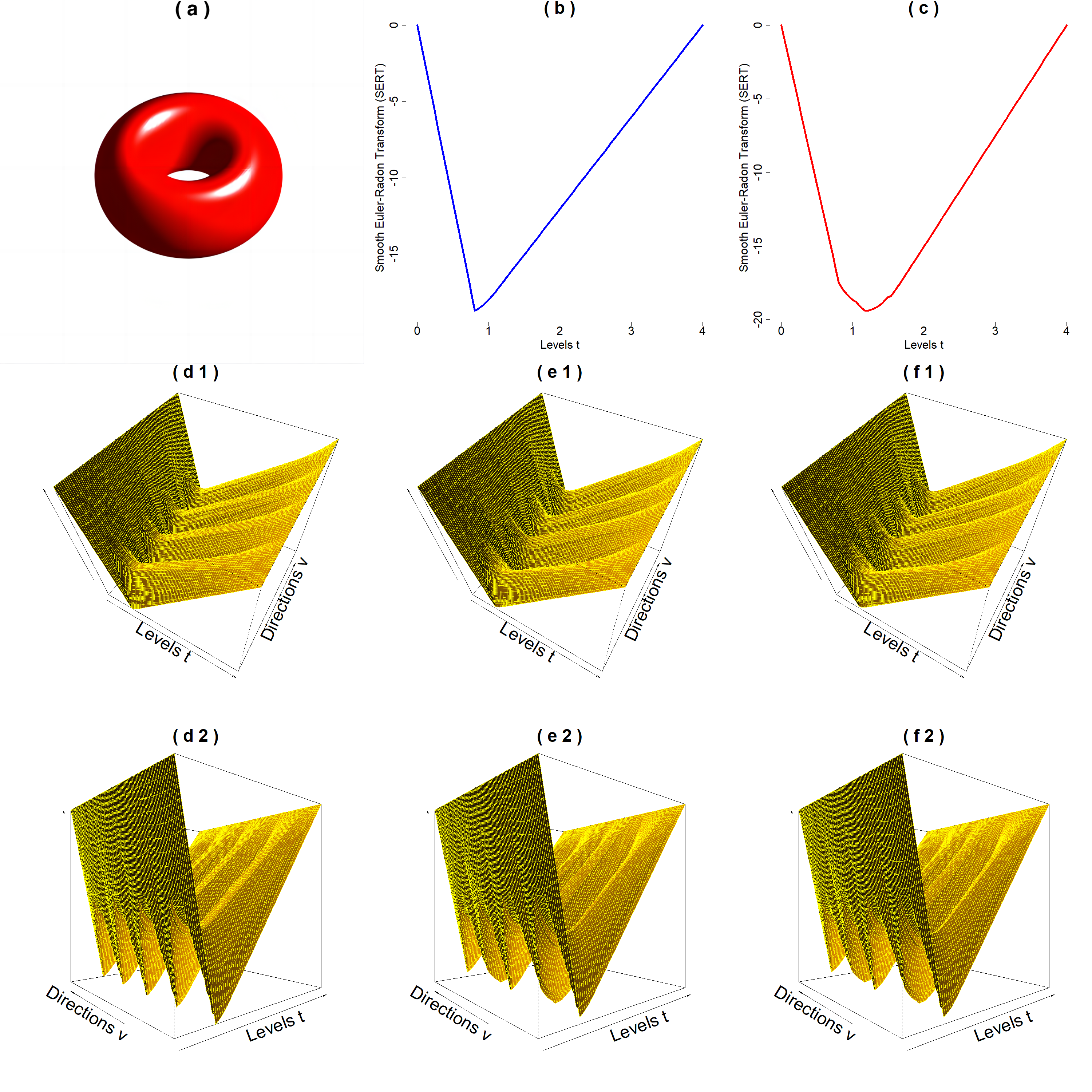}
    \caption{The torus in panel (a) presents the level set $\{x\in\R^3:\, g(x)=0.0834\}$, where $g$ is defined in Eq.~\eqref{eq: proof-of-concept scalar field example}. Panel (b) presents curve $t\mapsto\operatorname{SERT}(g)(\nu,t)$ with $\nu=(0,1,0)^\T$. Panel (c) presents curve $t\mapsto\operatorname{SERT}(g)(\nu,t)$ with $\nu=(0,0,1)^\T$. Panels (d1, d2, e1, e2, f1, f2) present the surfaces $(\theta, t) \mapsto \operatorname{SERT}(g)(\nu,t)$ with different directions $\nu$ as in Eq.~\eqref{eq: segments for visualizations}. Panels (d1, d2) present the same surface from different angles. Similar for panels (e1, e2) and (f1, f2).}
    \label{fig: Merged_document}
\end{figure}

\section{A Proof-of-Concept Example}\label{subsection: proof-of-concept simulation}

Our proposed SERT plays an important role in the statistical inference discussed in Section \ref{subsection: proof-of-concept simulation}, given its ability to convert grayscale images into functional data. Before delving into the SERT-based statistical inference, in this section, we illustrate the SERT via a proof-of-concept example. Here, we focus on dimensionality $d=3$ and the following scalar field
\begin{align}\label{eq: proof-of-concept scalar field example}
    g(x):=\left\{\left(\sqrt{\frac{3}{4}\left(x_1^2+x_2^2\right)}-\frac{1}{2}\right)^2+\frac{3x_3^2}{4}\right\}\cdot\mathbbm{1}_{\{-1\le x_1,\, x_2,\, x_3\le 1\}}, \ \ \ \text{ where }x=(x_1, x_2, x_3)^\T.
\end{align}
The $g$ in Eq.~\eqref{eq: proof-of-concept scalar field example} is a grayscale function (in the sense of Definition \ref{def: grayscale iamges}), where $g \in \mathfrak{D}_{R,d}$ with $d=3$ and $R=2$. A level set $\{x\in\R^3:\, g(x)=0.0834\}$ of $g$ is presented in Figure \ref{fig: Merged_document} (the level 0.0834 was specifically chosen to make the visualization of the level set look like a torus). We computed the ERT of the $g$ in a sequential manner. We first compute the LECT and SELECT using the MATLAB isosurface procedure. Next, we compute the ERT utilizing the Lebesgue integration representation in Eq.~\eqref{eq: Lebesgue representation of ERT}. Finally, the SERT is derived from the ERT via a standard numerical integration method.

The SERT of the $3$-dimensional grayscale image $g$ defined in Eq.~\eqref{eq: proof-of-concept scalar field example} is a scalar field over the 3-dimensional product manifold $\mathbb{S}^2\times[0,4]$, where $(\nu,t)\mapsto\operatorname{SERT}(g)(\nu,t)$ with $\nu=(\nu_1,\nu_2,\nu_3)^\T\in \mathbb{S}^2$. We visualize the following segments of the scalar field
\begin{align}\label{eq: segments for visualizations}
    \begin{aligned}
        & (\theta, t) \mapsto \operatorname{SERT}(g)(\nu,t) \ \text{with} \ \nu=(\cos\theta,\,\sin\theta, \, 0)^\T \ \text{in Figure \ref{fig: Merged_document}(d1, d2)}; \\
    & (\theta, t) \mapsto \operatorname{SERT}(g)(\nu,t) \ \text{with} \ \nu=(\cos\theta, \, 0, \, \sin\theta)^\T \ \text{in see Figure \ref{fig: Merged_document}(e1, e2)};  \\
    & (\theta, t) \mapsto \operatorname{SERT}(g)(\nu,t) \ \text{with}\ \nu=(0, \, \cos\theta, \, \sin\theta)^\T \ \text{in Figure \ref{fig: Merged_document}(f1, f2)};
    \end{aligned}
\end{align}
where $\theta\in[0,2\pi]$ and $t\in[0,4]$. The maps in Eq.~\eqref{eq: segments for visualizations} are scalar fields over $[0, 2\pi]\times[0,4]$ and are presented in the second and third rows of Figure \ref{fig: Merged_document}.

In Figure \ref{fig: Merged_document}, the surfaces corresponding to $(\theta, t) \mapsto \operatorname{SERT}(g)(\nu,t)$ consistently exhibit an approximate periodicity of $\pi/2$ in the variable $\theta$ (indicated by the axis label ``Direction $\nu$"). This approximate periodicity is fundamentally derived from the ``box-shape" indicator function $\mathbbm{1}_{\{-1\le x_1,, x_2,, x_3\le 1\}}$ in Eq.~\eqref{eq: proof-of-concept scalar field example}. The surfaces depicted in Figures \ref{fig: Merged_document}(e1, e2) and \ref{fig: Merged_document}(f1, f2) are indistinguishable, which is a characteristic attributed to the $x_1$-$x_2$ symmetry of the grayscale function $g$ defined in Eq.~\eqref{eq: proof-of-concept scalar field example}. The surfaces presented in Figure \ref{fig: Merged_document}(d1, d2) exhibit subtle distinctions from those in \ref{fig: Merged_document}(e1, e2, f1, f2). The curves and surfaces in Figure \ref{fig: Merged_document}, as well as the scalar field $(\nu,t)\mapsto\operatorname{SERT}(g)(\nu,t)$, suggest a potential association of the SERT with manifold learning \citep{yue2016parameterization, dunson2021inferring, meng2021principal, li2022efficient}.

\begin{figure}[h]
    \centering
    \includegraphics[scale=0.23]{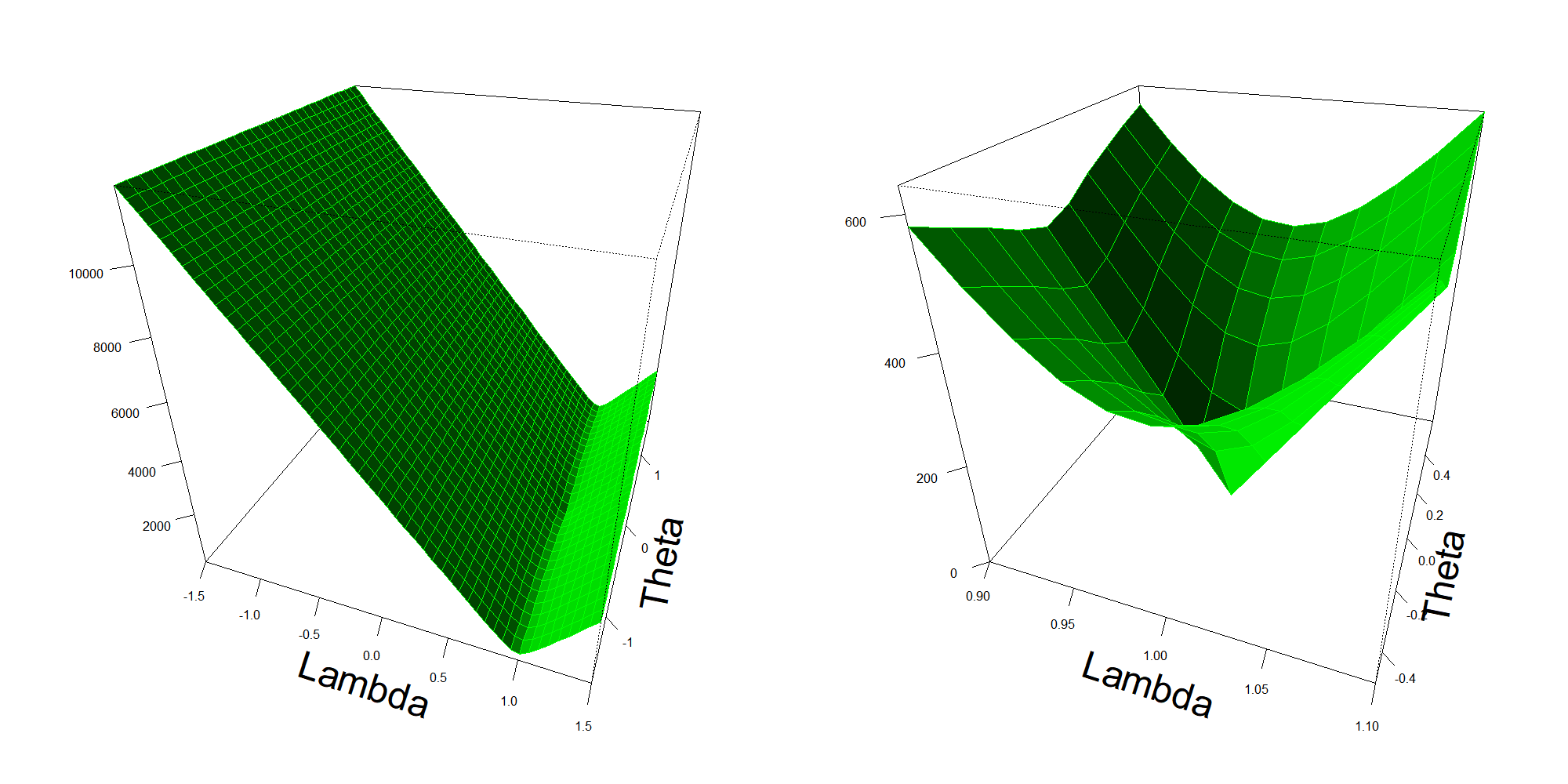}
    \caption{The left panel presents the surface of the function defined in Eq.~\eqref{eq: (theta, lambda) surface} for $\theta\in[-\pi/2,\, \pi/2]$ and $\lambda\in[-3/2,\,3/2]$. Notably, we have also considered the negative scaling parameters $\lambda<0$, which correspond to the ``white-to-black" transition. The right panel presents the same function with $\theta\in[-0.45,\,0.45]$ and $\lambda\in[0.9,\,1.1]$. Both surfaces indicate that the function in Eq.~\eqref{eq: (theta, lambda) surface} is minimized when $\lambda=1$ and $\theta=0$.}
    \label{fig: Alignment_distance}
\end{figure}

\section{Alignment of Images and Invariance of the ERT}\label{section: Alignment of Images and the Invariance of ERT}

In various applications, images are typically aligned before analysis \citep[e.g.,][]{bankman2008handbook}. \cite{wang2021statistical} utilized an ECT-based approach to align shapes (equivalently, binary images) through orthogonal actions. The primary aim of the ECT-based strategy is to lessen the difference between a pair of shapes resulting from the orthogonal movements. An in-depth exposition of the ECT-based method, along with a proof-of-concept example, can be found in Supplementary Section 4 of \cite{wang2021statistical}. Analogous correspondence free alignment techniques are also needed for grayscale images. \cite{kirveslahti2023representing} examined the invariance of the LECT with respect to orthogonal alignments and introduced an LECT/SELECT-based alignment method tailored for grayscale images. Motivated by the studies in both \cite{wang2021statistical} and \cite{kirveslahti2023representing}, in this section, we introduce an ERT-based alignment approach. This approach will serve as a preprocessing step for the ERT-based statistical inference detailed in Section \ref{section: Statistical Inference of Grayscale Functions}.

Let $g^\diamondsuit$ represent a reference image designated as a template. For any source image $g$ under consideration, we consider a collection of transforms, denoted by $\mathscr{T}$, encompassing all transformations $T$ of interest. From $\mathscr{T}$, we select the transformation $T^\blacklozenge\in\mathscr{T}$ such that the transformed image $T^\blacklozenge(g)$ is the most similar to $g^\diamondsuit$ based on a specified dissimilarity metric. Subsequent analysis is then conducted on the aligned image $T^\blacklozenge(g)$ instead of the original source image $g$. A potential choice for the dissimilarity metric can be one of the distances defined in Eq.\eqref{eq: semi distances between grayscale functions}, denoted as $\operatorname{dist}$. This can be summarized as follows
\begin{align}\label{eq: def of optimal transform}
    T^\blacklozenge:=\argmin_{T\in\mathscr{T}}\, \operatorname{dist}\Big(g^\diamondsuit, T(g)\Big).
\end{align}
\noindent Beyond the orthogonal actions implemented in \cite{wang2021statistical} and \cite{kirveslahti2023representing}, we also consider scaling transforms of grayscale images, including the ``white-to-black" transition. In this paper, we focus on the following collection of transforms
\begin{align}\label{eq: the collection of all scaling, rotation, and reflection transforms}
    \mathscr{T}=\left\{T:\, (Tg)(x)=\lambda\cdot g(\pmb{A}^{-1}x),\text{ where } \lambda\in\mathbb{R} \text{ and }\pmb{A}\in\operatorname{O}(d)\right\},
\end{align}
where $\operatorname{O}(d)$ denotes the orthogonal group in dimension $d$, and it contains all the rotations and reflections in $\mathbb{R}^d$. For ease of notation, we define the dual $\pmb{A}_*$ of $\pmb{A}$ by $(\pmb{A}_*g)(x):=g\left(\pmb{A}^{-1}x\right)$. The goal of employing minimization across the collection in Eq.~\eqref{eq: the collection of all scaling, rotation, and reflection transforms} is to mitigate disparities between two images arising from orthogonal movements and variations in pixel intensity scales.

A challenge in using the criterion presented in Eq.~\eqref{eq: def of optimal transform}, combined with any of the dissimilarity metrics in Eq.\eqref{eq: semi distances between grayscale functions}, is the computational cost of obtaining $\operatorname{ERT}(T(g))$ for all $T\in\mathscr{T}$. More specifically, the computation of $\operatorname{ERT}(\lambda\cdot \pmb{A}_*g)$ is required for every $\lambda\in\mathbb{R}$ and $\pmb{A}\in\operatorname{O}(d)$. Fortunately, the homogeneity of the ERT (see Theorem \ref{thm: homogeneity of ERT}) implies $\operatorname{ERT}(\lambda\cdot \pmb{A}_*g)=\lambda\cdot \operatorname{ERT}(\pmb{A}_*g)$ for all $\lambda\in\mathbb{R}$. Therefore, we can simply calculate $\operatorname{ERT}(\pmb{A}_*g)$ and get the $\operatorname{ERT}(\lambda\cdot \pmb{A}_*g)$ for all $\lambda\in\mathbb{R}$ by simply scaling the computed $\operatorname{ERT}(\pmb{A}_*g)$. Similarly, if we further have the ``$\pmb{A}_*$-homogeneity", where ``$\operatorname{ERT}(\pmb{A}_*g)=\pmb{A}_*\operatorname{ERT}(g)$," the amount of computation required in Eq.~\eqref{eq: def of optimal transform} is further reduced. The ``$\pmb{A}_*$-homogeneity" is true and accurately presented by the following result.
\begin{theorem}\label{thm: O(d)-invariance}
    For any $g\in\mathfrak{D}_{R,d}$ and $\pmb{A}\in\operatorname{O}(d)$, we have $\operatorname{ECT}(\pmb{A}_*g)(\nu,t)=\operatorname{ECT}(g)(\pmb{A}^{-1}\nu,t)=:\pmb{A}_*\operatorname{ECT}(g)(\nu,t)$ for all $\nu\in\mathbb{S}^{d-1}$ and $t\in[0,T]$.
\end{theorem}
\noindent Theorem \ref{thm: O(d)-invariance} is a direct result of ``Proposition 2.18" of \cite{kirveslahti2023representing} via Eq.~\eqref{eq: Lebesgue representation of ERT}; hence, we omit its proof. Combining the scalar homogeneity and $\pmb{A}_*$-homogeneity, we have the following
\begin{align}\label{eq: scaling and O(d) invariance of ERT}
    \operatorname{ECT}(\lambda\cdot\pmb{A}_*g)(\nu,t)=\lambda\cdot\operatorname{ECT}(g)(\pmb{A}^{-1}\nu,t)=:\lambda\cdot \pmb{A}_*\operatorname{ECT}(g)(\nu,t),
\end{align}
for all $g\in\mathfrak{D}_{R,d}$, $\nu\in\mathbb{S}^{d-1}$, and $t\in[0,T]$. Using $\operatorname{dist}=\operatorname{dist}^{\operatorname{ERT}}_{p,q}$ as an example (see Eq.~\eqref{eq: semi distances between grayscale functions} for the definition of $\operatorname{dist}^{\operatorname{ERT}}_{p,q}$), an optimal transformation across the collection in Eq.~\eqref{eq: the collection of all scaling, rotation, and reflection transforms} can be represented via Eq.~\eqref{eq: scaling and O(d) invariance of ERT} as follows
\begin{align}\label{eq: representation of optimal transform}
    \argmin_{\lambda\in\mathbb{R} \text{ and }\pmb{A}\in\operatorname{O}(d)}\, \left\Vert \,\operatorname{ERT}(g) - \lambda\cdot\pmb{A}_*\operatorname{ERT}(g) \,\right\Vert_{L_\nu^q L_t^p}.
\end{align}
In Eq.~\eqref{eq: representation of optimal transform}, we only need to compute the ERT of $g^\diamondsuit$ and $g$ instead of all the $\lambda\cdot\pmb{A}_*g$ for all $\lambda\in\mathbb{R}$ and $\pmb{A}\in\operatorname{O}(d)$. Notably, when both $g^\diamondsuit$ and $g$ are indicator functions representing constructible sets, the alignment method detailed in Eq.~\eqref{eq: representation of optimal transform} is equivalent to the ECT-based alignment method proposed in \cite{wang2021statistical}.

To illustrate the performance of the alignment approach described in Eq.~\eqref{eq: representation of optimal transform}, we present a proof-of-concept example. Our benchmark criterion for assessing the efficacy of the alignment approach is its ability to eliminate differences between images arising from rotation and scaling. Let $g$ denote the grayscale function defined in Eq.~\eqref{eq: proof-of-concept scalar field example}. We will study the scaled and rotated version $\lambda\cdot\pmb{A}_{\theta,*}g$ of $g$, where $\pmb{A}_{\theta,*}$ denotes the dual of the rotation matrix
$\pmb{A}_\theta$ defined as
\begin{align}\label{eq: A_theta}
\pmb{A}_\theta:=
    \begin{pmatrix}
        \cos\theta & 0 & \sin\theta \\
        0 & 1 & 0 \\
        -\sin\theta & 0 & \cos\theta
    \end{pmatrix},
\end{align}
which represents rotation on the $(x_1,x_3)$-plane by angle $\theta$. Obviously, the differences between the source image $\lambda\cdot\pmb{A}_{\theta,*}g$ and the reference image $g$ vanish if and only if $\theta=0$ and $\lambda=1$ (i.e., no rotation or scaling). Hence, in this example, our proposed alignment approach is effective if 
\begin{align}\label{eq: alignment experiment, goal}
    (0,1)=\argmin_{(\theta,\lambda)} \left\{\left\Vert \,\operatorname{ERT}(g) - \lambda\cdot\pmb{A}_{\theta,*}\operatorname{ERT}(g) \,\right\Vert_{L_\nu^2 L_t^2}\right\}.
\end{align}
To validate Eq.~\eqref{eq: alignment experiment, goal}, we analyze the surface of the following function of $(\theta,\lambda)$
\begin{align}\label{eq: (theta, lambda) surface}
    (\theta,\lambda) \mapsto \left\Vert \,\operatorname{ERT}(g) - \lambda\cdot\pmb{A}_{\theta,*}\operatorname{ERT}(g) \,\right\Vert_{L_\nu^2 L_t^2},
\end{align}
which is presented in Figure \ref{fig: Alignment_distance}. The surfaces presented in Figure \ref{fig: Alignment_distance} confirm the minimization in Eq.~\eqref{eq: alignment experiment, goal} --- the minimum point of the surface corresponds to the coordinates $(\theta,\lambda)=(0,1)$, implying that our alignment approach delineated in Eq.~\eqref{eq: representation of optimal transform} is effective.

To illustrate the performance of the proposed alignment approach in our proof-of-concept example, we consider the \( g \) defined in Eq.~\eqref{eq: proof-of-concept scalar field example} as the reference image, as depicted in Figure~\ref{fig: alignment}(a). Next, \( \lambda \cdot \mathbf{A}_{\theta,*} g(x) \) with parameters \((\theta,\lambda) = (\pi/6,\, 1/5)\) is selected as the source image, as illustrated in Figure~\ref{fig: alignment}(b). The source image is a rotated and scaled version of the reference image. Aligning the source image with respect to the reference yields the post-aligned source image shown in Figure~\ref{fig: alignment}(c). The reference image in Figure~\ref{fig: alignment}(a) and post-aligned source image in Figure~\ref{fig: alignment}(c) are nearly congruent. This similarity underscores the efficacy of the proposed alignment methodology in mitigating discrepancies attributable to rotations and scaling.

\begin{figure}
    \centering
    \includegraphics[scale=0.4]{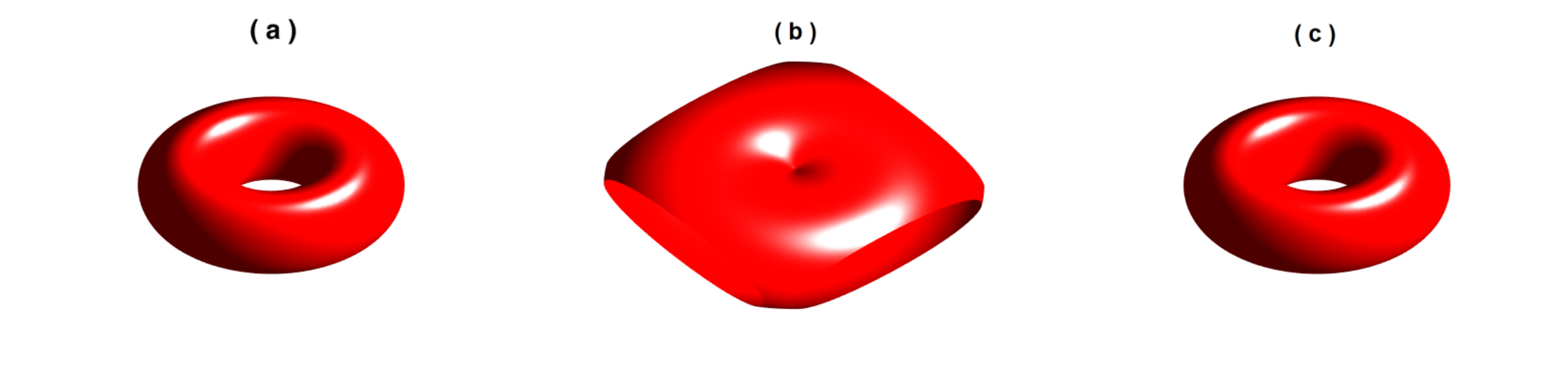}
    \caption{In panel~(a), the level set \( \{ x \in \mathbb{R}^3 : g(x) = 0.0834 \} \) is illustrated with \( g \) being defined as in Eq.~\eqref{eq: proof-of-concept scalar field example}. Panel~(b) depicts the level set \( \{ x \in \mathbb{R}^3 : \lambda \cdot \mathbf{A}_{\theta,*} g(x) = 0.0834 \} \) for parameters \( \lambda = 1/5 \) and \( \theta = \pi/6 \). Here, we consider \( g \) to be the reference image and \( \lambda \cdot \mathbf{A}_{\theta,*} g(x) \) with $(\theta,\lambda)=(\pi/6,\, 1/5)$ representing the source image that is under investigation. Panels (a) and (b) show that the reference and source images are drastically different. Employing the methodology from Eq.~\eqref{eq: representation of optimal transform}, we align \( \lambda \cdot \mathbf{A}_{\theta,*} g(x) \) which results in the aligned image denoted as \( \tilde{g} \). Panel~(c) showcases the level set \( \{ x \in \mathbb{R}^3 : \tilde{g}(x) = 0.0834 \} \), closely mirroring the level set displayed in panel~(a).}
    \label{fig: alignment}
\end{figure}

\section{Statistical Inference of Grayscale Functions}\label{section: Statistical Inference of Grayscale Functions}

We now provide our second major contribution --- approaches for statistical inference on grayscale images. The grayscale functions presented in the previous sections of this work have been viewed as deterministic. In this section, we now view grayscale functions as random where we assume that they are generated from underlying distributions satisfying some regularity conditions (see Assumptions \ref{assumption: finite second moments of SERT} and \ref{assumption: continuity of covariance function} in Section \ref{section: KL decomposition based SI}). Let $\Omega$ denote the collection of grayscale functions of interest, and assume that $\Omega$ is equipped with a $\sigma$-field $\mathscr{F}$. Next, suppose that there are two underlying grayscale function-generating distributions (probability measures), $\mathbb{P}^{(1)}$ and $\mathbb{P}^{(2)}$, defined on the sample space $(\Omega,\mathscr{F})$. Our data are two collections of random grayscale functions sampled from the two distributions: $\{g_i^{(1)}\}_{i=1}^{n}\overset{iid}{\sim}\mathbb{P}^{(1)}$ and $\{g_i^{(2)}\}_{i=1}^{n}\overset{iid}{\sim}\mathbb{P}^{(2)}$. Here, we provide approaches to testing if the two collections of functions are significantly different. More precisely, we propose methods of testing the following hypotheses
\begin{align}\label{eq: hypotheses to test for the permutation-based method}
    H_0^*: \ \ \mathbb{P}^{(1)}=\mathbb{P}^{(2)}\ \ \ \text{vs.}\ \ \ H_1^*: \ \ \mathbb{P}^{(1)}\ne\mathbb{P}^{(2)}.
\end{align}
Without loss of generality, hereafter, we assume that the grayscale images $\{g_i^{(1)}\}_{i=1}^{n}$ and $\{g_i^{(2)}\}_{i=1}^{n}$ have been aligned using the ERT-based alignment method proposed in Section \ref{section: Alignment of Images and the Invariance of ERT}.

\subsection{$\chi^2$-test via the Karhunen–Loève Expansion}\label{section: KL decomposition based SI}

In this subsection, we propose $\chi^2$-based hypothesis testing procedures via the Karhunen–Loève expansion and central limit theorem (CLT). These can be as generalizations of results presented in \cite{meng2022randomness} for the SECT and binary images. Testing the hypotheses in Eq.~\eqref{eq: hypotheses to test for the permutation-based method} is a highly nonparametric problem, and the $\chi^2$-test approaches transform it into a parametric problem.

Suppose the grayscale image $g$ is random and $g\sim\mathbb{P}^{(j)}$ for either $j=1$ or $j=2$. For each fixed $(\nu,t) \in \mathbb{S}^{d-1}\times[0,T]$, we have $\operatorname{SERT}(g)(\nu,t)$ as a real-valued random variable. Hence, for each fixed direction $\nu\in\mathbb{S}^{d-1}$, $\operatorname{SECT}(g)(\nu):=\{\operatorname{SERT}(g)(\nu,t)\}_{t\in[0,T]}$ is a stochastic process. Note that it is straightforward that the sample paths of the stochastic process are continuous (see Eq.~\eqref{eq: def of SERT}). In this section, we assume the following regarding $\mathbb{P}^{(1)}$ and $\mathbb{P}^{(2)}$.
\begin{assumption}\label{assumption: finite second moments of SERT}
    For each $j\in\{1,2\}$ and $(\nu,t)\in\mathbb{S}^{d-1}\times[0,T]$, we have the finite second moment $\mathbb{E}^{(j)}\left\vert\operatorname{SERT}(\nu,t) \right\vert^2 := \int_
    \Omega\vert\operatorname{SERT}(g)(\nu,t)\vert^2 \, \mathbb{P}^{(j)}(dg)<\infty$.
\end{assumption}
\noindent Under Assumption \ref{assumption: finite second moments of SERT}, we define the mean and covariance functions of $\operatorname{SECT}(g)(\nu)$ as follows
\begin{align*}
    & m_\nu^{(j)}(t):=\mathbb{E}^{(j)}\left\{\operatorname{SERT}(\nu,t) \right\}=\int_\Omega \operatorname{SERT}(g)(\nu,t) \, \mathbb{P}^{(j)}(dg), \\
    & \kappa_\nu^{(j)}(s,t):= \int_\Omega \Big(\operatorname{SERT}(g)(\nu,s) - m_\nu^{(j)}(s)\Big)\cdot \Big(\operatorname{SERT}(g)(\nu,t) - m_\nu^{(j)}(t)\Big) \,\mathbb{P}^{(j)}(dg),
\end{align*}
for $j\in\{1,2\}$, $\nu\in\mathbb{S}^{d-1}$, and $s,t\in[0,T]$, where $\mathbb{E}^{(j)}$ denotes the expectation associated with the probability measure $\mathbb{P}^{(j)}$. Furthermore, we need the following assumption on the covariance functions.
\begin{assumption}\label{assumption: continuity of covariance function}
    For each $j\in\{1,2\}$ and fixed $\nu\in\mathbb{S}^{d-1}$, the function $(s,t) \mapsto \kappa_\nu^{(j)}(s,t)$ is continuous on the product space $[0,T]\times[0,T]$.
\end{assumption}
\noindent Under Assumption \ref{assumption: continuity of covariance function}, the stochastic process $\operatorname{SECT}(g)(\nu)$ is mean-square continuous, which is a direct result of ``Lemma 4.2" of \cite{alexanderian2015brief}. The mean-square continuity implies that $t\mapsto m_\nu^{(j)}(t)$ is a continuous function over the compact interval $[0,T]$. 

Distinguishing the two collections of grayscale images, $\{g_i^{(1)}\}_{i=1}^{n}\overset{iid}{\sim}\mathbb{P}^{(1)}$ and $\{g_i^{(2)}\}_{i=1}^{n}\overset{iid}{\sim}\mathbb{P}^{(2)}$, is done by rejecting the null hypothesis $H_0^*$ in Eq.~\eqref{eq: hypotheses to test for the permutation-based method}. To reject the null $H_0^*$ in Eq.~\eqref{eq: hypotheses to test for the permutation-based method}, it suffices to reject the null hypothesis $H_0$ in the following test
\begin{align}\label{eq: hypotheses on means}
    \begin{aligned}
        & H_0: \, m_{\nu^*}^{(1)}(t) = m_{\nu^*}^{(2)}(t) \text{ for all }t\in[0,T] \ \ \ vs. \ \ \ H_1: \, m_{\nu^*}^{(1)}(t') = m_{\nu^*}^{(2)}(t') \text{ for some }t'\in[0,T], \\
    & \text{where }\ \ \nu^* := \argmax_{\nu\in\mathbb{S}^{d-1}} \left\{\sup_{t\in[0,T]}\left\vert m_{\nu^*}^{(1)}(t) - m_{\nu^*}^{(2)}(t) \right\vert\right\}.
    \end{aligned}
\end{align}
We need the following assumption to perform a $\chi^2$-test for the hypotheses in Eq.~\eqref{eq: hypotheses on means}.
\begin{assumption}\label{assumption: equal covariance assumption}
$\kappa_{\nu^*}^{(1)}=\kappa_{\nu^*}^{(2)}$ where the direction $\nu^*$ is defined in Eq.~\eqref{eq: hypotheses on means}.
\end{assumption}
\noindent Assumption \ref{assumption: equal covariance assumption} is true under the null $H_0^*: \mathbb{P}^{(1)} = \mathbb{P}^{(2)}$ in Eq.~\eqref{eq: hypotheses to test for the permutation-based method}. Under Assumption \ref{assumption: equal covariance assumption}, we denote $\kappa:=\kappa_{\nu^*}^{(1)}=\kappa_{\nu^*}^{(2)}$. Under Assumption \ref{assumption: continuity of covariance function}, we have $\kappa\in L^2([0,T]\times[0,T])$ which further implies that the integral operator $f\mapsto \int_0^T f(s)\cdot\kappa(s,\cdot)\, ds$ defined on $L^2(0,T)$ is a compact, positive, and self-adjoint (see ``Lemma 5.1" of \cite{alexanderian2015brief}). The Hilbert-Schmidt theorem (see ``Theorem VI.16" of \cite{reed2012methods}) indicates that this integral operator has countably many orthonormal eigenfunctions  $\{\phi_l\}_{l=1}^\infty$ and nonnegative eigenvalues $\{\lambda_l\}_{l=1}^\infty$. Without loss of generality, we assume $\lambda_1\ge\lambda_2\ge\ldots\ge0$. Following the proof of the ``Karhunen-Loève expansion" in \cite{meng2022randomness}, one can show the following result.
\begin{theorem}\label{thm: Karhunen–Loève expansion}
Suppose $g^{(1)}\sim\mathbb{P}^{(1)}$ and $g^{(2)}\sim\mathbb{P}^{(2)}$ are independent. Let $\mathbb{P}^{(1)}\otimes\mathbb{P}^{(2)}$ denote a product probability measure. For each fixed $l\in\mathbb{N}$, the following identity holds with probability one
    \begin{align}\label{eq: Karhunen–Loève expansion}
    \frac{1}{\sqrt{2\lambda_l}}\int_0^T \left\{\, \operatorname{SERT}(g^{(1)})(\nu^*,t) - \operatorname{SERT}(g^{(2)})(\nu^*,t) \,\right\}\cdot\phi_l(t) \,dt = \theta_l + \left(\frac{Z_l^{(1)}(g^{(1)})-Z_l^{(2)}(g^{(1)})}{\sqrt{2}}\right),
\end{align}
that is, $\mathbb{P}^{(1)}\otimes\mathbb{P}^{(2)}\left\{\text{Eq.~\eqref{eq: Karhunen–Loève expansion} holds}\right\}=1$, where,
\begin{align}\label{eq: def of theta_l and Z}
    \begin{aligned}
        & \theta_l:= \frac{1}{\sqrt{2\lambda_l}} \int_0^T \left\{\, m_{\nu^*}^{(1)}(t) - m_{\nu^*}^{(2)}(t) \,\right\}\cdot\phi_l(t) \,dt, \\
    & Z_l^{(j)}(g)=\frac{1}{\sqrt{\lambda_l}} \int_0^T \left\{\, \operatorname{SECT}(g)(\nu^*,t)-m_{\nu^*}^{(j)}(t) \,\right\}\cdot\phi_l(t) \,dt, \ \ \ \text{ for }j\in\{1,2\}.
    \end{aligned}
\end{align}
Furthermore, for each $j\in\{1,2\}$, random variables $\{Z_l^{(j)}\}_{l=1}^\infty$ are defined on the probability space $(\Omega,\mathscr{F},\mathbb{P}^{(j)})$, mutually uncorrelated, and have mean 0 and variance 1. 
\end{theorem}

Following the discussion in \cite{meng2022randomness}, one can show that the null hypothesis $H_0$ in Eq.~\eqref{eq: hypotheses on means} is equivalent to $\theta_l=0$ for all $l=1,2,3,\ldots$. It is infeasible to check $\theta_l$ for all positive integers $l$. In addition, a small $\lambda_l$ in the denominator (e.g., see Eq.~\eqref{eq: def of theta_l and Z}) induces numerical instability. Therefore, we only consider $\theta_l$ for $l=1,2,\ldots,L$, where $L$ is given as the following
\begin{align}\label{eq: def of L}
    L := \min\left\{ \, k\in\mathbb{N} \,:\, \frac{\sum_{k'=1}^k \lambda_{k'}}{\sum_{k''=1}^\infty \lambda_{k''}} > 0.99\right\}.
\end{align}
Here, 0.99 can be replaced with any value in $(0,\,1)$; we take 0.99 as an example. The $L$ defined in Eq.~\eqref{eq: def of L} is motivated by principal component analysis \citep{jolliffe2002principal} and indicates that we maintain at least $99\%$ of the cumulative variance in the data. Hence, to test the hypotheses in Eq.~\eqref{eq: hypotheses on means}, we may test the following approximate hypotheses
\begin{align}\label{eq: approximate hypotheses}
    \widehat{H_0}: \, \theta_0=\theta_1=\cdots=\theta_L=0 \ \ \ vs. \ \ \ \widehat{H_1}: \, \text{there exists }k'\in\{1,2,\ldots,L\} \text{ such that }\theta_{k'} \ne 0.
\end{align}
Given data $\{g^{(1)}_i\}_{i=1}^n \overset{iid}{\sim} \mathbb{P}^{(1)}$ and $\{g^{(2)}_i\}_{i=1}^n \overset{iid}{\sim} \mathbb{P}^{(2)}$, the approximate hypotheses in Eq.~\eqref{eq: approximate hypotheses} can be tested using the random variables $\{\xi_{l,i}:\,l=1,\ldots,L \text{ and } i=1,\ldots,n\}$ defined as follows
\begin{align*}
    \xi_{l,i} &:= \frac{1}{\sqrt{2\lambda_l}}\int_0^T \left\{\, \operatorname{SERT}(g_i^{(1)})(\nu^*,t) - \operatorname{SERT}(g_i^{(2)})(\nu^*,t) \,\right\}\cdot\phi_l(t) \,dt = \theta_l + \left(\frac{Z_l^{(1)}(g_i^{(1)})-Z_l^{(2)}(g_i^{(1)})}{\sqrt{2}}\right).
\end{align*}
Theorem \ref{eq: def of L} implies that the random variables $\xi_{l,i}$ satisfy the following properties:
\begin{itemize}
    \item For each $l\in\{1,\ldots,L\}$ and $i\in\{1,\ldots,n\}$, the random variable $\xi_{l,i}$ has mean $\theta_l$ and variance 1.
    \item For each fixed $i\in\{1,\ldots,n\}$, the random variables $\xi_{1,i},\ldots, \xi_{L,i}$ are mutually uncorrelated.
    \item For each fixed $l\in\{1,\ldots,L\}$, the random variables $\xi_{l,1},\ldots, \xi_{l,n}$ are iid.
\end{itemize}
The properties above indicate:
\begin{enumerate}
    \item For each fixed $l\in\{1,\ldots,L\}$, the standardized $\frac{1}{\sqrt{n}}\sum_{i=1}^n \xi_{l,i}$ asymptotically follows a standard Gaussian distribution $N(0,1)$ under the null hypothesis $\widehat{H_0}$ in Eq.~\eqref{eq: approximate hypotheses}.
    \item The asymptotic normality implies that random variables $\frac{1}{\sqrt{n}}\sum_{i=1}^n \xi_{1,i},\ldots, \frac{1}{\sqrt{n}}\sum_{i=1}^n \xi_{L,i}$ are asymptotically independent.
\end{enumerate}
Hence, $\sum_{l=1}^L\left(\frac{1}{\sqrt{n}}\sum_{i=1}^n \xi_{l,i}\right)^2$ is asymptotically $\chi^2_L$ under the null hypothesis $\widehat{H_0}$ in Eq.~\eqref{eq: approximate hypotheses}, and we reject $\widehat{H_0}$ with asymptotic significance $\alpha\in(0,1)$ if
\begin{align}\label{eq: chi-square test}
    \sum_{l=1}^L\left(\frac{1}{\sqrt{n}}\sum_{i=1}^n \xi_{l,i}\right)^2 > \chi_{L,1-\alpha}^2 =  \text{ the $1-\alpha$ lower quantile of the $\chi^2_L$ distribution}.
\end{align}
\noindent Overall, we summarize our hypothesis testing problem as follows. Our goal is to test the hypotheses in Eq.~\eqref{eq: hypotheses to test for the permutation-based method}. Through the Karhunen–Loève expansion (see Theorem \ref{thm: Karhunen–Loève expansion}), it suffices to test the approximate hypotheses in Eq.~\eqref{eq: approximate hypotheses}, which can be achieved by the $\chi^2$-test in Eq.~\eqref{eq: chi-square test}.

Suppose we have two groups of grayscale functions, $\{g^{(1)}_i\}_{i=1}^n \overset{iid}{\sim} \mathbb{P}^{(1)}$ and $\{g^{(2)}_i\}_{i=1}^n \overset{iid}{\sim} \mathbb{P}^{(2)}$. We calculate the discretized SERT of the grayscale images, denoted as $\mathcal{D}^{(j)}:=\{\operatorname{SERT}(g_i^{(j)})(\nu_p,t_q):\,p=1,\ldots,\Gamma \text{ and }q=1,\ldots,\Delta\}_{i=1}^n$ for $j\in\{1,2\}$. Then, we apply $\mathcal{D}^{(1)}$ and $\mathcal{D}^{(2)}$ as our input data to test the hypotheses in Eq.~\eqref{eq: hypotheses on means}, which are approximated by that in Eq.~\eqref{eq: approximate hypotheses}. We may apply ``Algorithm 1" in \cite{meng2022randomness} to test the hypotheses in Eq.~\eqref{eq: approximate hypotheses} by simply replacing the ``SECT of two collections of shapes" therein with the ``SERT of two collections of grayscale functions''. The replacing of the SECT with SERT approach is explicitly summarized in Algorithm \ref{algorithm: chi-sq test}.
\begin{algorithm}[H]
	\caption{: $\chi^2$-test}\label{algorithm: chi-sq test}
	\begin{algorithmic}[1]
		\INPUT
        \noindent (i) Two collections $\{g_i^{(1)}\}_{i=1}^{n}$ and $\{g_i^{(2)}\}_{i=1}^{n}$ of grayscale functions; (ii) desired asymptotic confidence level $1-\alpha$ with asymptotic significance $\alpha\in(0,1)$.
		\OUTPUT \texttt{Accept} or \texttt{Reject} the null hypothesis $\widehat{H_0}$ in Eq.~\eqref{eq: approximate hypotheses}. (Rejecting the $\widehat{H_0}$ implies rejecting the null hypothesis $H_0^*$ in Eq.~\eqref{eq: hypotheses to test for the permutation-based method}.)
		\State Compute the discretized SERT $\{\operatorname{SERT}(g_i^{(j)})(\nu_p,t_q):\,p=1,\ldots,\Gamma \text{ and }q=1,\ldots,\Delta\}_{i=1}^n$ for $j\in\{1,2\}$ of the input grayscale functions.
  \State Replace the input ``SECT" in ``Algorithm 1" of \cite{meng2022randomness} with the SERT computed in the previous step.
  \State Implement ``Algorithm 1" of \cite{meng2022randomness} and get the output.
		\end{algorithmic}
\end{algorithm}

Although the null hypothesis in Eq.~\eqref{eq: hypotheses to test for the permutation-based method} theoretically implies Assumption \ref{assumption: equal covariance assumption}, the finite sample size $n<\infty$ may numerically violate Assumption \ref{assumption: equal covariance assumption} due to the inaccuracy in the estimation of covariance functions. The (numerical) violation of Assumption \ref{assumption: equal covariance assumption} tends to lead to type-I error inflation. To reduce the type-I error rate,  we apply a permutation technique \citep{good2013permutation}. That is, we first apply Algorithm \ref{algorithm: chi-sq test} to our original grayscale images $g_i^{(j)}$ and then repeatedly re-apply Algorithm \ref{algorithm: chi-sq test} to the grayscale images with shuffled group labels $j$. Next, we compare how the $\chi^2$ statistic derived from the original data (see Eq.~\eqref{eq: chi-square test}) differs from that computed on the shuffled data. The idea behind the permutation approach is that shuffling the group labels $j$ of images $g_i^{(j)}$ should not significantly change the test statistic under the null hypothesis. The combination of the permutation technique and Algorithm \ref{algorithm: chi-sq test} is an analog of the permutation test proposed in \cite{meng2022randomness}. We summarize this method in Algorithm \ref{algorithm: permutation-based chi-sq test}. Among the algorithms we propose throughout, we particularly recommend using Algorithm \ref{algorithm: permutation-based chi-sq test} in practice --- this is also supported by simulation study results presented in Section \ref{section: Numerical Experiments}. Specifically, we will show that Algorithm \ref{algorithm: permutation-based chi-sq test} is uniformly powerful under the alternative hypotheses and does not suffer from type I error inflation.
\begin{algorithm}[H]
	\caption{: Permutation-based $\chi^2$-test}\label{algorithm: permutation-based chi-sq test}
	\begin{algorithmic}[1]
		\INPUT
        \noindent (i) Two collections $\{g_i^{(1)}\}_{i=1}^{n}$ and $\{g_i^{(2)}\}_{i=1}^{n}$ of grayscale functions; (ii) desired asymptotic confidence level $1-\alpha$ with asymptotic significance $\alpha\in(0,1)$; (iii) the number of permutations $\Pi$.
		\OUTPUT \texttt{Accept} or \texttt{Reject} the null hypothesis $\widehat{H_0}$ in Eq.~\eqref{eq: approximate hypotheses}. (Rejecting the $\widehat{H_0}$ implies rejecting the null hypothesis $H_0^*$ in Eq.~\eqref{eq: hypotheses to test for the permutation-based method}.)
		\State Compute the discretized SERT $\{\operatorname{SERT}(g_i^{(j)})(\nu_p,t_q):\,p=1,\ldots,\gamma \text{ and }q=1,\ldots,\Delta\}_{i=1}^n$ for $j\in\{1,2\}$ of the input grayscale functions.
  \State Replace the input ``SECT" in ``Algorithm 2" of \cite{meng2022randomness} with the SERT computed in the previous step.
  \State Implement ``Algorithm 2" of \cite{meng2022randomness} and get the output.
		\end{algorithmic}
\end{algorithm}

\subsection{Full Permutation Test}

In Section \ref{section: KL decomposition based SI}, we proposed two parametric-based approaches --- Algorithms \ref{algorithm: chi-sq test} and \ref{algorithm: permutation-based chi-sq test} --- to testing the hypotheses in Eq.~\eqref{eq: hypotheses to test for the permutation-based method}. Although Algorithm \ref{algorithm: permutation-based chi-sq test} involves a permutation-like technique, it still heavily depends on the $\chi^2$-test in Eq.~\eqref{eq: chi-square test}. In contrast, we also propose a full permutation hypothesis test. We will also compare this approach with Algorithms \ref{algorithm: chi-sq test} and \ref{algorithm: permutation-based chi-sq test} using simulations in Section \ref{section: Numerical Experiments}. The simulation studies therein indicate that our proposed Algorithms \ref{algorithm: chi-sq test} and \ref{algorithm: permutation-based chi-sq test} tend to be more powerful than the fully permutation-based test. Following a strategy proposed by \cite{robinson2017hypothesis}, we apply the full permutation test based on the following loss function
\begin{align}\label{eq: loss function for the permutation test}
    L\left(\{g_i^{(1)}\}_{i=1}^{n},\, \{g_i^{(2)}\}_{i=1}^{n}\right) := \frac{1}{2n(n-1)} \sum_{k,l=1}^{n} \left\{\, \operatorname{dist}\left(\, g_k^{(1)},\, g_l^{(1)}\,\right) + \operatorname{dist}\left(\, g_k^{(2)},\, g_l^{(2)}\,\right) \,\right\},
\end{align}
where $\operatorname{dist}\in\{  \operatorname{dist}^{\operatorname{ERT}}_{p,q}, \operatorname{dist}^{\operatorname{SERT}}_{p,q}, \operatorname{dist}^{\operatorname{SELECT}}_{p}, \operatorname{dist}^{\operatorname{MEC}}_{p}\}$ (see Eq.~\eqref{eq: semi distances between grayscale functions} and Eq.~\eqref{eq: metrics defined in kirveslahti2023representing}). The full permutation test based on Eq.~\eqref{eq: loss function for the permutation test} is summarized in Algorithm \ref{algorithm: permutation-based test}.

    

\begin{algorithm}[h]
	\caption{: Full Permutation Test}\label{algorithm: permutation-based test}
	\begin{algorithmic}[1]
		\INPUT
        \noindent (i) Two collections $\{g_i^{(1)}\}_{i=1}^{n}$ and $\{g_i^{(2)}\}_{i=1}^{n}$ of grayscale functions; (ii) desired asymptotic confidence level $1-\alpha$ with $\alpha\in(0,1)$; (iii) the number $\Pi$ of permutations; (iv) distance function $\operatorname{dist}\in\{  \operatorname{dist}^{\operatorname{ERT}}_{p,q}, \operatorname{dist}^{\operatorname{SERT}}_{p,q}, \operatorname{dist}^{\operatorname{SELECT}}_{p}, \operatorname{dist}^{\operatorname{MEC}}_{p}\}$ with prespecified parameters $p$ and $q$.
		\OUTPUT \texttt{Accept} or \texttt{Reject} the null hypothesis $H_0^*$ in Eq.~\eqref{eq: hypotheses to test for the permutation-based method}.
		\State Apply Eq.~\eqref{eq: loss function for the permutation test} to the original input grayscale functions and compute the value of the loss $$\mathfrak{S}_0 := L\left(\{g_i^{(1)}\}_{i=1}^{n},\, \{g_i^{(2)}\}_{i=1}^{n}\right).$$
\FORALL{$k=1,\cdots,\Pi$, }
\State Randomly permute the group labels $j\in\{1,2\}$ of the input grayscale functions where the permuted grayscale functions are denoted as $\{\Tilde{g}_i^{(1)}\}_{i=1}^{n}$ and $\{\Tilde{g}_i^{(2)}\}_{i=1}^{n}$.
\State Apply Eq.~\eqref{eq: loss function for the permutation test} to the permuted grayscale functions and compute the value of the loss $$\mathfrak{S}_k := L\left(\{\Tilde{g}_i^{(1)}\}_{i=1}^{n},\, \{\Tilde{g}_i^{(2)}\}_{i=1}^{n}\right).$$
\ENDFOR
\State Compute $k^* :=\lfloor\alpha\cdot\Pi\rfloor:=$ the largest integer smaller than $\alpha\cdot\Pi$. 
\State \texttt{Reject} the null hypothesis $H_0$ if $\mathfrak{S}_0<\mathfrak{S}_{k^*}$ and report the output.
		\end{algorithmic}
\end{algorithm}

\begin{figure}[h]
    \centering
    \includegraphics[scale=0.22]{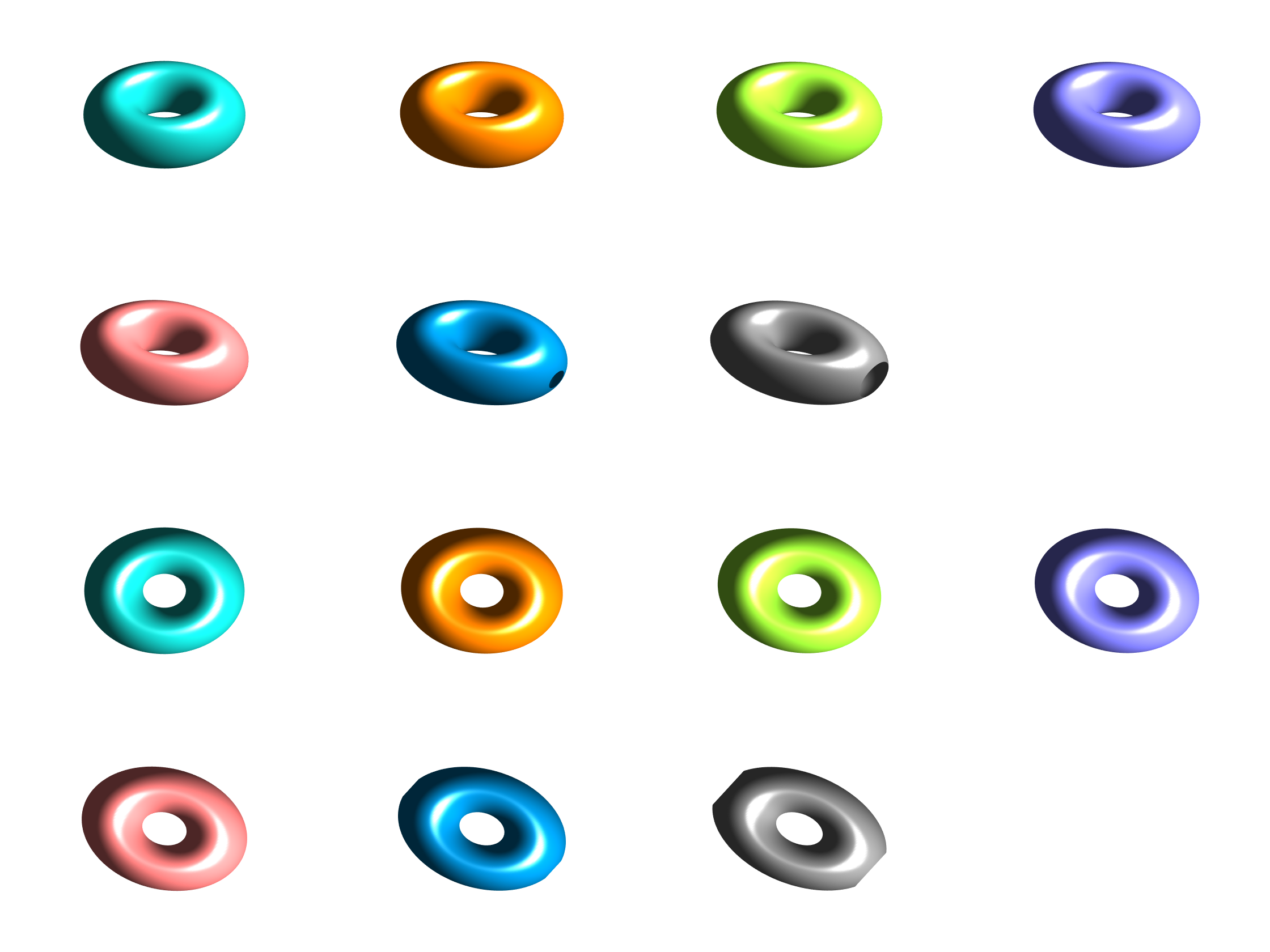}
    \caption{The first two rows and the last two rows present the same collection of seven surfaces from different angles. The seven surfaces in each collection represent the level set $\{x\in\mathbb{R}^3:\, h^{(\epsilon)}(x)=0.0834\}$ for seven indices $\epsilon\in\{0.7, \, 0.8, \, 0.85, \, 0.875, \, 0.9, \, 0.95, \, 1 \}$, respectively.}
    \label{fig: donuts}
\end{figure}

\begin{figure}[h]
    \centering
    \includegraphics[scale=0.8]{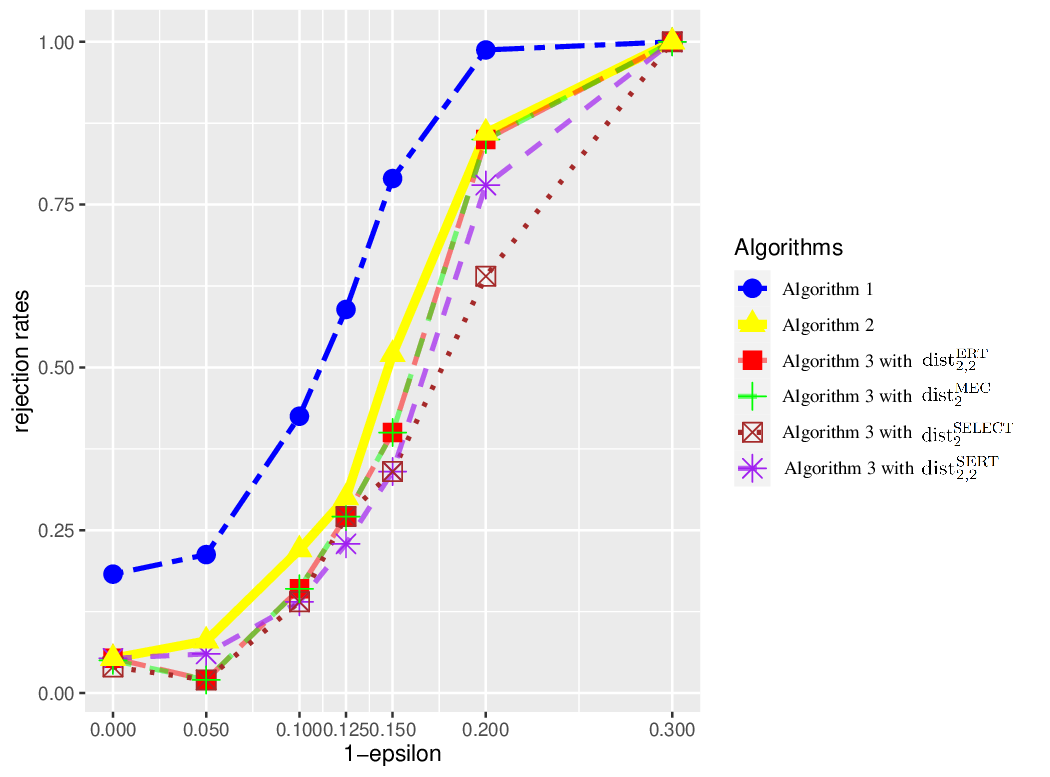}
    \caption{Rejection rates of Algorithms \ref{algorithm: chi-sq test}, \ref{algorithm: permutation-based chi-sq test}, and \ref{algorithm: permutation-based test} across different indices $\epsilon$. The verticle axis indicates the rejection rates across the 100 data replicates and the horizontal axis indicates $1-\epsilon$. In the figure, the label ``Algorithm \ref{algorithm: permutation-based test} with \(\operatorname{dist}_{2,2}^{\operatorname{ERT}}\)" refers to the implementation of Algorithm \ref{algorithm: permutation-based test} where the input is \(\operatorname{dist}=\operatorname{dist}_{2,2}^{\operatorname{ERT}}\). Analogously, labels with \(\operatorname{dist}_{2,2}^{\operatorname{SERT}}\), \(\operatorname{dist}_{2}^{\operatorname{SELECT}}\), and \(\operatorname{dist}_{2}^{\operatorname{MEC}}\) function in the same manner.}
    \label{fig: rates}
\end{figure}

\section{Numerical Experiments}\label{section: Numerical Experiments}

In this section, we show the performance of our proposed Algorithms \ref{algorithm: chi-sq test}, \ref{algorithm: permutation-based chi-sq test}, and \ref{algorithm: permutation-based test} using simulations. Specifically, we generate grayscale functions from a family of random fields and apply our proposed algorithms to them. Motivated by the simulation designs in \cite{meng2022randomness} and \cite{kirveslahti2023representing}, we apply Algorithms \ref{algorithm: chi-sq test}-\ref{algorithm: permutation-based test} to the following family of random grayscale functions
\begin{align*}
    h^{(\epsilon)}(x_1, x_2, x_3):= \left\{\left(\sqrt{\frac{\alpha}{\epsilon}\cdot x_1^2 + \epsilon\cdot\beta\cdot x_2^2}-\delta\right)^2 + \gamma\cdot x_3^2\right\}\cdot\mathbbm{1}_{\{-1\le x_1, x_2, x_3\le 1\}}, \ \ \ \text{ for }\epsilon\in[0.7, \, 1].
\end{align*}
where $\epsilon$ is a deterministic index, $\alpha, \beta, \gamma$ are iid $\operatorname{Unif}(0.5, 1)$ random variables, $\delta\sim \operatorname{Unif}(0.4, 0.6)$, and all the random variables are independent. Let $\mathbb{P}^{(\epsilon)}$ denote the underlying distribution corresponding to $h^{(\epsilon)}$. All the realizations of $h^{(\epsilon)}$ belong to $\mathfrak{D}_{R,d}$ with $R=2$ and $d=3$. The level sets $\{x\in\mathbb{R}^3:\, h^{(\epsilon)}(x)=0.0834\}$ for different indices $\epsilon\in\{0.7, \, 0.8, \, 0.85, \, 0.875, \, 0.9, \, 0.95, \, 1 \}$ are presented in Figure \ref{fig: donuts}.

\setlength{\extrarowheight}{2pt}
\begin{table}[h]
\centering
\caption{Rejection rates (RRs) of Algorithms \ref{algorithm: chi-sq test}, \ref{algorithm: permutation-based chi-sq test}, and \ref{algorithm: permutation-based test} across different indices $\varepsilon$. In the table, the label ``Algorithm \ref{algorithm: permutation-based test} with \(\operatorname{dist}_{2,2}^{\operatorname{ERT}}\)" refers to the implementation of Algorithm \ref{algorithm: permutation-based test} where the input is \(\operatorname{dist}=\operatorname{dist}_{2,2}^{\operatorname{ERT}}\). Analogously, labels with \(\operatorname{dist}_{2,2}^{\operatorname{SERT}}\), \(\operatorname{dist}_{2}^{\operatorname{SELECT}}\), and \(\operatorname{dist}_{2}^{\operatorname{MEC}}\) function in the same manner.}
    \label{table: epsilon vs. rejection rates}
    \vspace*{0.5em}
\begin{tabular}{|c||c||c|c|c|c|c|c|}
\hline
& \textbf{Null} $\boldsymbol{H_0}$ & \multicolumn{6}{c|}{\textbf{Alternative} $\boldsymbol{H_1}$} \\ [2pt]
\hline
\textbf{Indices $\boldsymbol{(1-\varepsilon)}$}  & \textbf{0.000}  & \textbf{0.050}  & \textbf{0.100} & \textbf{0.125} & \textbf{0.150} & \textbf{0.200} & \textbf{0.300}  \\ [2pt]\hline
RRs of Algorithm \ref{algorithm: chi-sq test} & 0.18 & 0.21 & 0.43 & 0.59   & 0.79 & 0.99 & 1.00 \\ [2pt]\hline
RRs of Algorithm \ref{algorithm: permutation-based chi-sq test} & 0.05 & 0.08 & 0.22 & 0.30  & 0.52 & 0.86 & 1.00 \\ [2pt]\hline
RRs of Algorithm \ref{algorithm: permutation-based test} with $\operatorname{dist}^{\operatorname{ERT}}_{2,2}$ & 0.05 & 0.02 & 0.16 & 0.27  & 0.40 & 0.86 & 1.00\\ [2pt]\hline
RRs of Algorithm \ref{algorithm: permutation-based test}  with $\operatorname{dist}^{\operatorname{SERT}}_{2,2}$  & 0.05 & 0.06 & 0.14 & 0.23  & 0.34 & 0.78 & 1.00 \\ [2pt]\hline
RRs of Algorithm \ref{algorithm: permutation-based test} with $\operatorname{dist}^{\operatorname{SELECT}}_{2}$& 0.04 & 0.02 & 0.14 & 0.27  & 0.34 & 0.64 & 1.00 \\ [2pt]\hline
RRs of Algorithm \ref{algorithm: permutation-based test} with $\operatorname{dist}^{\operatorname{MEC}}_{2}$& 0.05 & 0.02 & 0.16 & 0.27  & 0.40 & 0.86 & 1.00
 \\ [2pt]\hline
\end{tabular}
\end{table}

We apply our proposed algorithms to test the following hypotheses
\begin{align}\label{eq: hypotheses for the family indexed by epsilon}
    H_0: \ \ \mathbb{P}{^{(1)}} = \mathbb{P}{^{(\epsilon)}}\ \ vs. \ \ H_1: \ \ \mathbb{P}{^{(1)}} \ne \mathbb{P}{^{(\epsilon)}}.
\end{align}
The null hypothesis $H_0$ in Eq.~\eqref{eq: hypotheses for the family indexed by epsilon} is true if and only if $\epsilon=1$. We generate $n=30$ realizations of $h^{(1)}$. Then, for each $\epsilon\in[0.7, 1]$, we generate $n=30$ realizations of $h^{(\epsilon)}$. We apply Algorithms \ref{algorithm: chi-sq test}, \ref{algorithm: permutation-based chi-sq test}, and \ref{algorithm: permutation-based test} to the two collections of generated grayscale functions to test the hypotheses in Eq.~\eqref{eq: hypotheses for the family indexed by epsilon} with significance $0.05$ (i.e., the expected type I error rate is 0.05). We repeat this procedure 100 times, go through values of $\epsilon\in[0.7, 1]$, and present the rejection rates across the 100 repetitions in Table \ref{table: epsilon vs. rejection rates} and Figure \ref{fig: rates}. The numerical experiment results can be summarized as follows:
\begin{enumerate}
    \item Among all the algorithms, Algorithm \ref{algorithm: chi-sq test} is the most powerful under the alternative hypothesis where $\epsilon\ne1$. However, it suffers from type I error inflation --- meaning that the expected rejection rate when $\epsilon=1$ is 0.05 but the rejection rate of Algorithm \ref{algorithm: chi-sq test} is higher. As previously mentioned, the type I error inflation of Algorithm \ref{algorithm: chi-sq test} stems from the numerical violation of Assumption \ref{assumption: equal covariance assumption}.

    \item Algorithm \ref{algorithm: permutation-based chi-sq test}, which is a combination of the permutation technique and Algorithm \ref{algorithm: chi-sq test}, does not suffer from type I error inflation. While the power of Algorithm \ref{algorithm: permutation-based chi-sq test} is lower than that of Algorithm \ref{algorithm: chi-sq test} under the alternative hypothesis, it is still uniformly more powerful than the full permutation test in Algorithm \ref{algorithm: permutation-based test} with all the four distance inputs $\{  \operatorname{dist}^{\operatorname{ERT}}_{2,2}, \operatorname{dist}^{\operatorname{SERT}}_{2,2}, \operatorname{dist}^{\operatorname{SELECT}}_{2}, \operatorname{dist}^{\operatorname{MEC}}_{2}\}$.

    \item The four distance inputs for Algorithm \ref{algorithm: permutation-based test} result in comparable hypothesis testing performance. Particularly, the $\operatorname{dist}^{\operatorname{ERT}}_{2,2}$ and $\operatorname{dist}^{\operatorname{MEC}}_{2}$ result in the same performances when we apply them to Algorithm \ref{algorithm: permutation-based test}, which results from the similarity between the ERT and MEC. One theoretical advantage of the ERT over the MEC is that the ERT is homogeneous (i.e., $\operatorname{ERT}(\lambda\cdot g) = \lambda\cdot\operatorname{ERT}(g)$ for all $\lambda\in\mathbb{R}$).
\end{enumerate}
The results described above are similar to those shown in \cite{meng2022randomness} for numerical experiments conducted on random binary images. Based on the experiment results concluded above, we recommend Algorithm \ref{algorithm: permutation-based chi-sq test} in applications for the following reason: it is uniformly powerful (compared with the fully permutation-based Algorithm \ref{algorithm: permutation-based test}) and does not suffer from type I error inflation.

\section{Discussion}\label{section: Conclusions and Future Research}

The ultimate goal of our study is to generalize a series of ECT-based methods \citep{crawford2020predicting, wang2021statistical, meng2022randomness, marsh2022detecting} to the analysis of grayscale images. In this paper, we took an initial step towards this goal by proposing an ECT-like topological summary, the ERT. The framework proposed in \cite{baryshnikov2010euler} provides solid mathematical foundations for our proposed ERT. Building upon the ERT, we introduced the SERT as a generalization of the SECT \citep{crawford2020predicting, meng2022randomness}. Importantly, the SERT represents grayscale images as functional data. By applying the Karhunen–Loève expansion to the SERT, we have proposed effective statistical algorithms (see Algorithms \ref{algorithm: chi-sq test}-\ref{algorithm: permutation-based test}) designed to detect significant differences between two sets of grayscale images. Particularly, Algorithm \ref{algorithm: permutation-based chi-sq test} was shown in simulations to be uniformly powerful while not suffering from type I error inflation.

There are many motivating questions for future research. A few of them from the biomedical perspective include:
\begin{enumerate}
\item Significantly different images usually correspond to different clinical outcomes (e.g., survival rates). \cite{crawford2020predicting} used the SECT on binary images of GBM tumors as the predictors in statistical inference. Here, the authors showed that the SECT has the power to predict clinical outcomes better than existing tumor quantification approaches. A natural generalization of the approach in \cite{crawford2020predicting} is the development of an SECT-like statistic designed for grayscale images which could prove to be powerful in terms of predicting clinical outcomes. For instance, one may consider analyzing the grayscale images in Figure \ref{fig: Grayscale_image_from_Prof_Duan} to predict the clinical outcomes of the corresponding lung cancer patients.

\item Suppose grayscale images can successfully predict a clinical outcome of interest. In that case, a subsequent question from the sub-image analysis viewpoint is: \textit{can we identify the physical features in the grayscales image that are most relevant to the clinical outcome?} For binary images (equivalently, shapes), \cite{wang2021statistical} proposed an efficient method of seeking the desired physical features of shapes via the ECT. One may consider generalizing the method in \cite{wang2021statistical} to deal with grayscale images.

\item Tumors change over time. Hence, having the ability to study dynamically changing/longitudinal grayscale images is an area of interest. Using the ECT and SECT, \cite{marsh2022detecting} introduced the DETECT framework to analyze the dynamic changes in shapes. One may consider generalizing the DETECT approach to analyze the dynamic changes in grayscale images.
\end{enumerate}
Lastly, another future direction would be to employ statistical methods analogous to those described in Section \ref{section: Statistical Inference of Grayscale Functions} for the analysis of networks using curvature-based approaches \citep{wu2022subsampling}.

\section*{Software Availability}

Code for implementing the Euler-Radon transform (ERT), the smooth Euler-Radon transform (SERT), as well as the lifted Euler characteristic transform (LECT) and the super lifted Euler characteristic transform (SELECT) is freely available at \url{https://github.com/JinyuWang123/ERT}.

\section*{Acknowledgements}

LC would like to acknowledge the support of a David \& Lucile Packard Fellowship for Science and Engineering. Research reported in this publication was partially supported by the National Institute On Aging of the National Institutes of Health under Award Number R01AG075511. The content is solely the responsibility of the authors and does not necessarily represent the official views of the National Institutes of Health.

\section*{Statements and Declarations}

The authors declare no competing interests.


\newpage
\begin{appendix}

\section{Proofs}\label{Appendix: Proofs}

\subsection{Proof of Eq.~\eqref{eq: dX = [dX] for integer-valued functions}}\label{section: Proof of eq: dX = [dX] for integer-valued functions}

\begin{proof}
    For any $g\in\operatorname{CF}(X)$ and $n \in \Z$, we have $\lceil n\cdot g \rceil = n\cdot g = \lfloor n\cdot g \rfloor$. Hence,
    \begin{align*}
        \int_X g(x) \lceil \,d\chi(x) \rceil & = \lim_{n \to \infty} \frac{1}{n} \int_X \lceil n\cdot g(x) \rceil \,d\chi(x) \\
        & = \lim_{n \to \infty} \frac{1}{n} \int_X n\cdot g(x) \,d\chi(x) \\
        & = \lim_{n \to \infty} \int_X g(x) \,d\chi(x) \\
        & = \int_X g(x) \,d\chi(x).
    \end{align*}
Similarly, we have that $\int_X g(x) \lfloor d\chi(x) \rfloor = \int_X g(x) d\chi(x)$. Therefore, we have $\int_X g(x) [ d\chi(x) ] = \int_X g(x) d\chi(x)$. Since $\mathbbm{1}_K\in\operatorname{CF}(X)$ for all $K\in\operatorname{CS}(X)$, we obtain Eq.~\eqref{eq: dX = [dX] for integer-valued functions}.
\end{proof}

\subsection{Proof of Theorem \ref{thm: tameness theorem}}\label{proof of thm: tameness theorem}

We need the following lemma in \cite{van1998tame}.
\begin{lemma}[rephrased ``(2.10) Proposition," Chapter 4 of \cite{van1998tame}]\label{lemma: (2.10) Proposition of Dries}
    Let $S\subseteq\mathbb{R}^{m+d}$ be definable and $S_a:=\{x\in\mathbb{R}^d:\, (a,x)\in S\}$ for each $a \in \mathbb{R}^m$. Then, $\chi(S_a)$ takes only finitely many values as $a$ runs through $\mathbb{R}^m$. Furthermore, for each integer $z$, the set $\{a\in\mathbb{R}^m:\,\chi(S_a)=z\}$ is definable.
\end{lemma}

With Lemma \ref{lemma: (2.10) Proposition of Dries}, we provide the proof of Theorem \ref{thm: tameness theorem} as follows
\begin{proof}
    (Proof of Theorem \ref{thm: tameness theorem}.) We implement Lemma \ref{lemma: (2.10) Proposition of Dries} by defining the following 
    \begin{enumerate}
        \item $m := d+1$;
        \item $a=(\nu,t)\in \Sbb^{d-1} \times [0, T] \subseteq \mathbb{R}^m=\mathbb{R}^{d}\times\mathbb{R}$;
        \item $S := \left\{(\nu,t,x)\in\mathbb{R}^{d}\times\mathbb{R}\times\mathbb{R}^d: x\in K \text{ and } x\cdot \nu\le t-R\right\}$. Since $S=\mathbb{R}^d\times\mathbb{R}\times K \cap \{(\nu,t,x)\in \mathbb{R}^d\times\mathbb{R}\times\mathbb{R}^d:\,x\cdot\nu-t+R\le 0\}$, the set $S$ is definable under Assumption \ref{Assumption: basic requirements for o-minimal structures of interest}.
    \end{enumerate} 
    Then, for each fixed $a=(\nu,t)\in \mathbb{R}^{d}\times\mathbb{R}=\mathbb{R}^m$, we have
    \begin{align*}
        S_a=S_{(\nu,t)}=\{x\in K:\, x\cdot\nu\le t-R\}=K_t^\nu.
    \end{align*}
    Lemma \ref{lemma: (2.10) Proposition of Dries} implies that $\chi(S_a)=\chi(K_t^\nu)$ takes only finitely many values as $a=(\nu,t)$ runs through $\mathbb{R}^m=\mathbb{R}^{d}\times\mathbb{R}$. Therefore, $\chi(K_t^\nu)$ takes only finitely many values as $(\nu,t)$ runs through $\mathbb{S}^{d-1}\times[0,T]$.

    Furthermore, Lemma \ref{lemma: (2.10) Proposition of Dries} indicates that $\{(\nu,t)\in\mathbb{R}^m:\, \chi(K_t^\nu)=z\}$ is definable for every integer $z$. Because $\mathbb{S}^{d-1}\times[0,T]$ is definable (under Assumption \ref{Assumption: basic requirements for o-minimal structures of interest}), we have that $\{(\nu,t)\in \mathbb{S}^{d-1}\times[0,T]:\, \chi(K_t^\nu)=z\}=\mathbb{S}^{d-1}\times[0,T]\cap\{(\nu,t)\in\mathbb{R}^m:\, \chi(K_t^\nu)=z\}$ is definable for every integer $z$. The proof of the first result of Theorem \ref{thm: tameness theorem} is completed.

    The second result is a straightforward corollary of the first. The third result of Theorem \ref{thm: tameness theorem} is implied by the first result and the ``monotonicity theorem" in Chapter 3 of \cite{van1998tame}.
\end{proof}

\subsection{Proof of Theorem \ref{thm: SERT preserves all information of ERT}}\label{eq: proof, ERT and SERT determine each other}

\begin{proof}
    It is straightforward that $\operatorname{ERT}(g)$ determines $\operatorname{SERT}(g)$. It suffices to show that $\operatorname{SERT}(g)$ determines $\operatorname{ERT}(g)$. 
    
    For every fixed direction $\nu\in\mathbb{S}^{d-1}$, the definition of $\operatorname{SERT}(g)$ implies
\begin{align*}
    \frac{d}{dt} \operatorname{SERT}(g)(\nu,t)= \operatorname{ERT}(g)(\nu,t) + \left( - \frac{1}{T} \int_0^T \operatorname{ERT}(g)(\nu,\tau) \,d\tau \right),
\end{align*}
for all $t\in[0,T]$ that are not discontinuities of $t\mapsto \operatorname{ERT}(g)(\nu,t)$. Recall that the support of $g$ is strictly smaller than the domain $B_{\mathbb{R}^d}(0,R)$. For any $t^*< \operatorname{dist}\left(\supp(g), \partial B_{\mathbb{R}^d}(0,R)\right)$, we have $g(x)\cdot R(x,\nu,t^*)=0$ for all $x\in B_{\mathbb{R}^d}(0,R)$, which indicates that $\operatorname{ERT}(g)(\nu,t^*)=\int_{B_{\mathbb{R}^d}(0,R)} g(x) \cdot R(x,\nu,t^*) \, [d\chi(x)]=0$. Hence, 
\begin{align*}
    \lim_{t\rightarrow0+}\frac{d}{dt} \operatorname{SERT}(g)(\nu,t)= - \frac{1}{T} \int_0^T \operatorname{ERT}(g)(\nu,\tau) \,d\tau,
\end{align*}
which implies
\begin{align}\label{eq: SERT determines ERT}
    \operatorname{ERT}(g)(\nu,t) = \frac{d}{dt} \operatorname{SERT}(g)(\nu,t) - \lim_{t\rightarrow0+}\frac{d}{dt} \operatorname{SERT}(g)(\nu,t),
\end{align}
for all $t\in[0,T]$ that are not discontinuities of $t\rightarrow \operatorname{ERT}(g)(\nu,t)$. The right continuity of $t \mapsto \operatorname{ERT}(g)(\nu,t)$ implies that Eq.~\eqref{eq: SERT determines ERT} holds for all $t\in[0,T]$. That is, $\operatorname{SERT}(g)$ determines $\operatorname{ERT}(g)$ through Eq.~\eqref{eq: SERT determines ERT}. The proof is completed.
\end{proof}

\subsection{Proof of Theorem \ref{thm: invertibility on piecewise images}}\label{proof of thm: invertibility on piecewise images}

Before the proof of Theorem \ref{thm: invertibility on piecewise images}, we suggest the audiences read Appendix \ref{section: Discussions on the Invertibility of ERT and SERT}, especially Proposition \ref{prop::fubini-assumption-invert}, as a prerequisite for the proof. A takeaway message from Appendix \ref{section: Discussions on the Invertibility of ERT and SERT} is that the ERT is invertible if the ``Fubini condition" (Eq.~\eqref{eq: the Fubini condition}) is satisfied. In addition to Appendix \ref{section: Discussions on the Invertibility of ERT and SERT}, we need the following lemma as a prerequisite
\begin{lemma}\label{eq: linearity in some sense}
    If $f,g\in\operatorname{CF}(X)$, we have
    \begin{align*}
        \int_X \left\{a\cdot f(x) + b\cdot g(x)\right\}\,[d\chi(x)] = a\cdot \int_X f(x)\,d\chi(x) + b\cdot \int_X g(x)\,d\chi(x)
    \end{align*}
    for all $a,b\in\mathbb{R}$.
\end{lemma}
\begin{proof}
Both $a \cdot f(x)$ and $b \cdot g(x)$ are compactly supported real-valued definable functions; the images of $a \cdot f(x)$ and $b \cdot g(x)$ are both finite-point sets in $\R$. The equality above then follows from Lemma~\ref{lem::new_prop_72} and the homogeneity of $[d\chi]$.
\end{proof}

The following lemma implies the Fubini condition is satisfied by piecewise constant definable functions with compact support.
\begin{lemma}\label{lem::pwc-fubini}
    Suppose $X$ is a definable set and $Y$ is a bounded subset of $\mathbb{R}^N$ for some positive integer $N$. If function $f \in \operatorname{Def}(X;\mathbb{R})$ takes finitely many values, we have the following
    \[\int_Y \left( \int_{F^{-1}(y)} f(x) \,[d\chi(x)] \right) [d\chi(y)] = \int_X f(x) [d\chi(x)]  \]
    for all $F \in \operatorname{Def}(X; Y)$.
\end{lemma}
\begin{proof}
Suppose $f$ takes values in $\{a_i\}_{i=1}^m$ with $m<\infty$. Denote $D_i:=\{x\in X: f(x)=a_i\}$ for $i=1,\ldots,m$ and $\mathcal{D}:=\{D_i\}_{i=1}^m$. By the ``cell decomposition theorem" (see Chapter 3 of \cite{van1998tame}), there exists a finite decomposition $\mathcal{E}$ of $X$ such that, for each $E\in\mathcal{E}$, the restriction $F\vert_{E}$ is continuous. Denote $\mathcal{G}:=\{D\cap E:\, D\in\mathcal{D}\text{ and }E\in\mathcal{E}\}=:\{G_i\}_{i=1}^n$; obviously, $\mathcal{G}$ is a finite definable partition of $X$. Furthermore, for every component $G_i\in\mathcal{G}$, $F$ is continuous on $G_i\in\mathcal{G}$ and $f$ is a constant (say $b_i$) on $G_i$.
 
For each $i=1,\ldots,n$, define the restriction $F_i:=F\vert_{G_i}: G_i \to Y$. Since $F_i$ is continuous on $G_i$, the ``trivialization theorem" (see Chapter 9 of \cite{van1998tame}) implies that there exists a finite definable partition $\{Y_{ij}\}_{j}$ of $Y$ such that $F_i$ is definably trivial over $Y_{ij}$ for each $j$. That is, for each $j$, there exists a definable set $U_{ij}\subseteq\mathbb{R}^N$ for some $N$ and a definable map $\lambda_{ij}: F_i^{-1}(Y_{ij})\rightarrow U_{ij}$ such that $h_{ij}:=(F_i\vert_{F_i^{-1}(Y_{ij})},\,\lambda_{ij}): F_i^{-1}(Y_{ij}) \rightarrow Y_{ij}\times U_{ij}$ is a homeomorphism; furthermore, $F_i\vert_{F_i^{-1}(Y_{ij})}=\pi\circ h_{ij}$, where $\pi: Y_{ij}\times U_{ij}\rightarrow Y_{ij}$ is a projection map, and $F_i^{-1}(y)$ is definably homeomorphic to $U_{ij}$ for each $y\in Y_{ij}$.

The additivity of Euler characteristics implies
\begin{align*}
    \int_Y \left( \int_{F^{-1}(y)} f(x) \,[d\chi(x)] \right) [d\chi(y)] &=\int_Y \sum_{i=1}^n\left( \int_{F_i^{-1}(y)} f(x) \,[d\chi(x)] \right) [d\chi(y)]\\ 
    &= \int_Y \sum_{i=1}^n b_i\cdot\chi\left(F_i^{-1}(y) \right) [d\chi(y)]
\end{align*}
Lemma \ref{lemma: (2.10) Proposition of Dries} implies that $y\mapsto \chi(F_i^{-1}(y))=\chi(\{x\in X: F_i(x)=y\})$ is a constructible function defined on $Y$. Therefore, Lemma \ref{eq: linearity in some sense}, together with the additivity of Euler characteristics, implies the following
\begin{align*}
    \int_Y \left( \int_{F^{-1}(y)} f(x) \,[d\chi(x)] \right) [d\chi(y)] &= \sum_{i=1}^n b_i \int_Y \chi\left(F_i^{-1}(y) \right) d\chi(y) \\
    &= \sum_{i=1}^n b_i \sum_j \int_{Y_{ij}} \chi\left(F_i^{-1}(y) \right) d\chi(y)
\end{align*}
Since all fibers $F_i^{-1}(y)$ for $y\in Y_{ij}$ are all definably homeomorphic to $U_{ij}$, we have $\chi\left(F_i^{-1}(y) \right)=\chi(U_{ij})$ for all $y\in Y_{ij}$. Hence, we have
\allowdisplaybreaks
\begin{align*}
    \int_Y \left( \int_{F^{-1}(y)} f(x) \,[d\chi(x)] \right) [d\chi(y)] &= \sum_{i=1}^n b_i \sum_j \int_{Y_{ij}} \chi(U_{ij}) \,[d\chi(y)] \\
    &= \sum_{i=1}^n b_i \sum_j \chi(Y_{ij})\cdot\chi(U_{ij}) \\
    &= \sum_{i=1}^n b_i \sum_j \chi(Y_{ij}\times U_{ij}) \\
    &= \sum_{i=1}^n b_i \sum_j \chi\left(F_i^{-1}(Y_{ij})\right).
\end{align*}
Since, for each $i$, $\{Y_{ij}\}_j$ is a partition of $Y$, we have
\begin{align*}
    \int_Y \left( \int_{F^{-1}(y)} f(x) \,[d\chi(x)] \right) [d\chi(y)] &= \sum_{i=1}^n b_i \chi\left(F_i^{-1}(Y)\right) \\
    &= \sum_{i=1}^n b_i \chi\left(G_i\right) \\
    &= \sum_{i=1}^n b_i \int_X \mathbbm{1}_{G_i}(x) \,d\chi(x).
\end{align*}
Applying Lemma \ref{eq: linearity in some sense} again, we have
\begin{align*}
    \int_Y \left( \int_{F^{-1}(y)} f(x) \,[d\chi(x)] \right) [d\chi(y)] &= \int_X \sum_{i=1}^n b_i\cdot\mathbbm{1}_{G_i}(x) \,[d\chi(x)] \\
    & =\int_X f(x) \,[d\chi(x)] 
\end{align*}
This concludes the proof.
\end{proof}

\begin{proof}[Proof of Theorem~\ref{thm: invertibility on piecewise images}]
Our goal is to prove that Eq.~\eqref{eq: the Fubini condition} is true if $g\in\mathfrak{D}_{R,d}^{pc}$, which implies the desired invertibility via Proposition~\ref{prop::fubini-assumption-invert}. 

For the ease of notations, we will denote $S = B_{\R^d}(0, R)$ and $T = \mathbb{S}^{d-1}\times[0,T]$. Let $x'$ be any point in $S$ and fixed. The kernel functions $R(x, \nu, t)$ and $ R'(\nu, t, x')$ are defined in Eq.~\eqref{eq: Euler function representation of ECT} and Eq.~\eqref{eq: dual kernel R'}, respectively. Let $g \in \mathfrak{D}_{R, d}^{pc}$. Then, the function $(x,\nu,t)\mapsto g(x)\cdot R(x, \nu, t)\cdot R'(\nu, t, x')$ belongs to $\operatorname{Def}(S\times T;\mathbb{R})$ and takes finitely many values.
\allowdisplaybreaks
Consider the standard projection maps $p_1: S \times T \to S$ and $p_2: S \times T \to T$ as follows 
\begin{enumerate}
    \item Applying Lemma~\ref{lem::pwc-fubini} to $p_1$, we have the following
    \begin{align*}
        &\int_{S \times T} g(x)\cdot R(x, \nu, t)\cdot R'(\nu, t, x') \,[d\chi(x, \nu, t)] \\
        &= \int_{S} \left(\int_{p_1^{-1}(x)} g(x)\cdot R(x, \nu, t)\cdot R'(\nu, t, x') \,[d\chi(x, \nu, t)] \right) [d\chi(x)]\\
        &= \int_{S} \left(\int_{\{x\}\times T} g(x)\cdot R(x, \nu, t)\cdot R'(\nu, t, x') \,[d\chi(x, \nu, t)] \right) [d\chi(x)] \\
    &= \int_{S} g(x) \left(\int_{T} R(x, \nu, t)\cdot R'(\nu, t, x') \,[d\chi(\nu, t)] \right) [d\chi(x)].
    \end{align*}
    \item Applying Lemma~\ref{lem::pwc-fubini} to $p_2$, we have the following
    \begin{align*}
        &\int_{S \times T} g(x)\cdot R(x, \nu, t)\cdot R'(\nu, t, x') \,[d\chi(x, \nu, t)] \\
        &= \int_{T} \left(\int_{p_2^{-1}(\nu, t)} g(x)\cdot R(x, \nu, t)\cdot R'(\nu, t, x') \,[d\chi(x, \nu,t)] \right) [d\chi(\nu, t)]\\
        &= \int_{T} \left(\int_{S\times\{(\nu, t)\}} g(x)\cdot R(x, \nu, t)\cdot R'(\nu, t, x') \,[d\chi(x, \nu,t)] \right) [d\chi(\nu, t)]\\
 &= \int_{T} \left(\int_{S} g(x)\cdot R(x, \nu, t)\cdot R'(\nu, t, x') \,[d\chi(x)] \right) [d\chi(\nu, t)] \\
 &= \int_{T} \left(\int_{S} g(x)\cdot R(x, \nu, t) \,[d\chi(x)] \right) R'(\nu, t, x') \,[d\chi(\nu, t)].
    \end{align*}
\end{enumerate}
Combining the two equations above gives the desired Eq.~\eqref{eq: the Fubini condition}. 

Lastly, the invertibility of the SERT follows from Theorem \ref{thm: SERT preserves all information of ERT}. The proof is completed.
\end{proof}

\subsection{Proof of Theorem \ref{thm: relationship between our proposed ERT and the referred existing transforms}}\label{proof: relationship between our proposed ERT and the referred existing transforms}

\begin{proof}
Denote $S_{\nu, t} := B_{\mathbb{R}^d}(0, R) \cap \{x \in \R^d: x\cdot\nu\le t-R\}$. A direct computation shows that
\begin{align*}
    \int_{\mathbb{R}} s \cdot \operatorname{LECT}(g)(\nu, t, s) \, [d\chi(s)] &= \int_{\R} s \cdot \chi\left(\left\{x\in B_{\mathbb{R}^d}(0,R):\, x\cdot\nu\le t-R \text{ and } g(x)=s\right\}\right) \,[d\chi(s)]\\
    &= \int_{\R} s \cdot \chi\left(\left\{ x \in S_{\nu, t} : g(x) = s\right\}\right) \,[d\chi(s)].
\end{align*}
``Corollary 8" of \cite{baryshnikov2010euler} implies the following
\begin{align*}
    \int_{\R} s \cdot \chi\left(\left\{ x \in S_{\nu, t} : g(x) = s\right\}\right) \,[d\chi(s)] & = \int_{S_{\nu, t}} g(x) \,[d\chi(x)] \\
    &= \int_{B_{\mathbb{R}^d}(0,R)} g(x) \cdot R(x,\nu,t) \, [d\chi(x)] \\
    & = \operatorname{ERT}(g)(\nu,t).
\end{align*}
Then, Eq.~\eqref{eq: Euler representation of ERT via LECT} follows.

We can write ``Proposition 2" of \cite{baryshnikov2010euler} using the LECT and SELECT as follows
\begin{align*}
    \lfloor\operatorname{ERT}\rfloor(g)(\nu,t) &= \int_{B_{\mathbb{R}^d}(0,R)} g(x)\cdot R(x, \nu, t) \,\lfloor d\chi(x) \rfloor \\
    & = \int_{S_{\nu,t}} g(x) \,\lfloor d\chi(x) \rfloor \\
    & = \int_0^\infty \operatorname{SELECT}(g)(\nu, t, s) - \operatorname{SELECT}(-g)(\nu, t, s) + \operatorname{LECT}(-g)(\nu, t, s) \,ds.
\end{align*}
Similarly, we have
\begin{align*}
    \lceil\operatorname{ERT}\rceil(g)(\nu,t) = \int_0^\infty \operatorname{SELECT}(g)(\nu, t, s) - \operatorname{LECT}(g)(\nu, t, s) - \operatorname{SELECT}(-g)(\nu, t, s) \,ds.
\end{align*}
Thus, the proof of Eq.~\eqref{eq: Lebesgue representations of the floor and ceiling ERT} is completed. Taking the average of the two expressions in Eq.~\eqref{eq: Lebesgue representations of the floor and ceiling ERT} gives Eq.~\eqref{eq: Lebesgue representation of ERT}.
\end{proof}

    

\section{Definability vs. Tameness of Functions}\label{section: Definability vs. Tameness}

In the literature on TDA, one may often come across the concept of tameness. The word ``tame" is also often used interchangeably with ``definable." The concept of definability is presented in Definition \ref{def: definability}, and the concept of tameness can be found in \cite{bobrowski_borman_2012}. In this section, we analyze the relationship between them. To avoid confusion, we will not interchange the words ``definable" and ``tame" in this paper.

\subsection{Tameness}

We first go through the concept of tameness as follows, which is a generalized version of ``Definition 2.2" in \cite{bobrowski_borman_2012}. 
\begin{definition}\label{def: def of tameness}
    Let $X$ be a topological space with finite $\chi(X)$ and $f: X \to \R$ a continuous bounded function. For each $\alpha \in \R$, we define the \textit{super-level set} at $\alpha$ as $X_{\alpha^+}^f \coloneqq \{x \in X: \ f(x) \geq \alpha\}$ and the \textit{sub-level set} $\alpha$ $X_{\alpha^-}^f \coloneqq \{x \in X: \ f(x) \leq \alpha\}$. The function $f$ is said to be \textbf{tame} if it satisfies the following two conditions
    \begin{itemize}
        \item The homotopy types of $X_{\alpha^+}^f$ and $X_{\alpha^-}^f$ change finitely many times as $\alpha$ varies through $\R$;
        \item the homology groups of $X_{\alpha^+}^f$ and $X_{\alpha^-}^f$ are all finitely generated for all $\alpha \in \R$.
    \end{itemize}
\end{definition}
Similar to Definition \ref{def: def of tameness}, we define the following
\begin{definition}\label{def: def of EC-tameness}
    Let $X$ be a topological space with finite $\chi(X)$. A (not necessarily continuous) bounded function $f: X \to \R$ is said to be \textbf{EC-tame} if the Euler characteristics $\chi(X_{\alpha^+}^f)$ and $\chi(X_{\alpha^-}^f)$ are finite for all $\alpha\in\mathbb{R}$ and change only finitely many times as $\alpha$ varies through $\R$.
\end{definition}

Note that the definition of ``tame functions" in \cite{bobrowski_borman_2012} are equivalent to continuous EC-tame functions on compact topological space $X$ in our context.

\subsection{Relationship between Tameness and Definability}

In general, it is not the case that a tame function is definable in the o-minimal sense, which is illustrated by the following example. 
\begin{example}
    Let $W$ denote the \textit{Warsaw circle} (see Figure~\ref{fig:warsaw_cirlce}) defined as the union of the closed topologist's sine curve and an arc $J$ ``joining" the two ends of the topologist's sine curve:
    \[W \coloneqq \{(x, \sin(\frac{2\pi}{x})\ |\ x \in (0, 1]\} \cup \{(0, y)\ |\ -1 \leq y \leq 1\} \cup J\]
    \[J \coloneqq \{(0.5 + R \cos(t), -2 + R \sin(t)\ |\ \beta \leq t \leq 2\pi + \alpha\}\}\]
    \[R \coloneqq \sqrt{(\frac{1}{2})^2 + 2^2}, \alpha \coloneqq \arctan(\frac{2}{0.5}), \beta = \pi - \alpha\]
\begin{figure}[!tbp]
\centering
    \includegraphics[scale=0.3]{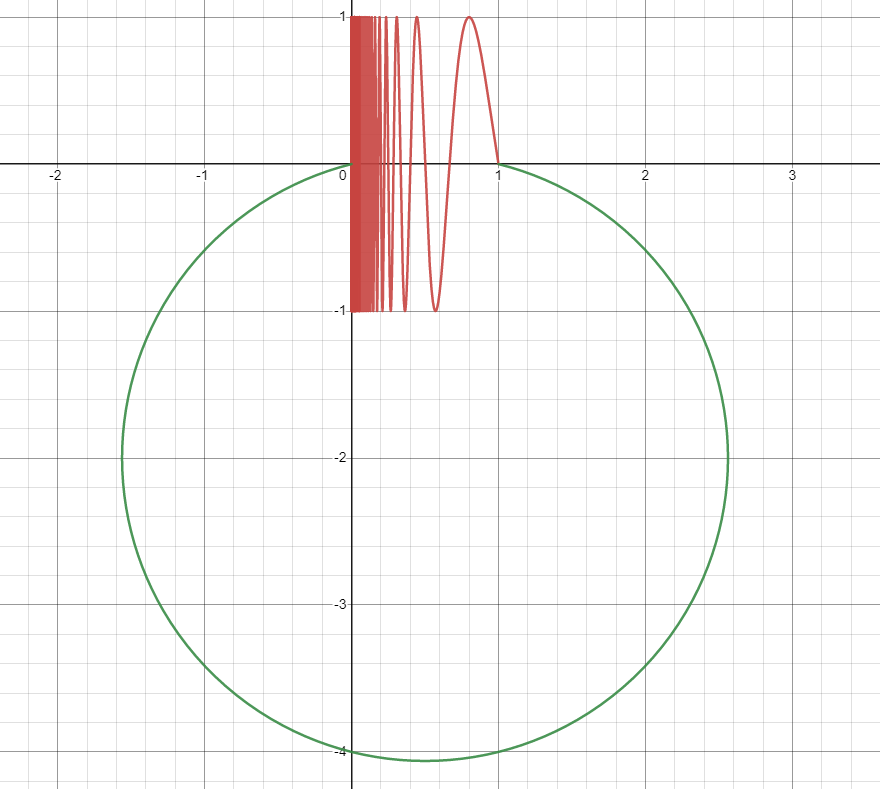}
    \caption{The Warsaw Circle $W$}
    \label{fig:warsaw_cirlce}
\end{figure}

    $W$ itself is not definable. However, $W$ is compact as it is bounded and is the union of two closed sets (the closed topologist's sine curve and $J$). Now consider the constant continuous function $f: W \to \R$ that sends every point to $0$ - the graph of this function is $W \times \{0\} \subseteq \R^3$ and is not definable. Hence, $f$ is not definable.
    
    On the other hand, the function is a tame function. This is because the Warsaw Circle is known to be simply connected and has all trivial homology groups beyond dimension $0$.
\end{example}

\begin{remark}
    If $f: X \to Y$ is a tame function between two definable sets, would $f$ be definable? The answer is no. Consider the indicator function $\mathbbm{1}_W: \R^2 \to \R$ on the Warsaw circle $W$. 
\end{remark}

Conversely, a function that is definable in the o-minimal sense does not have to be tame either. The most obvious obstruction comes from the distinction that definable functions need not be continuous nor bounded. However, when we remove the trivial distinctions between the two, we do have the following result:
\begin{proposition}\label{prop::cont_definable_imply_tame}
  Suppose $X \subseteq \R^n$ is definable. If the function $f: X \to \R$ is continuous, bounded, and definable, then $f$ is tame.  
\end{proposition}

\begin{proof}
Let $\Gamma(f) \subset X \times \R$ be the graph of $f$ and let $\pi: \mathbb{R}^n\times\mathbb{R} \to \R^n$ be the standard projection function, we observe that $X_{\alpha^+}^f$ is the set $\pi(\Gamma(f) \cap X \times [\alpha, +\infty))$ and is thus definable. 

By the triangulation theorem (see Chapter 8 of \cite{van1998tame}), it follows that $X_{\alpha^+}^f$ is definably homeomorphic to a subcollection of open simplices in some finite Euclidean simplicial complex, which implies that the homology groups of $X_{\alpha^+}^f$ are all finitely generated. The case for $X_{\alpha^-}^f$ is similar.

Since $f: X \to \R$ is a continuous definable function between definable sets, the ``trivialization theorem" (see Chapter 9 of \cite{van1998tame}) asserts that there exists a definable partition of $\R$ into finitely many definable sets $R_1, ..., R_n$ such that, for any $i \in \{1, ..., n\}$, there exists a definable set $Y_i\subseteq\mathbb{R}^N$, for some dimension $N$, making the following diagram commute
\[\begin{tikzcd}
	{f^{-1}(R_i)} && {R_i \times Y_i} \\
	& {R_i}
	\arrow["{f_i}"', from=1-1, to=2-2]
	\arrow["{h_i=(f_i, \, \lambda_i)}", from=1-1, to=1-3]
	\arrow["{\pi_i}", from=1-3, to=2-2]
\end{tikzcd}\]
where $f_i$ is the function $f$ restricted to $f^{-1}(R_i)$ with $f_i:=f\vert_{f^{-1}(R_i)}$, $h_i$ is a homeomorphism, $\lambda_i: f^{-1}(R_i)\rightarrow Y_i$ is a continuous map, and $\pi_i: R_i \times Y_i\rightarrow R_i$ is the standard projection map.

It is straightforward to have the following disjoint union
\begin{align}\label{eq: a disjoint union of some super-level sets}
    \begin{aligned}
        X_{\alpha^+}^f &= \{x \in X: \ f(x) \geq \alpha\} \\
    & = \bigcup_{i=1}^n \left\{x\in f^{-1}(R_i): \, f_i(x)\ge \alpha\right\} \\
    & = \bigcup_{i=1}^n \left\{x\in f^{-1}(R_i): \, \pi_i\circ h_i(x)\ge \alpha\right\} \\
    & = \bigcup_{i=1}^n \left\{\xi\in h_i\left( f^{-1}(R_i)\right): \, \pi_i(\xi)\ge \alpha\right\} \\
    & = \bigcup_{i=1}^n \left\{\xi\in R_i\times Y_i: \, \pi_i(\xi)\ge \alpha\right\}.
    \end{aligned}
\end{align}
To show that the homotopy type of $X_{\alpha^+}^f$ changes finitely many times, Eq.~\eqref{eq: a disjoint union of some super-level sets} indicates that it suffices to verify that, for each projection map $\pi_i$, the homotopy type of the super-level sets of $\pi_i$ changes finitely many times. Indeed, since $R_i \in \mathcal{O}_1$, the set $R_i$ is a finite union of points and open intervals. The homotopy type of $\left\{\xi\in R_i\times Y_i: \, \pi_i(\xi)\ge \alpha\right\}$ changes only when $\alpha$ crosses the isolated points and boundary points of $R_i$. The verification for the sub-level sets is similar. Hence, $f$ is a tame function.
\end{proof}


\subsection{A Useful Formula}

We use Proposition \ref{prop::cont_definable_imply_tame} to prove a variant of ``Proposition 7.2" in \cite{bobrowski_borman_2012}, which will be implemented in Appendix \ref{section: Discussions on the Invertibility of ERT and SERT}.
\begin{lemma}\label{lem::new_prop_72}
    Suppose the topological space $X$ is definable, and functions $h, f: X \to \R$ are definable. If the image of $h$ is discrete and $f$ is bounded, then we have the following formula
    \[\int_X (h + f) \lceil d\chi \rceil = \int_X h \lceil d\chi \rceil + \int_X f \lceil d\chi \rceil\]
    The formula holds similarly for $\lfloor d\chi \rfloor$.
\end{lemma}
\begin{proof}
    Since $h(X)$ belongs to $\mathcal{O}_1$ and is discrete, $h(X)$ must be a finite point set, say $\{a_1, ..., a_n\}$. We can then partition $X$ into $A_1, ..., A_n$ such that
    \[h(x) = \sum_{i = 1}^n a_i \mathbbm{1}_{A_i}(x).\]
In addition, the ``cell decomposition theorem" (see Chapter 3 of \cite{van1998tame}) indicates that there exists a cell decomposition $\mathcal{D}$ of $X$ such that $f$ is continuous on each cell in $\mathcal{D}$. Hence, without loss of generality, we may assume that $f$ is continuous on each $A_i$.
    
    By additivity of Euler characteristics, we can decompose $\int_X (h + f) \lceil d\chi \rceil$ as follows
    \begin{align*}
        \int_X (h + f) \lceil d\chi \rceil & = \int_{X} \sum_{i = 1}^n (a_i + f)\cdot\mathbbm{1}_{A_i} \lceil d\chi \rceil\\
        &= \lim_{k\rightarrow\infty}\frac{1}{k} \int_X \left\lceil \sum_{i=1}^n k\cdot(a_i+f)\cdot\mathbbm{1}_{A_i}\right\rceil  \,d\chi\\
        &= \lim_{k\rightarrow\infty}\frac{1}{k} \int_X  \sum_{i=1}^n \left\lceil k\cdot(a_i+f)\right\rceil\cdot\mathbbm{1}_{A_i}  \,d\chi \\
        &= \lim_{k\rightarrow\infty}\frac{1}{k} \sum_{i=1}^n \int_{X}  \left\lceil k\cdot(a_i+f)\right\rceil\cdot\mathbbm{1}_{A_i}  \,d\chi \\
        &= \sum_{i=1}^n \lim_{k\rightarrow\infty}\frac{1}{k} \int_{A_i}  \left\lceil k\cdot(a_i+f)\right\rceil  \,d\chi  \\
        & = \sum_{i = 1}^n \int_{A_i} (a_i + f) \lceil d\chi \rceil.
    \end{align*}
    It then suffices to verify $\int_{A_i} a_i + f \lceil d\chi \rceil = \int_{A_i} f \lceil d\chi \rceil + \int_{A_i} a_i \lceil d\chi \rceil$. Since $a_i + f$ is continuous on $A_i$ for each $i$, Proposition \ref{prop::cont_definable_imply_tame} implies that $(a_i + f)\vert_{A_i}$ is tame. Then, it follows from ``Proposition 2.4" of \cite{bobrowski_borman_2012} that
    \begin{align*}
        \int_{A_i} (a_i + f) \, \lceil d\chi \rceil &= \sum_{v \in \CV(a_i + f)} \Delta_\chi(a_i+f, v) v\\
        &= \sum_{v \in \CV(f)} \Delta_\chi(f, v) (v + a_i) \\
        &= \sum_{v \in \CV(f)} \Delta_\chi(f, v) v + a_i \sum_{v \in \CV(f)} \Delta_\chi(f, v)\\
        &= \int_{A_i} f \lceil d\chi \rceil + a_i \sum_{v \in \CV(f)} \Delta_\chi(f, v)
    \end{align*}
where $\CV(f)$ is the set of values $\alpha$ (referred to as critical values) at which the homotopy type of $\{x\in A_i:\, f(x)\le \alpha\}$ changes; and $\Delta_\chi(f, v)$ is the change in Euler characteristic:
\begin{align}\label{eq: Delta(f,v)}
\Delta_\chi(f, v)=\chi(\{x\in A_i:\, f(x)\le v+\varepsilon\})-\chi(\{x\in A_i:\, f(x)\le v-\varepsilon\})    
\end{align}
for sufficiently small $\varepsilon$.
    
Since $f$ is bounded, so there exists $a \le b$ such that $\{x \in A_i:\, f(x) \leq b\} = X$ and $\{x \in A_i:\, f(x) \leq a\} = \emptyset$. The sum $\sum_{v \in \CV(f)} \Delta_\chi(f, v)$ then collapse as a telescoping sum to $\chi(A_i) - \chi(\emptyset) = \chi(A_i)$ (see Eq.~\eqref{eq: Delta(f,v)}), hence
    \[\int_{A_i} a_i + f \lceil d\chi \rceil = \int_{A_i} f \lceil d\chi \rceil + a_i \chi(A_i) = \int_{A_i} f \lceil d\chi \rceil + \int_{A_i} a_i \lceil d\chi \rceil\]
    The proof is completed.
\end{proof}

\section{Discussions on the Invertibility of the ERT}\label{section: Discussions on the Invertibility of ERT and SERT}

In this section, we discuss the invertibility of the ERT, especially its dependence on the ``Fubini condition." This section provides a prerequisite for the proof of Theorem \ref{thm: invertibility on piecewise images}.

Let $\mu = 1 - (-1)^{d}$ and $\lambda = 1$. \cite{schapira1995tomography} and the proof of ``Theorem 5" in \cite{ghrist2018persistent} show the following
\begin{align}\label{eq: Schapira inversion}
    \int_{\mathbb{S}^{d-1}\times[0,T]} R(x, \nu, t) \cdot R'(\nu, t, x') \, d\chi(\nu, t) = (\mu - \lambda) \delta_{\Delta}(x,x') + \lambda,
\end{align}
where $R'(\nu, t, x')$ is the dual kernel defined in Eq.~\eqref{eq: dual kernel R'}, and $\delta_{\Delta}(x,x') = 1$ if $x = x'$ and is $0$ otherwise.

We recall the dual Euler-Radon transform (DERT) in Equation~\ref{eq: def of dual Euler-Radon transform} as follows
\begin{align*}
\begin{aligned}
\operatorname{DERT}:\ \ & \operatorname{Def}(\mathbb{S}^{d-1}\times [0,T]) \rightarrow \mathbb{R}^{B_{\R^d}(0, R)},\\
& h \mapsto \operatorname{DERT}(h)=\left\{\operatorname{DERT}(h)(x):=\int_{\mathbb{S}^{d-1} \times [0, T]} h(\nu, t) \cdot R'(\nu, t, x) \, [d\chi(\nu, t)]\right\}_{x \in B_{\R^d}(0, R)}.
\end{aligned}
\end{align*}
The following proposition shows the relationship between the ERT and DERT, which is the core of the proof of Theorem \ref{thm: invertibility on piecewise images}.
\begin{proposition}\label{prop::fubini-assumption-invert}
    Suppose $g \in \mathfrak{D}_{R,d}$. If the following condition (referred to as the ``Fubini condition" hereafter) holds
\begin{align}\label{eq: the Fubini condition}
    \begin{aligned}
        &\int_{\mathbb{S}^{d-1}\times[0,T]} \left( \int_{B_{\R^d}(0, R)} g(x)\cdot R(x, \nu, t) [d\chi(x)] \right) R'(\nu, t, x') [d\chi(v, t)] \\
    &=  \int_{B_{\R^d}(0, R)} g(x) \left( \int_{\mathbb{S}^{d-1}\times[0,T]} R(x, \nu, t) \cdot R'(\nu, t, x') \, d\chi(\nu, t) \right) [d\chi(x)],
    \end{aligned}
\end{align}
we have the following formula
\begin{align}\label{eq: inversion formula of ERT}
    (\operatorname{DERT} \circ \operatorname{ERT})(g)(x') = (\mu - \lambda)\cdot g(x') + \lambda \left(\int_{B_{\R^d}(0, R)} g [d\chi] \right),\ \ \ \text{for all }x'\in B_{\R^d}(0, R),
\end{align}
where $\mu = 1 - (-1)^{d}$ and $\lambda = 1$.
\end{proposition}

Before providing the proof of Proposition~\ref{prop::fubini-assumption-invert}, we explain how Eq.~\eqref{eq: inversion formula of ERT} implies the invertibility of the ERT. Since $g$ has compact support, we have $\lim_{\xi\rightarrow R\mathbb{S}^{d-1}} g(\xi)=0$, where $\lim_{\xi\rightarrow R\mathbb{S}^{d-1}}$ means that $\xi$ converges to a point on the sphere $R\mathbb{S}^{d-1}=\{x\in\R^d:\, \Vert x\Vert=R\}$. Therefore, we have
\begin{align*}
    \lim_{\xi\rightarrow R\mathbb{S}^{d-1}} \frac{1}{\mu-\lambda}\cdot (\operatorname{DERT} \circ \operatorname{ERT})(g)(\xi)= \lim_{\xi\rightarrow R\mathbb{S}^{d-1}} g(\xi) + \frac{\lambda}{\mu-\lambda} \left(\int_{B_{\R^d}(0, R)} g [d\chi] \right)= \frac{\lambda}{\mu-\lambda} \left(\int_{B_{\R^d}(0, R)} g [d\chi] \right).
\end{align*}
The limit above implies
\begin{align*}
    g(x') = \frac{1}{\mu-\lambda}\cdot (\operatorname{DERT} \circ \operatorname{ERT})(g)(x') - \lim_{\xi\rightarrow R\mathbb{S}^{d-1}} \frac{1}{\mu-\lambda}\cdot (\operatorname{DERT} \circ \operatorname{ERT})(g)(\xi),
\end{align*}
which is the inversion formula in Eq.~\eqref{eq: formula inversion formula of ERT} and shows the invertibility of the ERT.

We provide the proof of Proposition~\ref{prop::fubini-assumption-invert} as follows
\begin{proof}[Proof of Proposition~\ref{prop::fubini-assumption-invert}]
For ease of notation, let $X = B_{\R^d}(0, R)$ and $Y = \Sbb^{d-1} \times [0, T]$. Then, we have
\allowdisplaybreaks
\begin{align*}
    \left(\operatorname{DERT} \circ \operatorname{ERT}\right)(g)(x') &= \operatorname{DERT}\left(\int_X g(x)\cdot R(x, \cdot, \cdot)\, [d\chi(x)]\right)(x')\\
    &= \int_{Y} \left(\int_X g(x)\cdot R(x, \nu, t) [d\chi(x)] \right) R'(\nu, t, x') \,[d\chi(\nu, t)].
\end{align*}
The Fubini condition in Eq.~\eqref{eq: the Fubini condition} implies
\begin{align*}
    \left(\operatorname{DERT} \circ \operatorname{ERT}\right)(g)(x') &= \int_X g(x) \left(\int_Y  R(x, \nu, t) R'(\nu, t, x') [d\chi(\nu, t)] \right) \,[d\chi(x)].
\end{align*}
Eq.~\eqref{eq: Schapira inversion} indicates the following
\begin{align*}
    \left(\operatorname{DERT} \circ \operatorname{ERT}\right)(g)(x') &= \int_X g(x) \left\{(\mu - \lambda) \delta_{\Delta}(x, x') + \lambda \right\} \,[d\chi(x)] \\
    &= \int_X (\mu - \lambda)\cdot g(x)\cdot \delta_{\Delta}(x, x') + \lambda\cdot g(x) \, [d\chi(x)]
\end{align*}
For each fixed $x'$, the function $(\mu - \lambda) g(x) \delta_{\Delta}(x, x')$ of $x$ is clearly discrete. Then, Lemma~\ref{lem::new_prop_72} implies
\begin{align*}
    \left(\operatorname{DERT} \circ \operatorname{ERT}\right)(g)(x') = \int_X (\mu - \lambda) g(x) \delta_{\Delta}(x, x') [d\chi(x)] + \int_X \lambda g(x) [d\chi(x)].
\end{align*}
Evaluating the two integrals above and keeping in mind that $\int (\cdot) [d\chi(x)]$ is homogeneous, we have that
\[ (\operatorname{DERT} \circ \operatorname{ERT})(h)(x') = (\mu - \lambda) g(x') + \lambda \left(\int_{B_{\R^d}(0, R)} g [d\chi] \right), \]
that is, the proof of Eq.~\eqref{eq: inversion formula of ERT} is completed.
\end{proof}

The Fubini condition specified above does fail in general. Plenty of examples are given in ``Corollary 6" of \cite{baryshnikov2010euler}. In ``Theorem 7" of \cite{baryshnikov2010euler}, this condition does hold when the definable function preserves fibers, ie. if $F: X \to Y$ is definable and $h \in \Def(X, \R)$ is constant on the fibers of $F$, then
\[ \int_X h [d\chi(x)] =  \int_Y \left(\int_{F^{-1}(y)} h(x) [d\chi(x)] \right) [d\chi(y)] \]
Unfortunately, this does not help much in the discussion of invertibility. The typical Fubini's Theorem for $d\chi$ that swaps the order of integration
\[\int_X \int_Y f(x, y) d\chi(y) d\chi(x) = \int_Y \int_X f(x, y) d\chi(x) d\chi(y)\]
is a consequence of choosing $F$ to be the projection maps $p_X: X \times Y \to X$ and $p_Y: X \times Y \to Y$. However, if we additionally impose the constraint that $f$ is constant on the fibers of $p_X$ and $p_Y$, this is the same as requiring $f$ to be identically constant on $X \times Y$.

\section{Discussion on $\lim_{\sigma\rightarrow0} \operatorname{ERT}(\phi_\sigma * \mathbbm{1}_K)$}\label{section: discussion on the limit of ERT as sigma -> 0}

Let $\phi_\sigma$ be a kernel function with a bandwidth $\sigma$ (e.g., see Eq.~\eqref{eq: an example of kernel functions}), and $\phi_\sigma * \mathbbm{1}_K$ the convolution of two functions. Although $\lim_{\sigma\rightarrow0} \phi_\sigma * \mathbbm{1}_K(x) = \mathbbm{1}_K(x)$ almost everywhere (e.g., see ``Theorem 4.1" of \cite{stein2011fourier}), it is not generally true that $\lim_{\sigma\rightarrow0} \operatorname{ERT}(\phi_\sigma * \mathbbm{1}_K)=\operatorname{ERT}(\mathbbm{1}_K)=\operatorname{ECT}(K)$. This phenomenon is symbolic of the general principle that a set of Lebesgue measure zero may not be of Euler characteristic zero. 

Let $\phi_\sigma$ be a kernel function with a bandwidth $\sigma$. For example
\begin{align}\label{eq: an example of kernel functions}
    \begin{aligned}
        & \phi_\sigma(x) = \begin{cases}
    \frac{C_d}{\sigma^d}\cdot \exp\left(-\frac{\sigma^2}{\sigma^2 - \|x\|^2}\right), & \|x\| < \sigma\\
    0, & \|x\| \geq \sigma
    \end{cases},\\ 
    & \text{where } C_d =\left(\int_{\|x\| < 1} e^{-\frac{1}{1 - \|x\|^2}}dx \right)^{-1}.
    \end{aligned}
\end{align}
Let $g_\sigma:=\phi_\sigma * \mathbbm{1}_K$ denote the convolution of two functions,
\begin{align*}
    g_\sigma(x) := \phi_\sigma * \mathbbm{1}_K(x) =\int_{\mathbb{R}^d} \phi_\sigma(y)\cdot\mathbbm{1}_K(x-y) dy.
\end{align*}
Furthermore, we set $\sigma\in[0,\frac{1}{10}]$, $d=2$, and $R=2$ for $\mathfrak{D}_{R,d}$.

Although $\lim_{\sigma\rightarrow0} \phi_\sigma * \mathbbm{1}_K(x) = \mathbbm{1}_K(x)$ almost everywhere (e.g., see ``Theorem 4.1" of \cite{stein2011fourier}), the following limit is generally not true
\begin{align}\label{eq: failure of an approximation to the identity}
    \lim_{\sigma\rightarrow0} \operatorname{ERT}(\phi_\sigma * \mathbbm{1}_K)=\operatorname{ERT}(\mathbbm{1}_K)=\operatorname{ECT}(K)
\end{align}
The failure of Eq.~\eqref{eq: failure of an approximation to the identity} is emblematic of the general principle that a set of Lebesgue measure zero may not be of Euler characteristic zero. The failure of Eq.~\eqref{eq: failure of an approximation to the identity} is illustrated by the following example.
\begin{example}\label{exp::ect_not_ert}
Consider the shape $K$ defined by the following where
\[\begin{aligned}
    & K  =  \left\{x\in\mathbb{R}^2 \,: \, \inf_{y\in S}\Vert x-y\Vert\le \frac{1}{10}\right\},\\ 
    & \text{where }\ \ S  =  \left\{\left(\frac{9}{10} \cos t, \frac{9}{10} \sin t\right) \,: \, 0\le t\le 2\pi \right\}.
\end{aligned}\]  

\begin{figure}[!tbp]
\centering
    \includegraphics[scale=0.6]{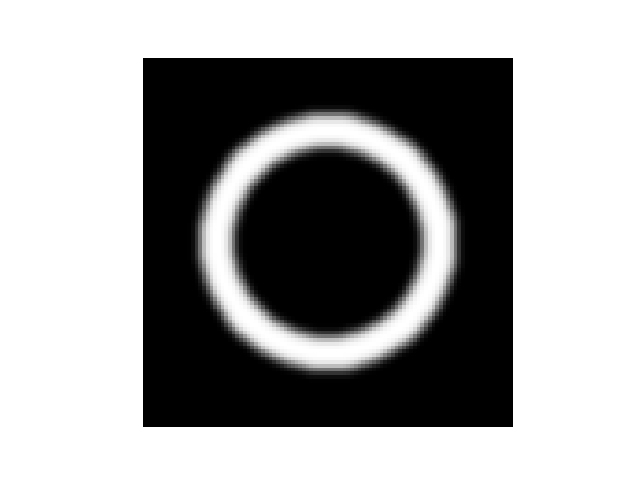}
    \caption{The grayscale image of $\phi_\sigma * \mathbbm{1}_K$ with small $\sigma$}
    \label{fig:gray_circle}
\end{figure}

Choose $\nu = (0, 1) \in \R^{2}$ and any $t \in (2-\frac{1}{100}, 2+\frac{1}{100})$, then $K \cap \{\nu \cdot x \leq t-R\}=K \cap \{\nu \cdot x \leq t-2\}$ has the homotopy type of $[0,1]$. Hence, $\operatorname{ECT}(K)((0,1),t) = 1$. 

On the other hand, since $0 \leq \phi_\sigma *\mathbbm{1}_K \leq 1$, it follows from Theorem~\ref{thm: relationship between our proposed ERT and the referred existing transforms} that,
\begin{align}
    \begin{aligned}
        &\operatorname{ERT}(\phi_\sigma * \mathbbm{1}_K)((0, 1), t) \\&= \int_0^1 \left\{\operatorname{SELECT}(\phi_\sigma * \mathbbm{1}_K)((0, 1), t, s) - \frac{1}{2} \operatorname{LECT}(\phi_\sigma * \mathbbm{1}_K)((0, 1), t, s)\right\} \,ds. 
    \end{aligned}
\end{align}
Omitting endpoint behaviors, for $0 < s < 1$ and sufficiently small $\sigma$ (see Figure~\ref{fig:gray_circle}), the level set for $\operatorname{LECT}(\phi_\sigma * \mathbbm{1}_K)$  has the homotopy type of the disjoint union of two circular arcs, hence its Euler characteristic is $2$. On the other hand, the super-level set for $\operatorname{SELECT}(\phi_\sigma * \mathbbm{1}_K)$ has the homotopy type of a solid circular arc, hence its Euler characteristic is $1$. Thus, we find that
\begin{align*}
    \operatorname{ERT}(\phi_\sigma * \mathbbm{1}_K)((0, 1), t) = \int_0^1 1 - \frac{1}{2} (2) \,ds = 0,\ \ \ \text{ for all }t \in (2-\frac{1}{100}, 2+\frac{1}{100}).
\end{align*}
That is, $0=\lim_{\sigma\rightarrow0} \operatorname{ERT}(\phi_\sigma * \mathbbm{1}_K)((0, 1), t)\ne\operatorname{ERT}(\mathbbm{1}_K)((0, 1), t)=\operatorname{ECT}(K)((0, 1), t)=1$ for all $t \in (2-\frac{1}{100}, 2+\frac{1}{100})$.
\end{example}

\section{$\operatorname{ERT}(g)(\nu, -)$ is not left continuous}\label{section: ECT is not left continuous}

In this section, we provide two simple examples that $\operatorname{ERT}(g)(\nu, -): [0, T] \to \R$ is not left continuous for some fixed direction $\nu \in \mathbb{S}^{d-1}$.

\begin{example}
    Let $R = 2$ and consider the indicator function on $K \coloneqq \{(x_1, x_2) \in \R^2\ |\ x_1^2 + x_2^2 = 1\}$. Since the input function is integer-valued, we have that $\operatorname{ERT}(\mathbbm{1}_{K})(\nu, -) = \operatorname{ECT}(K)(\nu, -)$ for any direction $\nu$. Computing $\operatorname{ECT}(K)(\nu, t)$ for all $t \in [0, 4]$, we have that $\operatorname{ECT}(K)(\nu, t) = 0$ for $t \in [0, 1) \cup [3, 4]$ and $\operatorname{ECT}(K)(\nu, t) = 1$ for $t \in [1, 3)$. In particular, this shows that $\operatorname{ERT}(g)(\mathbbm{1}_{K})(\nu, -)$ is right continuous but not left continuous.
\end{example}

\begin{example}
    Let $R = 2$ and $\nu = (1, 0)$. We consider the grayscale function defined as follows
    \[g: B_{\R^2}(0, 2) \to \R,\quad g(x_1, x_2) = \begin{cases}
        x_1+2,\quad (x_1,x_2) \in \mathbb{D}^2 \coloneqq \{(a, b) \in \R^2\ |\ a^2 + b^2 \leq 1\}\\
        0,\quad  \text{otherwise.}
    \end{cases}\]
Clearly, $g(x_1, x_2)=0$ whenever $x_1<-1$. For $t \in [0, 1)$, we have $g(x_1,x_2)\cdot\mathbbm{1}_{\{x_1\le t-2\}}=0$ for all $(x_1, x_2)\in B_{\R^2}(0, 2)$. Therefore, $\operatorname{ERT}(g)(\nu, t) = 0$ for all $t \in [0, 1)$. 
    
    Now for $t \in [1, 3)$, since $0 \leq g \leq 3$, by Eq.~\eqref{eq: Lebesgue representation of ERT}, we can write
    \[\operatorname{ERT}(g)(\nu, t) = \int_{0}^3 \operatorname{SELECT}(g)(\nu, t, s) - \frac{1}{2} \operatorname{LECT}(g)(\nu, t, s) \,ds.\]
    Recall from Equation~\ref{eq: def of LECT} that 
    \begin{align}\label{eq: LECT; example}
        \operatorname{LECT}(g)(\nu,t,s) = \chi\left(\left\{x\in B_{\mathbb{R}^2}(0,R):\, x_1\le t-R \text{ and } g(x)=s\right\}\right).
    \end{align}
    For any $t\in[1,3)$ and $s\in(0,3)$, the level set in Eq.~\eqref{eq: LECT; example} is either the empty set or is compact contractible. Specifically, we have the following 
    \begin{align*}
    & \operatorname{LECT}(g)(\nu,t,s) = 0, \ \ \text{ for all }t\in[1,3) \text{ and }s\in(0, 1), \\
        & \operatorname{LECT}(g)(\nu,t,s) = 1, \ \ \text{ for all }t\in[1,3) \text{ and }s\in[1, t], \\
        & \operatorname{LECT}(g)(\nu,t,s) = 0, \ \ \text{ for all }t\in[1,3) \text{ and }s\in(t, 3).
    \end{align*}
    Similarly, from Equation~\ref{eq: def of SELECT}, we have that
    \begin{align*}
        & \operatorname{SELECT}(g)(\nu,t,s) = 1, \ \ \text{for all }t\in[1,3) \text{ and }s\in(0,t), \\
        & \operatorname{SELECT}(g)(\nu,t,s) = 0, \ \ \text{for all }t\in[1,3) \text{ and }s\in(t,3).
    \end{align*}
    It follows that for $t \in [1, 3)$, 
    \[\operatorname{ERT}(g)(\nu, t) = \int_{0}^1 (1 - 0) ds + \int_{1}^{t} (1 - \frac{1}{2}(1)) ds = 1 + \frac{t-1}{2} = \frac{t + 1}{2}.\]
    Finally, for $t \in [3, 4]$, we have $g(x_1, x_2)\cdot\mathbbm{1}_{x_1\le t-2}=g(x_1, x_2)$. Hence, $\operatorname{ERT}(g)(\nu, t) = \operatorname{ERT}(g)(\nu, 3) = \frac{3 + 1}{2} = 2$. We conclude that $\operatorname{ERT}(g)(\nu, -)$ is right continuous but not left continuous.
\end{example}

\end{appendix}

\maketitle

\bibliography{sample}

\begin{thebibliography}{48}
\providecommand{\natexlab}[1]{#1}
\providecommand{\url}[1]{\texttt{#1}}
\expandafter\ifx\csname urlstyle\endcsname\relax
  \providecommand{\doi}[1]{doi: #1}\else
  \providecommand{\doi}{doi: \begingroup \urlstyle{rm}\Url}\fi

\bibitem[Adler et~al.(2007)Adler, Taylor, et~al.]{adler2007random}
R.~J. Adler, J.~E. Taylor, et~al.
\newblock \emph{Random fields and geometry}, volume~80.
\newblock Springer, 2007.

\bibitem[Aerts et~al.(2014)Aerts, Velazquez, Leijenaar, Parmar, Grossmann,
  Carvalho, Bussink, Monshouwer, Haibe-Kains, Rietveld,
  et~al.]{aerts2014decoding}
H.~J. Aerts, E.~R. Velazquez, R.~T. Leijenaar, C.~Parmar, P.~Grossmann,
  S.~Carvalho, J.~Bussink, R.~Monshouwer, B.~Haibe-Kains, D.~Rietveld, et~al.
\newblock Decoding tumour phenotype by noninvasive imaging using a quantitative
  radiomics approach.
\newblock \emph{Nature communications}, 5\penalty0 (1):\penalty0 4006, 2014.

\bibitem[Alexanderian(2015)]{alexanderian2015brief}
A.~Alexanderian.
\newblock A brief note on the {K}arhunen-{L}oève expansion.
\newblock \emph{arXiv preprint arXiv:1509.07526}, 2015.

\bibitem[Ashburner(2007)]{ashburner2007fast}
J.~Ashburner.
\newblock A fast diffeomorphic image registration algorithm.
\newblock \emph{Neuroimage}, 38\penalty0 (1):\penalty0 95--113, 2007.

\bibitem[Bankman(2008)]{bankman2008handbook}
I.~Bankman.
\newblock \emph{Handbook of medical image processing and analysis}.
\newblock Elsevier, 2008.

\bibitem[Baryshnikov and Ghrist(2010)]{baryshnikov2010euler}
Y.~Baryshnikov and R.~Ghrist.
\newblock Euler integration over definable functions.
\newblock \emph{Proceedings of the National Academy of Sciences}, 107\penalty0
  (21):\penalty0 9525--9530, 2010.

\bibitem[Baryshnikov et~al.(2011)Baryshnikov, Ghrist, and
  Lipsky]{baryshnikov2011inversion}
Y.~Baryshnikov, R.~Ghrist, and D.~Lipsky.
\newblock Inversion of {E}uler integral transforms with applications to sensor
  data.
\newblock \emph{Inverse problems}, 27\penalty0 (12):\penalty0 124001, 2011.

\bibitem[Bobrowski and Borman(2012)]{bobrowski_borman_2012}
O.~Bobrowski and M.~S. Borman.
\newblock Euler integration of gaussian random fields and persistent homology.
\newblock \emph{Journal of Topology and Analysis}, 04\penalty0 (01):\penalty0
  49–70, 2012.
\newblock \doi{10.1142/s1793525312500057}.

\bibitem[Bredon(2012)]{bredon2012sheaf}
G.~E. Bredon.
\newblock \emph{Sheaf theory}, volume 170.
\newblock Springer Science \& Business Media, 2012.

\bibitem[Brezis(2011)]{brezis2011functional}
H.~Brezis.
\newblock \emph{Functional analysis, Sobolev spaces and partial differential
  equations}, volume~2.
\newblock Springer, 2011.

\bibitem[Brooks and Grigsby(2013)]{brooks2013quantification}
F.~J. Brooks and P.~W. Grigsby.
\newblock Quantification of heterogeneity observed in medical images.
\newblock \emph{BMC medical imaging}, 13\penalty0 (1):\penalty0 1--12, 2013.

\bibitem[Carlsson(2009)]{carlsson2009topology}
G.~Carlsson.
\newblock Topology and data.
\newblock \emph{Bulletin of the American Mathematical Society}, 46\penalty0
  (2):\penalty0 255--308, 2009.

\bibitem[Chen and Rabad{\'a}n(2017)]{chen2017fast}
A.~X. Chen and R.~Rabad{\'a}n.
\newblock A fast semi-automatic segmentation tool for processing brain tumor
  images.
\newblock In \emph{Towards Integrative Machine Learning and Knowledge
  Extraction: BIRS Workshop, Banff, AB, Canada, July 24-26, 2015, Revised
  Selected Papers}, pages 170--181. Springer, 2017.

\bibitem[Crawford et~al.(2020)Crawford, Monod, Chen, Mukherjee, and
  Rabad{\'a}n]{crawford2020predicting}
L.~Crawford, A.~Monod, A.~X. Chen, S.~Mukherjee, and R.~Rabad{\'a}n.
\newblock Predicting clinical outcomes in glioblastoma: an application of
  topological and functional data analysis.
\newblock \emph{Journal of the American Statistical Association}, 115\penalty0
  (531):\penalty0 1139--1150, 2020.

\bibitem[Curry et~al.(2022)Curry, Mukherjee, and Turner]{curry2022many}
J.~Curry, S.~Mukherjee, and K.~Turner.
\newblock How many directions determine a shape and other sufficiency results
  for two topological transforms.
\newblock \emph{Transactions of the American Mathematical Society, Series B},
  9\penalty0 (32):\penalty0 1006--1043, 2022.

\bibitem[Dunson and Wu(2021)]{dunson2021inferring}
D.~B. Dunson and N.~Wu.
\newblock Inferring manifolds from noisy data using gaussian processes.
\newblock \emph{arXiv preprint arXiv:2110.07478}, 2021.

\bibitem[Eloyan et~al.(2020)Eloyan, Yue, and Khachatryan]{eloyan2020tumor}
A.~Eloyan, M.~S. Yue, and D.~Khachatryan.
\newblock Tumor heterogeneity estimation for radiomics in cancer.
\newblock \emph{Statistics in medicine}, 39\penalty0 (30):\penalty0 4704--4723,
  2020.

\bibitem[Ghrist et~al.(2018)Ghrist, Levanger, and Mai]{ghrist2018persistent}
R.~Ghrist, R.~Levanger, and H.~Mai.
\newblock Persistent homology and {E}uler integral transforms.
\newblock \emph{Journal of Applied and Computational Topology}, 2\penalty0
  (1):\penalty0 55--60, 2018.

\bibitem[Ghrist(2014)]{ghrist2014elementary}
R.~W. Ghrist.
\newblock \emph{Elementary applied topology}, volume~1.
\newblock Createspace Seattle, 2014.

\bibitem[Good(2013)]{good2013permutation}
P.~Good.
\newblock \emph{Permutation tests: a practical guide to resampling methods for
  testing hypotheses}.
\newblock Springer Science \& Business Media, 2013.

\bibitem[Howell(2006)]{howell2006handbook}
S.~B. Howell.
\newblock \emph{Handbook of CCD astronomy}, volume~5.
\newblock Cambridge University Press, 2006.

\bibitem[Hsing and Eubank(2015)]{hsing2015theoretical}
T.~Hsing and R.~Eubank.
\newblock \emph{Theoretical foundations of functional data analysis, with an
  introduction to linear operators}, volume 997.
\newblock John Wiley \& Sons, 2015.

\bibitem[Jiang et~al.(2020)Jiang, Kurtek, and Needham]{jiang2020weighted}
Q.~Jiang, S.~Kurtek, and T.~Needham.
\newblock The weighted {E}uler curve transform for shape and image analysis.
\newblock In \emph{Proceedings of the IEEE/CVF Conference on Computer Vision
  and Pattern Recognition Workshops}, pages 844--845, 2020.

\bibitem[Jolliffe(2002)]{jolliffe2002principal}
I.~T. Jolliffe.
\newblock \emph{Principal component analysis for special types of data}.
\newblock Springer, 2002.

\bibitem[Just(2014)]{just2014improving}
N.~Just.
\newblock Improving tumour heterogeneity mri assessment with histograms.
\newblock \emph{British journal of cancer}, 111\penalty0 (12):\penalty0
  2205--2213, 2014.

\bibitem[Kidder and Haar(1995)]{kidder1995satellite}
S.~Q. Kidder and T.~H.~V. Haar.
\newblock \emph{Satellite meteorology: an introduction}.
\newblock Gulf Professional Publishing, 1995.

\bibitem[Kirveslahti and Mukherjee(2023)]{kirveslahti2023representing}
H.~Kirveslahti and S.~Mukherjee.
\newblock Representing fields without correspondences: the lifted euler
  characteristic transform.
\newblock \emph{Journal of Applied and Computational Topology}, pages 1--34,
  2023.

\bibitem[Klenke(2020)]{klenke2013probability}
A.~Klenke.
\newblock \emph{Probability theory: a comprehensive course}.
\newblock Springer, 2020.

\bibitem[Li et~al.(2022)Li, Mukhopadhyay, and Dunson]{li2022efficient}
D.~Li, M.~Mukhopadhyay, and D.~B. Dunson.
\newblock Efficient manifold approximation with spherelets.
\newblock \emph{Journal of the Royal Statistical Society Series B: Statistical
  Methodology}, 84\penalty0 (4):\penalty0 1129--1149, 2022.

\bibitem[Maldonado et~al.(2021)Maldonado, Varghese, Rajagopalan, Duan, Balar,
  Lakhani, Antic, Massion, Johnson, Karwoski, et~al.]{maldonado2021validation}
F.~Maldonado, C.~Varghese, S.~Rajagopalan, F.~Duan, A.~B. Balar, D.~A. Lakhani,
  S.~L. Antic, P.~P. Massion, T.~F. Johnson, R.~A. Karwoski, et~al.
\newblock Validation of the broders classifier (benign versus aggressive nodule
  evaluation using radiomic stratification), a novel hrct-based radiomic
  classifier for indeterminate pulmonary nodules.
\newblock \emph{European Respiratory Journal}, 57\penalty0 (4), 2021.

\bibitem[Marsh and Beers(2023)]{marsh2023stability}
L.~Marsh and D.~Beers.
\newblock Stability and inference of the {E}uler characteristic transform.
\newblock \emph{arXiv preprint arXiv:2303.13200}, 2023.

\bibitem[Marsh et~al.(2022)Marsh, Zhou, Quin, Lu, Byrne, and
  Harrington]{marsh2022detecting}
L.~Marsh, F.~Y. Zhou, X.~Quin, X.~Lu, H.~M. Byrne, and H.~A. Harrington.
\newblock Detecting temporal shape changes with the {E}uler characteristic
  transform.
\newblock \emph{arXiv preprint arXiv:2212.10883}, 2022.

\bibitem[Meng and Eloyan(2021)]{meng2021principal}
K.~Meng and A.~Eloyan.
\newblock Principal manifold estimation via model complexity selection.
\newblock \emph{Journal of the Royal Statistical Society: Series B (Statistical
  Methodology)}, 83\penalty0 (2):\penalty0 369--394, 2021.

\bibitem[Meng et~al.(2022)Meng, Wang, Crawford, and Eloyan]{meng2022randomness}
K.~Meng, J.~Wang, L.~Crawford, and A.~Eloyan.
\newblock Randomness and statistical inference of shapes via the smooth {E}uler
  characteristic transform.
\newblock \emph{arXiv preprint arXiv:2204.12699}, 2022.

\bibitem[Milnor(1963)]{milnor1963morse}
J.~W. Milnor.
\newblock \emph{Morse theory}.
\newblock Number~51. Princeton university press, 1963.

\bibitem[Reed(2012)]{reed2012methods}
M.~Reed.
\newblock \emph{Methods of modern mathematical physics: Functional analysis}.
\newblock Elsevier, 2012.

\bibitem[Reed(2005)]{reed2005electron}
S.~J.~B. Reed.
\newblock \emph{Electron microprobe analysis and scanning electron microscopy
  in geology}.
\newblock Cambridge university press, 2005.

\bibitem[Robinson and Turner(2017)]{robinson2017hypothesis}
A.~Robinson and K.~Turner.
\newblock Hypothesis testing for topological data analysis.
\newblock \emph{Journal of Applied and Computational Topology}, 1:\penalty0
  241--261, 2017.

\bibitem[Schapira(1995)]{schapira1995tomography}
P.~Schapira.
\newblock Tomography of constructible functions.
\newblock In \emph{Applied Algebra, Algebraic Algorithms and Error-Correcting
  Codes: 11th International Symposium, AAECC-11 Paris, France, July 17--22,
  1995 Proceedings 11}, pages 427--435. Springer, 1995.

\bibitem[Stein and Shakarchi(2011)]{stein2011fourier}
E.~M. Stein and R.~Shakarchi.
\newblock \emph{Fourier analysis: an introduction}, volume~1.
\newblock Princeton University Press, 2011.

\bibitem[Taylor and Worsley(2008)]{taylor2008random}
J.~E. Taylor and K.~J. Worsley.
\newblock Random fields of multivariate test statistics, with applications to
  shape analysis.
\newblock 2008.

\bibitem[Turner et~al.(2014)Turner, Mukherjee, and Boyer]{turner2014persistent}
K.~Turner, S.~Mukherjee, and D.~M. Boyer.
\newblock Persistent homology transform for modeling shapes and surfaces.
\newblock \emph{Information and Inference: A Journal of the IMA}, 3\penalty0
  (4):\penalty0 310--344, 2014.

\bibitem[van~den Dries(1998)]{van1998tame}
L.~P.~D. van~den Dries.
\newblock \emph{Tame topology and o-minimal structures}, volume 248.
\newblock Cambridge university press, 1998.

\bibitem[Vishwanath et~al.(2020)Vishwanath, Fukumizu, Kuriki, and
  Sriperumbudur]{vishwanath2020limits}
S.~Vishwanath, K.~Fukumizu, S.~Kuriki, and B.~Sriperumbudur.
\newblock On the limits of topological data analysis for statistical inference.
\newblock \emph{arXiv e-prints}, pages arXiv--2001, 2020.

\bibitem[Wang et~al.(2021)Wang, Sudijono, Kirveslahti, Gao, Boyer, Mukherjee,
  and Crawford]{wang2021statistical}
B.~Wang, T.~Sudijono, H.~Kirveslahti, T.~Gao, D.~M. Boyer, S.~Mukherjee, and
  L.~Crawford.
\newblock A statistical pipeline for identifying physical features that
  differentiate classes of 3d shapes.
\newblock \emph{The Annals of Applied Statistics}, 15\penalty0 (2):\penalty0
  638--661, 2021.

\bibitem[Wang et~al.(2023)Wang, Meng, and Duan]{wang2023hypothesis}
J.~Wang, K.~Meng, and F.~Duan.
\newblock Hypothesis testing for medical imaging analysis via the smooth euler
  characteristic transform.
\newblock \emph{arXiv preprint arXiv:2308.06645}, 2023.

\bibitem[Wu et~al.(2022)Wu, Cheng, Cai, Ma, and Zhong]{wu2022subsampling}
S.~Wu, H.~Cheng, J.~Cai, P.~Ma, and W.~Zhong.
\newblock Subsampling in large graphs using ricci curvature.
\newblock In \emph{The Eleventh International Conference on Learning
  Representations}, 2022.

\bibitem[Yue et~al.(2016)Yue, Zipunnikov, Bazin, Pham, Reich, Crainiceanu, and
  Caffo]{yue2016parameterization}
C.~Yue, V.~Zipunnikov, P.-L. Bazin, D.~Pham, D.~Reich, C.~Crainiceanu, and
  B.~Caffo.
\newblock Parameterization of white matter manifold-like structures using
  principal surfaces.
\newblock \emph{Journal of the American Statistical Association}, 111\penalty0
  (515):\penalty0 1050--1060, 2016.

\end{thebibliography}

\end{document}